\documentclass[a4paper,11pt]{article}
\pdfoutput=1 % if you are submitting a pdflatex (i.e. if you have
\usepackage{jheppub} % for details on the use of the package, please
\usepackage[T1]{fontenc} % if needed

\usepackage{enumerate}
\usepackage{bm}
\usepackage{mathrsfs}
\usepackage{amsthm}
\usepackage{subfig}
\usepackage{physics}
\usepackage{verbatim}

\renewcommand{\tr}{\text{Tr}}

\renewcommand{\tilde}{\widetilde}

\newcommand{\area}{\operatorname{area}}
\newcommand{\RT}{\operatorname{RT}}
\newcommand{\KRT}{\operatorname{KRT}}
\newcommand{\arccosh}{\operatorname{arccosh}}
\newcommand{\arcsinh}{\operatorname{arcsinh}}
\newcommand{\KA}{\operatorname{KA}}
\newcommand{\KL}{\operatorname{KL}}
\newcommand{\A}{\operatorname{A}}
\renewcommand{\L}{\operatorname{L}}
\newcommand{\J}{\operatorname{J}}
\newcommand{\W}{\operatorname{W}}

\newcommand{\kett}[1]{\left.\left| #1 \right\rangle \right\rangle}
\newcommand{\bbra}[1]{\left\langle\left\langle #1 \right| \right.}
\newcommand{\bbrakett}[2]{\left\langle \left\langle #1 | #2 \right\rangle\right\rangle}

\newtheorem{theorem}{Theorem}
\newtheorem{lemma}[theorem]{Lemma}

\title{The Markov gap for geometric reflected entropy}

\author{Patrick Hayden,}
\author{Onkar Parrikar,}
\author{and Jonathan Sorce}

\affiliation{Stanford Institute for Theoretical Physics, Stanford University, 382 Via Pueblo Mall, Stanford, CA 94305-4060, U.S.A.}

\emailAdd{phayden@stanford.edu}
\emailAdd{parrikar@stanford.edu}
\emailAdd{jsorce@stanford.edu}

\abstract{The reflected entropy $S_R(A:B)$ of a density matrix $\rho_{AB}$ is a bipartite correlation measure lower-bounded by the quantum mutual information $I(A:B)$. In holographic states satisfying the quantum extremal surface formula, where the reflected entropy is related to the area of the entanglement wedge cross-section, there is often an order-$N^2$ gap between $S_R$ and $I$. We provide an information-theoretic interpretation of this gap by observing that $S_R - I$ is related to the fidelity of a particular Markov recovery problem that is impossible in any state whose entanglement wedge cross-section has a nonempty boundary; for this reason, we call the quantity $S_R - I$ the \emph{Markov gap}. We then prove that for time-symmetric states in pure AdS$_3$ gravity, the Markov gap is universally lower bounded by $\log(2) \ell_{\text{AdS}}/2 G_N$ times the number of endpoints of the cross-section. We provide evidence that this lower bound continues to hold in the presence of bulk matter, and comment on how it might generalize above three bulk dimensions. Finally, we explore the Markov recovery problem controlling $S_R - I$ using fixed area states. This analysis involves deriving a formula for the quantum fidelity --- in fact, for all the sandwiched R\'{e}nyi relative entropies --- between fixed area states with one versus two fixed areas, which may be of independent interest. We discuss, throughout the paper, connections to the general theory of multipartite entanglement in holography.}

\begin{document} 
\maketitle
\flushbottom

%%%%%%%%%%%%%%%%%%%%%%%
\section{Introduction}

Any operator $A$ acting on a Hilbert space $\mathcal{H}$ can be interpreted as a vector $\kett{A}$ in the Hilbert space $\mathcal{H} \otimes \mathcal{H}^*$, with $\mathcal{H}^*$ the dual of $\mathcal{H}.$ The inner product on this ``doubled'' Hilbert space is the Hilbert-Schmidt inner product
\begin{equation}
    \bbrakett{A}{B} = \tr(A^{\dagger} B).
\end{equation}
If $\rho$ is a density matrix on $\mathcal{H}$, then the state $\kett{\sqrt{\rho}}$ is a purification of $\rho$, i.e., it satisfies
\begin{equation}
    \rho = \tr_{\mathcal{H^*}} \kett{\sqrt{\rho}}\bbra{\sqrt{\rho}}.
\end{equation}
Because the state $\kett{\sqrt{\rho}}$ is defined without reference to any choice of basis on the Hilbert space $\mathcal{H}$, it is called the \emph{canonical purification} of $\rho$.

In \cite{dutta-faulkner}, Dutta and Faulkner provided an interesting interpretation of the canonical purification when $\rho$ is a semiclassical state in a holographic theory of quantum gravity. Being semiclassical, $\rho$ has a dual description in terms of an ``entanglement wedge'' $\W(\rho)$ --- a bulk domain of dependence bounded spatially by a quantum extremal surface.\footnote{We will assume the reader is familiar with the basics of the quantum extremal surface formula for boundary entropy and with the basics of entanglement wedge reconstruction. The relevant references are \cite{RT1, RT2, RT-homology, HRT, van2011patchwork, gravity-dual-density-matrix, maximin, LM, FLM, headrick2014causality, almheiri2015bulk, JLMS, DLR, QES, DHW, noisy-DHW, DL, hayden2019learning, akers-penington}.} The authors of \cite{dutta-faulkner} argued, using the equivalence of bulk and boundary path integrals in holography, that the canonical purification $\kett{\sqrt{\rho}}$ has a bulk description constructed by pasting $\W(\rho)$ to its CPT conjugate along the quantum extremal surface, then solving the bulk equations of motion with this ``doubled wedge'' as initial data. An example of this pasting for the simple case when $\rho$ describes an interval in the AdS$_3$ vacuum is sketched in figure \ref{fig:simple-pasting}.

\begin{figure}
    \centering
    \makebox[\textwidth][c]{
	\subfloat[\label{fig:simple-pasting-a}]{
	    \includegraphics{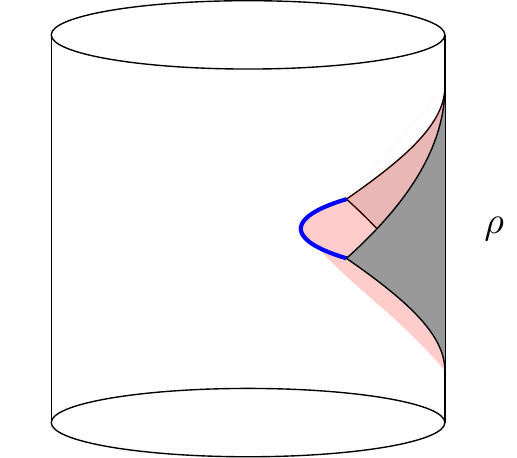}
	}
	\hspace{1em}
	\subfloat[\label{fig:simple-pasting-b}]{
    	\includegraphics{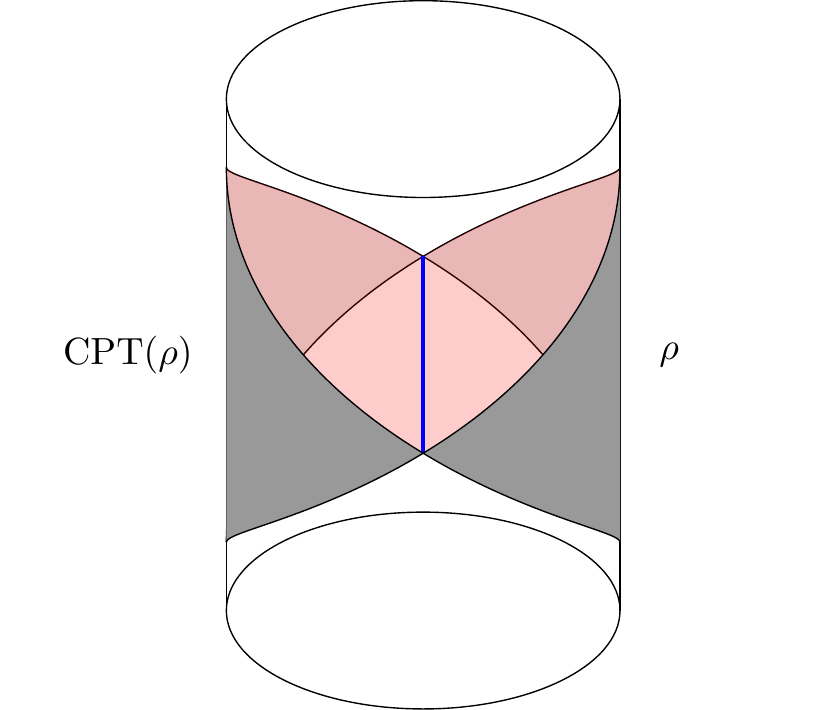}
	}
	}
    \caption{(a) A boundary state $\rho$ whose bulk dual is a Rindler wedge in AdS$_3$. (b) The initial data for the canonical purification $\kett{\sqrt{\rho}}$, formed by pasting $\W(\rho)$ to its CPT conjugate along the quantum extremal surface that forms the spatial boundary of $\W(\rho)$. The canonical purification is the AdS$_3$ vacuum state defined on a boundary circle whose circumference is twice the spatial extent of the domain of dependence on which $\rho$ was initially defined.}
    \label{fig:simple-pasting}
\end{figure}

In the case where $\rho_{AB}$ is a bipartite density matrix, Dutta and Faulkner defined the \emph{reflected entropy} as the entanglement entropy of the $AA^*$ system in the canonical purification:
\begin{equation}
    S_R(A:B) = S_{\kett{\sqrt{\rho_{AB}}}}(AA^*).
\end{equation}
If $\rho_{AB}$ is a semiclassical state satisfying the quantum extremal surface formula \cite{QES}, then its canonical purification ought to satisfy the quantum extremal surface formula as well. It follows, then, that the reflected entropy $S_R(A:B)$ is equal --- up to corrections exponentially small in $1/G_N$ --- to the generalized entropy\footnote{\label{footnote:gen-entropy}The generalized entropy of a codimension-$2$ bulk surface $\gamma$ with respect to a boundary domain of dependence $\Omega$ --- defined only when $\gamma$ is homologous to the spacelike slices of $\Omega$ --- is defined by
\begin{equation*}
    S_{\text{gen}}(\gamma|\Omega) = \frac{\area(\gamma)}{4 G_N} + S_{\text{bulk}}(\gamma|\Omega),
\end{equation*}
where $S_{\text{bulk}}(\gamma|\Omega)$ is the entropy of bulk quantum fields lying in the domain of dependence spacelike between $\gamma$ and $\Omega$.} of the minimal quantum extremal surface dividing the spacetime dual of $\kett{\sqrt{\rho_{AB}}}$ into one region through which the surface is homologous to $AA^*$ and another through which it is homologous to $BB^*$. This is sketched in figure \ref{fig:two-intervals-surfaces-b} for the special case when $\rho_{AB}$ is the density matrix of two equal-time intervals in the AdS$_3$ vacuum. 

Because the canonical purification is symmetric under interchange of $\W(\rho_{AB})$ and $\W(\rho_{A^*B^*})$, if the minimal quantum extremal surface computing $S_{\kett{\sqrt{\rho_{AB}}}}(AA^*)$ is unique then it must also obey this symmetry. The portion of this surface lying in the original spacetime $\W(\rho_{AB})$ --- equivalently, the image of this surface under the $A \leftrightarrow A^*, B \leftrightarrow B^*$ quotient --- is called the \emph{entanglement wedge cross-section} $\sigma_{A:B}$.\footnote{The entanglement wedge cross-section was originally defined in \cite{EOP1, EOP2} as a proposed dual to a boundary quantity called the \emph{entanglement of purification}. Similar claims have been made about the logarithmic negativity \cite{kudler2019entanglement, kusuki2019derivation} and the balanced partial entanglement \cite{wen2021balanced}. These conjectures are not incompatible with the proposed duality between entanglement wedge cross-sections and reflected entropy, but they are more speculative. As such, we will focus entirely on the reflected entropy within this paper; however, under the conjectures mentioned above, all of our results can equivalently be made into statements about the other proposed duals of the entanglement wedge cross-section.} The surface $\sigma_{A:B}$ divides the entanglement wedge $\W(\rho_{AB})$ into a portion through which $\sigma_{A:B}$ is homologous to $A$ and a portion through which $\sigma_{A:B}$ is homologous to $B$. See figure \ref{fig:two-intervals-surfaces-a} for a sketch.

\begin{figure}
    \centering
    \makebox[\textwidth][c]{
	\subfloat[\label{fig:two-intervals-surfaces-a}]{
	    \includegraphics[scale=0.9]{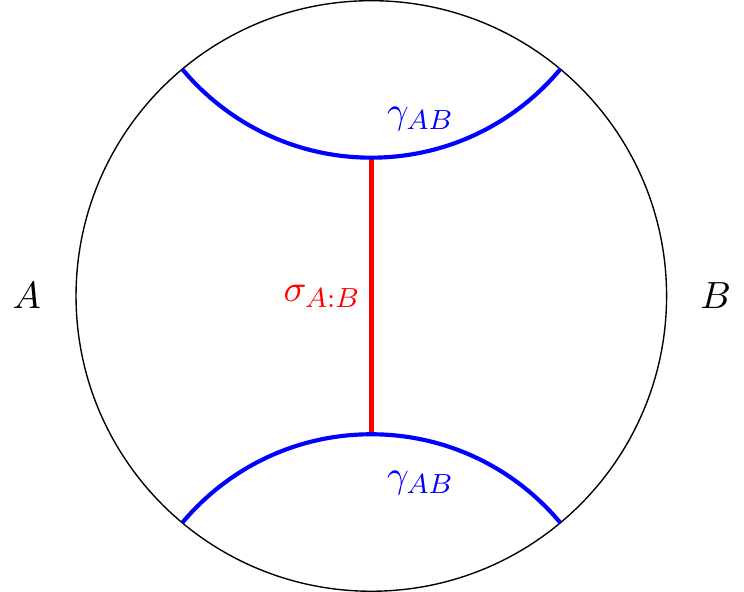}
	}
	\hspace{1em}
	\subfloat[\label{fig:two-intervals-surfaces-b}]{
    	\includegraphics[scale=0.9]{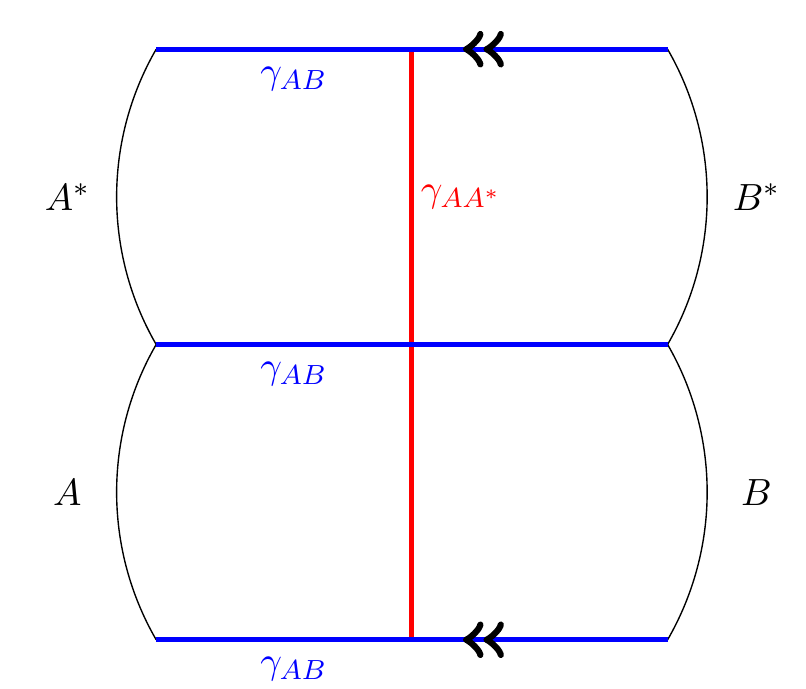}
	}
	}
    \caption{(a) A time slice of the AdS$_3$ vacuum with two boundary intervals $A$ and $B$ labeled. The intervals have been chosen large enough so that their minimal surface is $\gamma_{AB}$. The entanglement wedge cross-section is labeled $\sigma_{A:B}$. (b) A time slice of the canonical purification, formed by pasting $\W(\rho_{AB})$ to its CPT conjugate along the surface $\gamma_{AB}$. The top and bottom edges of this figure are identified. The minimal surface of boundary region $AA^*$ is the image of $\sigma_{A:B}$ under the symmetry $A \leftrightarrow A^*, B \leftrightarrow B^*$.}
    \label{fig:two-intervals-surfaces}
\end{figure}

The area contribution to the entropy $S_{\kett{\sqrt{\rho_{AB}}}}(AA^*)$ is twice the area of the entanglement wedge cross-section $\sigma_{A:B}$. The bulk entropy contribution is easily seen to be the reflected entropy of the bulk fields in $\W(\rho_{AB})$ with respect to the bipartition induced by the dividing surface $\sigma_{A:B}$. From this, we find
\begin{equation}
    S_{\text{gen}}(\gamma_{AA^*}|AA^*)
        = \frac{2 \area(\sigma_{A:B})}{4 G_N} + S_{R, \text{bulk}}(\sigma_{A:B}).
\end{equation}
The quantum extremal surface formula tells us that $S_{\kett{\sqrt{\rho_{AB}}}}(AA^*)$ is the minimum of the above quantity over all quantum extremal surfaces $\gamma_{AA^*}$, which we can rewrite without reference to the canonical purification as\footnote{A generalization of this formula, which includes explicit contributions from entanglement islands, has appeared in \cite{reflected-entropy-islands}.}
\begin{equation} \label{eq:dutta-faulkner}
    S_R(A:B)
        = \min_{\sigma_{A:B}} \left[ \frac{2 \area(\sigma_{A:B})}{4 G_N} + S_{R, \text{bulk}}(\sigma_{A:B}) \right],
\end{equation}
where the minimum is taken over all candidates for the entanglement wedge cross-section that locally extremize the quantity in brackets. We will henceforth use the notation $\sigma_{A:B}$ to refer to the surface that achieves this minimum.

Many interesting properties of the reflected entropy were explored in \cite{dutta-faulkner}. One of particular interest is that in any quantum state, the reflected entropy always exceeds the mutual information:
\begin{equation}
    S_R(A:B) \geq I(A:B).
\end{equation}
While this inequality holds for \emph{all} quantum states, it has a particularly suggestive geometric interpretation for holographic states satisfying the quantum extremal surface formula. For any bipartite boundary state $\rho_{AB}$ with a semiclassical dual, one can define a special surface within the homology class of $A$ by taking the union of the entanglement wedge cross-section $\sigma_{A:B}$ together with the portions of the $AB$ quantum extremal surface that lie ``between $\sigma_{A:B}$ and $A$''. This is sketched for a simple case in figure \ref{fig:KRT-defining-figure}. The surface thus constructed is homologous to $A$.\footnote{In fact, the requirement that there is a bipartition of the quantum extremal surface of $AB$ into pieces $\gamma^{(A)} \cup \gamma^{(B)}$ such that $\gamma^{(A)} \cup \sigma_{A:B}$ is homologous to $A$ and $\gamma^{(B)} \cup \sigma_{A:B}$ is homologous to $B$ could be taken as the definition of an entanglement wedge cross-section --- see section \ref{subsec:general-proof} for more discussion on this point.} We will call this surface the \emph{KRT surface} $\KRT(A)$ --- for ``kinked Ryu-Takayanagi,'' because the surface will have right-angled kinks where $\sigma_{A:B}$ meets the quantum extremal surface of $AB$, as in figure \ref{fig:KRT-defining-figure} --- and we will call the true minimal quantum extremal surface within that homology class $\RT(A)$.\footnote{Veterans of the field might object: we've taken such great pains to ensure the present discussion holds for all quantum extremal surfaces, why should we suddenly start calling them RT surfaces when the term ``RT'' is usually used to refer to classically minimal surfaces in time-symmetric states? Unfortunately, we've found the acronyms $\mathrm{QES}(A)$ and $\mathrm{KQES}(A)$ just slightly too long to fit in our figures; we have elected to simplify matters by calling all minimal quantum extremal surfaces ``RT surfaces.'' Our apologies to colleagues Dong, Engelhardt, Faulkner, Headrick, Hubeny, Lewkowycz, Maldacena, Rangamani, and Wall, all of whom have made important contributions to the general quantum extremal surface formula but who will hopefully understand our avoiding the acronym DEFHHLMRRTW.}

\begin{figure}
    \centering
    \includegraphics{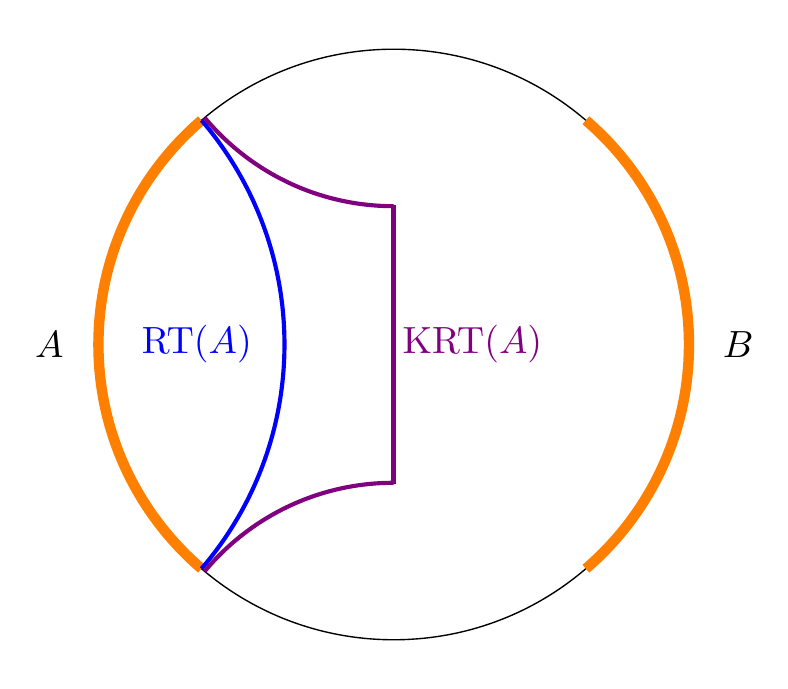}
    \caption{A sketch of the surfaces $\RT(A)$ and $\KRT(A)$ for the case where $A$ and $B$ are two equal-time intervals in the AdS$_3$ vacuum. $\KRT(A)$ is formed by taking the union of the entanglement wedge cross-section $\sigma_{A:B}$ with the portion of $\RT(AB)$ lying ``towards $A$'' from the cross-section. Both $\KRT(A)$ and $\RT(A)$ are homologous to $A$.}
    \label{fig:KRT-defining-figure}
\end{figure}

The area term in $S_R(A:B) - I(A:B)$ is given by
\begin{equation}
    [S_R - I]_{\area}
        = \frac{2 \area(\sigma_{A:B}) - \area(\RT(A)) - \area(\RT(B)) + \area(\RT(AB))}{4 G_N},
\end{equation}
which we may rewrite in terms of KRT surfaces as
\begin{equation} \label{eq:SRmI-area}
    [S_R - I]_{\area}
        = \frac{\area(\KRT(A)) - \area(\RT(A))}{4 G_N} + \frac{\area(\KRT(B)) - \area(\RT(B))}{4 G_N}.
\end{equation}
To put the bulk entropy term of $S_R - I$ in a suggestive form, we apply the general inequality $S_R \geq I$ for the bulk state in $\rho_{AB}$ subject to the bipartition induced by the entanglement wedge cross-section $\sigma_{A:B}.$ Using the definition of the mutual information $I(P:Q)=S(P)+S(Q)-S(PQ)$, we obtain
\begin{equation} \label{eq:SRmI-bulk}
    S_{R, \text{bulk}}(\sigma_{A:B}) \geq S_{\text{bulk}}(\KRT(A)|A) + S_{\text{bulk}}(\KRT(B)|B) - S_{\text{bulk}}(\RT(AB)|AB),
\end{equation}
where $S_{\text{bulk}}$ was defined in footnote \ref{footnote:gen-entropy}.

Using equations \eqref{eq:SRmI-area} and \eqref{eq:SRmI-bulk}, together with the boundary reflected entropy formula \eqref{eq:dutta-faulkner} and the boundary mutual information formula
\begin{equation}
    I(A:B) = S_{\text{gen}}(\RT(A)|A) + S_{\text{gen}}(\RT(B)|B) - S_{\text{gen}}(\RT(AB)|AB),
\end{equation}
we obtain the inequality
\begin{align} \label{eq:SRmI-Sgen}
    S_R(A:B) - I(A:B)
        & \geq \, [S_{\text{gen}}(\KRT(A)|A) - S_{\text{gen}}(\RT(A)|A)] \nonumber \\
               & \qquad + [S_{\text{gen}}(\KRT(B)|B) - S_{\text{gen}}(\RT(B)|B)].
\end{align}
Because $\KRT(A)$ (respectively $\KRT(B)$) is in the same homology class as $\RT(A)$ (respectively $\RT(B)$), we might expect that its generalized entropy is superminimal, and thus that each bracketed term on the right-hand side of the above expression is individually nonnegative. One has to be a little careful with this statement, though, because $\RT(A)$ does not have minimal generalized entropy within its homology class, but rather minimal generalized entropy \emph{among the quantum extremal surfaces in its homology class}. This caveat is necessary to allow for the fact that highly boosted surfaces in the $\RT(A)$ homology class can have arbitrarily small areas. This is an obstacle to understanding inequality \eqref{eq:SRmI-Sgen}, because KRT surfaces are not generally quantum extremal. For states with no quantum matter, however, nonnegativity of the right-hand side of \eqref{eq:SRmI-Sgen} can be established using the maximin formula \cite{maximin}. If one accepts the quantum focusing conjecture \cite{bousso2016quantum} and the quantum maximin formula \cite{quantum-maximin}, then nonnegativity of the right-hand side of inequality \eqref{eq:SRmI-Sgen} can be established in complete generality. The details of these arguments are given in appendix \ref{app:KRT}.

While we did have to apply $S_R \geq I$ to the bulk quantum state in order to obtain inequality \eqref{eq:SRmI-Sgen}, the final expression is quite suggestive. It tells us that the quantity $S_R - I$ is related to the difference in generalized entropy between two preferred surfaces in the $A$ and $B$ homology classes --- the surfaces $\KRT(A)$ and $\RT(A)$ (respectively $\KRT(B)$ and $\RT(B)$). The difference in generalized entropy between certain special candidate surfaces in a homology class is an interesting physical quantity; one might imagine that by studying it more carefully, we could improve our understanding of the inequality $S_R \geq I$ in holographic states.  The goal of the present work is to undertake that analysis, and in doing so to prove a stronger inequality than $S_R \geq I$ in certain holographic states. Indeed, holographic states generally satisfy some special constraints on their entanglement structure; the monogamy of mutual information 
\cite{hayden2013holographic} is an example of an inequality that is satisfied by classical holographic states, but not by general quantum states. Finding such constraints which are special to holographic states can in turn provide insight into the entanglement structure of these states.\footnote{There is an interesting related research program that attempts to classify all von Neumann entropy inequalities implied by the Ryu-Takayanagi formula; see for example \cite{bao2015holographiccone, hubeny2018holographicrelations}.}

The main technical contribution of this paper is a proof that in time-symmetric states of pure AdS$_3$ gravity satisfying the Ryu-Takayanagi formula, the inequality $S_R \geq I$ can be enhanced to
\begin{equation} \label{eq:big-technical-claim}
    S_R(A:B) - I(A:B) \geq \frac{\log(2) \ell_{\text{AdS}}}{2 G_N} \times (\text{\# of cross-section boundaries}) + o\left(\frac{1}{G_N}\right).
\end{equation}
In AdS$_3$ gravity, the entanglement wedge cross-section is a one-dimensional curve --- the ``\# of cross-section boundaries'' appearing in the above inequality is the number of endpoints of that curve. The main conceptual contribution of this paper will be an interpretation of the quantity $S_R - I$ --- which we call the \emph{Markov gap} for reasons made clear in section \ref{sec:info-theory} --- in terms of the optimal fidelity of a particular recovery process on the canonical purification $\kett{\sqrt{\rho_{AB}}}$. The interpretation of the Markov gap in terms of a recovery process will suggest a close link between the quantity $S_R - I$ and the geometry of the boundaries of the entanglement wedge cross-section; this link will be bolstered by the proof of inequality \eqref{eq:big-technical-claim}. We will also comment on potential generalizations of this inequality for non-time-symmetric states, states in higher-dimensional theories of gravity, and states with bulk matter.

Before proceeding to the plan of the paper, we detour to highlight a link between our work and a very interesting article by Akers and Rath \cite{akers-rath-tripartite}. In that paper, the authors used the Dutta-Faulkner formula for holographic reflected entropy to disprove a conjecture made in \cite{mostly-bipartite} about the entanglement structure of holographic states. Motivated by structural theorems about bit threads calculating holographic entropies, the original conjecture was that entanglement between three spatial regions in semiclassical boundary states is mostly \emph{bipartite} in nature, in the sense that any multipartite contribution would be subleading in $1/G_N$. (The analogous statement for four boundary regions, by contrast, was already known to be false~\cite{balasubramanian2014multiboundary}.) In \cite{akers-rath-tripartite}, Akers and Rath showed that states with mostly bipartite entanglement patterns satisfy $S_R - I \approx 0.$ They wrote down some particular boundary states for which $S_R - I$ is order $1/G_N$, proved some useful statements about the continuity of the reflected entropy, and concluded that the states in question must have significant amounts of tripartite entanglement. Combining their observations with our inequality \eqref{eq:big-technical-claim} makes for a tantalizing observation: the mere \textit{presence} of boundaries in an entanglement wedge cross-section requires significant quantities of tripartite entanglement in the corresponding boundary state. The fact that our inequality \eqref{eq:big-technical-claim} scales with the \emph{number} of cross-section boundaries seems to suggest that \emph{every} boundary in the entanglement wedge cross-section must be supported by its own pattern of tripartite entanglement in the boundary state --- this idea is taken up in considerably more detail in section \ref{subsec:bulk-gap}. 

The nature of that tripartite entanglement remains to be understood but at least one familiar form can be ruled out. The monogamy of mutual information~\cite{hayden2013holographic, maximin} eliminates a four- or higher-party GHZ entanglement structure in classical states in holography, because GHZ states explicitly violate the inequality.\footnote{More precisely, states that are \textit{entirely} GHZ are excluded by monogamy. It is conceivable that GHZ-type entanglement could co-exist with other types of entanglement in such a way that the monogamy property remained satisfied overall even if it were violated by a GHZ-entangled factor in the boundary Hilbert space.} Even the three-party GHZ state can be ruled out as follows: consider the two party state obtained by tracing out one of the factors, and construct the bulk geometry dual to the canonical purification of that two party state using the Dutta-Faulkner trick described above. Had the original state been GHZ, the canonical purification would necessarily also be GHZ and violate the monogamy of mutual information, which is then a contradiction.
It is our hope that inequality \eqref{eq:big-technical-claim}, and related inequalities that may be provable in more general holographic theories of gravity, will contribute to our understanding of the thorny but fascinating problem of multipartite holographic entanglement.\footnote{While not directly related to the present paper, we would also like to draw attention to the series of papers \cite{ning1,ning2,ning3,ning4}, which undertake an orthogonal approach to understanding multipartite entanglement by defining multi-party generalizations of entanglement wedge cross-sections.}

\begin{figure}
    \centering
    \includegraphics{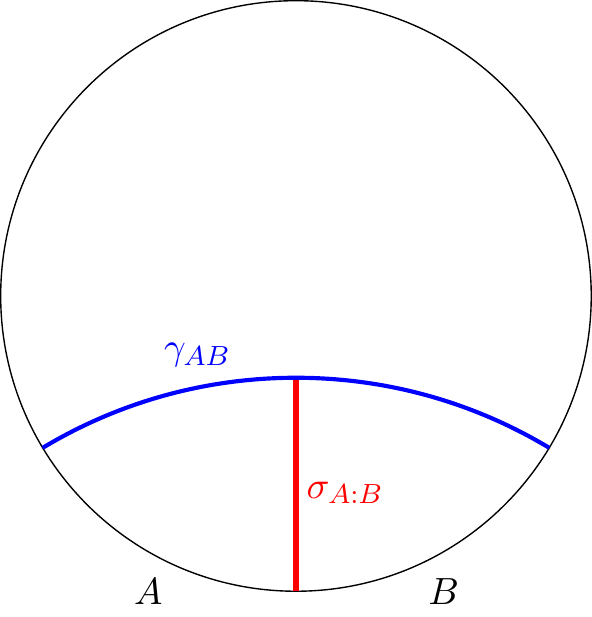}
    \caption{The minimal surface $\gamma_{AB}$ and entanglement wedge cross-section $\sigma_{A:B}$ for the case where $A$ and $B$ are neighboring intervals on a time slice of the AdS$_3$ boundary. In this case, the entanglement wedge cross-section has only one boundary. It has no boundary at infinity, since that point is at infinite distance.}
    \label{fig:neighboring-intervals}
\end{figure}

One other related paper worth highlighting is \cite{zou2021universal}. In that paper, the authors started from the observation that in the vacuum state of pure AdS$_3$ gravity, the quantity $S_R - I$ for two neighboring, equal-time intervals is independent of the size of those intervals and equal to $\log(2) \ell_{AdS} / 2 G_N$. This is a saturation of inequality \eqref{eq:big-technical-claim}, because in that case (sketched in figure \ref{fig:neighboring-intervals}), the entanglement wedge cross-section has only one boundary. The authors computed $S_R - I$ --- and the related quantity $2 E_P - I$ with $E_P$ the entanglement of purification --- in various critical spin chains and found interesting universal features.

The plan of the present paper is as follows.

In section \ref{sec:info-theory}, we introduce the idea of a \emph{Markov recovery process} from quantum information theory, and show that $S_R(A:B) - I(A:B)$ is lower bounded by a function of the fidelity of a particular Markov recovery process on $\kett{\sqrt{\rho_{AB}}}$. We give a preliminary holographic interpretation of this recovery process, and argue that boundaries of the entanglement wedge cross-section present obstructions to this recovery process that are reflected in the Markov gap. The geometric argument of this section is largely heuristic, and relies on the general principle from \cite{van2010building} that geometric connections in a bulk state are supported by entanglement patterns in the boundary state.

In section \ref{sec:geometric-proof}, we prove the claimed inequality \eqref{eq:big-technical-claim} for time-symmetric states in pure AdS$_3$ quantum gravity. The proof is an exercise in two-dimensional hyperbolic geometry; we first prove the desired lower bound on a quantity analogous to $\area(\KRT(A)) - \area(\RT(A))$ defined purely in the hyperbolic disc, then use the fact that every hyperbolic $2$-manifold is covered by the hyperbolic disc to prove the general bound. One charming feature of the proof is that it constructs a partial tiling of the homology region between $\KRT(A)$ and $\RT(A)$ by right-angled hyperbolic pentagons, with each kink in $\KRT(A)$ being a vertex of its own pentagon; the desired inequality follows from standard trigonometric relations among the side-lengths of these pentagons.

In section \ref{sec:generalizations}, we discuss potential generalizations of inequality \eqref{eq:big-technical-claim} beyond the regime of time-symmetric, pure, AdS$_3$ gravity. We provide evidence that the bound continues to hold in AdS$_3$ gravity upon the addition of bulk matter, and discuss a possible generalization of \eqref{eq:big-technical-claim} to bulk dimensions greater than $3$.

In section \ref{sec:recovery-models}, we return to the understanding developed in section \ref{sec:info-theory} relating the Markov gap to recovery processes on $\kett{\sqrt{\rho_{AB}}}$. We create a tractable model of the relevant recovery process using the fixed area states introduced in \cite{fixed-area-AR, fixed-area-DHM}, and compute a bound on the fidelity of the process using the gravitational path integral. This computation reproduces some features of the inequality \eqref{eq:SRmI-Sgen} directly from the perspective of Markov recovery processes. The techniques developed in this section may be of independent interest --- the story we tell about recovery processes can be reduced to a more general question that is physically interesting in its own right:
\begin{itemize}
    \item Given a boundary region $R$ and a state $\rho_{R}$ whose bulk entanglement wedge is bounded by the bulk surface $\gamma$, and which contains another covariantly defined surface $\tilde{\gamma}$ with greater generalized entropy, define $\tilde{\rho}_{R}$ to be a state whose entanglement wedge ends at $\tilde{\gamma}$. What is the fidelity $F(\rho_R, \tilde{\rho}_{R})$?\footnote{The curious reader is encouraged to look ahead to figure \ref{fig:two-QES-state} for an illustration of this setup.}
\end{itemize} 
In a fixed-area-state model of this question, when the state $\tilde{\gamma}$ is extremal, we compute not only the fidelity but all of the sandwiched R\'{e}nyi relative entropies defined in \cite{muller-lennert_quantum_2013, wilde_strong_2014}; this result is given in equations \eqref{eq:sandwiched-Renyis} and \eqref{eq:reversed-sandwiched-Renyis}. Readers primarily interested in understanding this result can skip directly to the matching bullet of section \ref{subsec:fixed-area} without losing any essential context.

Finally, in section \ref{sec:discussion}, we conclude with a summary of our results and some possible directions for future work. We pay particular attention to what the present explorations might teach us about multipartite holographic entanglement, and to what techniques might be used to generalize inequality \eqref{eq:big-technical-claim}.

In appendix \ref{app:KRT}, we show that the inequality $S_{\text{gen}}(\KRT(A)|A) \geq S_{\text{gen}}(\RT(A)|A)$ follows from the quantum maximin formula.

Throughout the paper we set $c = \hbar = 1$ while leaving $G_N$ explicit. We frequently work in units where the AdS radius is set to $1$. The term ``area'' is used to refer to the volume of any codimension-$2$ bulk hypersurface, regardless of the total bulk dimension. Once or twice we use the term ``area'' to refer to the volume of a codimension-$3$ hypersurface, but we signpost this explicitly in the relevant sections. When distinguishing between the regimes of validity of the ``minimal'' versus ``extremal'' surface formulas for holographic entanglement entropy, we use the term ``time-symmetric'' to refer to any spacetime that has a moment of time symmetry --- i.e., a complete, achronal slice whose extrinsic curvature tensor vanishes. Time-symmetric states do not necessarily have a global $t \mapsto -t$ symmetry.

%%%%%%%%%%%%%%%%%%%%%%%
\section{Information-theoretic origin of the Markov gap}
\label{sec:info-theory}

As mentioned in the introduction, the quantity $S_R(A:B) - I(A:B)$ can be related to the fidelity of a particular Markov recovery process on the canonical purification of $\rho_{AB}$. In subsection \ref{subsec:QI-prelims}, we review and explain the results from quantum information theory needed to understand this claim. We define Markov recovery maps and the quantum fidelity of states, and give a refinement of the inequality $S_R - I \geq 0$ in terms of the fidelity of a Markov recovery map. In subsection \ref{subsec:bulk-gap}, we give a holographic interpretation of the refined inequality, and explain why the Markov gap must be nonzero at order $1/G_N$ whenever the entanglement wedge cross-section of $\rho_{AB}$ has a nontrivial boundary.

%%%%%%%%%%%
\subsection{Quantum information preliminaries}
\label{subsec:QI-prelims}

Given a three-party quantum state $\rho_{ABC}$, a \emph{Markov recovery map} is a quantum channel from a one-party subsystem into a two-party subsystem. For example, we might have a map $\mathcal{R}_{B \rightarrow BC}$ that takes system $B$ into system $BC$. By acting on the reduced state $\rho_{AB}$ with this channel, we produce a tripartite state on the whole system:
\begin{equation}
    \tilde{\rho}_{ABC} = \mathcal{R}_{B \rightarrow BC}(\rho_{AB}).
\end{equation}
By a \emph{Markov recovery process}, we will mean the problem of trying to reproduce the tripartite state $\rho_{ABC}$ from one of its bipartite reduced states (for example $\rho_{AB}$) using a Markov recovery map (for example $\mathcal{R}_{B \rightarrow BC}$.)

The name ``Markov'' was first associated with this kind of problem in \cite{markov-states} --- based on earlier work in \cite{accardi1983markovian} on an analogous problem involving countably many systems --- where a state $\rho_{ABC}$ that can be \emph{perfectly} recovered from $\rho_{AB}$ via a quantum channel on $B$ was called a \emph{short quantum Markov chain} for the ordering $A \rightarrow B \rightarrow C$. I.e., $\rho_{ABC}$ is a quantum Markov chain for $A \rightarrow B \rightarrow C$ if there exists a quantum channel $\mathcal{R}_{B \rightarrow BC}$ satisfying
\begin{equation}
    \rho_{ABC} = \mathcal{R}_{B \rightarrow BC}(\rho_{AB}).
\end{equation}
It follows from the results of \cite{petz1986sufficient} that $\rho_{ABC}$ is a quantum Markov chain for $A \rightarrow B \rightarrow C$ if and only if the conditional mutual information $I(A:C|B)$ vanishes.\footnote{Classically, the random variables $X, Y, Z$ form a Markov chain if and only if their joint distribution factorizes as $p(x,y,z) = p(x)p(y|x)p(z|y)$. When $I(A:C|B)$ vanishes in the quantum case, this factorization generalizes to a canonical form for $\rho_{ABC}$~\cite{markov-states}.}

The conditional mutual information is defined as the following linear combination of entanglement entropies:
\begin{equation}
    I(A:C|B) = S(AB) + S(BC) - S(ABC) - S(B) = I(A:BC) - I(A:B).
\end{equation}
The famous strong subadditivity inequality is exactly the statement that conditional mutual information is always nonnegative. Roughly speaking, $I(A:C|B)$ measures the amount of correlation between the $A$ and $C$ subsystems that does not ``pass through'' the $B$ subsystem. When $I(A:C|B)$ vanishes, all correlations between $A$ and $C$ are visible to the $B$ subsystem, and can be reproduced by a quantum channel acting only on $B$; this is the intuition behind the statement that $A \rightarrow B \rightarrow C$ is a quantum Markov chain iff $I(A:C|B)$ vanishes.

In \cite{fawzi2015quantum}, the authors proved a refinement of the relationship between Markov chains and conditional mutual information.\footnote{A series of follow-up papers culminating in \cite{junge2018universal} actually proved a much more general result --- a refinement of the monotonicity of relative entropy under quantum channels --- of which strong subadditivity is a special case. A good review of the general theorem is presented in chapter 12 of \cite{WildeBook}, with the application to conditional mutual information given in section 12.6.1.} They showed that states with small (but nonzero) conditional mutual information are \emph{approximate} Markov chains --- that is, they admit Markov recovery processes that do a good job reproducing $\rho_{ABC}$ as measured by the quantum fidelity. The formal statement is as follows:
\begin{equation} \label{eq:first-fidelity-inequality}
    \max_{\mathcal{R}_{B \rightarrow BC}} F(\rho_{ABC}, \mathcal{R}_{B \rightarrow B C}(\rho_{AB})) \geq e^{-I(A:C|B)}. 
\end{equation}
The quantity $F$ appearing on the left-hand side of this inequality is the quantum fidelity, defined by
\begin{equation}
    F(\rho, \sigma) = \left[ \tr\sqrt{\sqrt{\rho} \sigma \sqrt{\rho}} \right]^2.
\end{equation}
It is symmetric in its arguments, lies in the range $0 \leq F(\rho, \sigma) \leq 1$, equals $1$ if and only if $\rho$ and $\sigma$ are equal, and equals $0$ if and only if $\rho$ and $\sigma$ have orthogonal support. If two states have fidelity close to $1$, then for any bounded observable $\mathcal{O}$ the expectation values $\tr(\rho \mathcal{O})$ and $\tr(\sigma \mathcal{O})$ are similar.\footnote{The precise version of this ``similarity'' statement goes as follows. The fidelity is related to the one-norm distance by (see eqs. (9.100-9.101) of \cite{nielsen-chuang}, though note their definition of fidelity differs by a power of $2$)
\begin{equation*}
   \lVert \rho - \sigma \rVert_1 \leq 2 \sqrt{1 - F(\rho, \sigma)},
\end{equation*}
so when $F$ is close to one, the one-norm distance is close to zero. The Schatten norms satisfy a form of H\"{o}lder's inequality, giving
\begin{equation*}
    |\tr\left((\rho - \sigma) \mathcal{O}\right)| \leq \tr\left( \left|(\rho - \sigma) \mathcal{O} \right| \right) = \lVert (\rho - \sigma) \mathcal{O} \rVert_1 \leq \lVert \rho - \sigma \rVert_1 \lVert \mathcal{O} \rVert_{\infty},
\end{equation*}
where $\lVert \cdot \rVert_{\infty}$ is the \emph{operator norm}, equal to the largest absolute value among the eigenvalues of $\mathcal{O}.$} Elementary properties of the fidelity are nicely reviewed in chapter 9.2.2 of \cite{nielsen-chuang}.

Let us now take a closer look at inequality \eqref{eq:first-fidelity-inequality}. The left-hand side of the inequality involves a maximum over all Markov recovery channels $\mathcal{R}_{B \rightarrow BC}$. If $I(A:C|B)$ is close to zero, then the right-hand side of the inequality is close to 1; this means that there must exist \emph{some} Markov recovery map $\mathcal{R}_{B \rightarrow BC}$ for which the fidelity $F(\rho_{ABC}, \mathcal{R}_{B\rightarrow BC}(\rho_{AB}))$ is close to one. So when $I(A:C|B)$ is small, there exists some Markov recovery map on $B$ that approximately recovers $\rho_{ABC}$ from $\rho_{AB}$. The way this inequality was proved in \cite{fawzi2015quantum,junge2018universal} was by constructing a particular Markov recovery map --- called the \emph{rotated Petz map} or \emph{twirled Petz map} --- that beats $e^{- I(A:C|B)}$ in terms of fidelity.\footnote{A curious feature of the twirled Petz map is that the reconstructed state
\begin{equation*}
    \tilde{\rho}_{ABC} = \mathcal{R}_{\text{tPetz}, B \rightarrow BC}(\rho_{AB})
\end{equation*}
is not only very close to $\rho_{ABC}$ in terms of the fidelity, but is \emph{exactly} equal to $\rho_{ABC}$ on the $BC$ subsystem; i.e., we have $\rho_{BC} = \tilde{\rho}_{BC}.$ Again, we refer the reader to section 12.6.1 of \cite{WildeBook} for a review of this fact.}

In its current form, inequality \eqref{eq:first-fidelity-inequality} looks like a bound on the fidelity of an optimal Markov recovery process in terms of the conditional mutual information. But we are free to rewrite it as a bound on the conditional mutual information in terms of the fidelity of the optimal Markov recovery process:
\begin{equation} \label{eq:second-fidelity-inequality}
    I(A:C|B) \geq - \max_{\mathcal{R}_{B \rightarrow BC}} \log F(\rho_{ABC}, \mathcal{R}_{B \rightarrow B C}(\rho_{AB})).
\end{equation}
We can view this as a state-dependent enhancement of the strong subadditivity inequality. The ordinary strong subadditivity inequality is $I(A:C|B) \geq 0$. If we happen to know, however, that the optimal Markov recovery process for the chain $A \rightarrow B \rightarrow C$ is imperfect --- i.e., its fidelity is less than one --- then inequality \eqref{eq:second-fidelity-inequality} tells us that $I(A:C|B)$ must be bounded away from zero.

For the reflected entropy of a bipartite state $\rho_{AB}$, the quantity $S_R(A:B) - I(A:B)$ can be written as a conditional mutual information of the canonical purification:
\begin{equation}
    S_R(A:B) - I(A:B) = I(A:B^*|B) = I(B:A^*|A).
\end{equation}
This observation was made in \cite{dutta-faulkner}, and it is in fact how they proved the inequality $S_R \geq I$. Using inequality \eqref{eq:second-fidelity-inequality}, however, we can prove two stronger inequalities:
\begin{align}
    S_R(A:B) - I(A:B) \label{eq:BBstar-Markov}
        & \geq - \max_{\mathcal{R}_{B \rightarrow B B^*}} \log F(\rho_{A B B^*}, \mathcal{R}_{B \rightarrow B B^*}(\rho_{AB})). \\
    S_R(A:B) - I(A:B)
        & \geq - \max_{\mathcal{R}_{A \rightarrow A A^*}} \log F(\rho_{A A^* B}, \mathcal{R}_{A \rightarrow A A^*}(\rho_{AB})).
\end{align}
These inequalities are what lead us to call the quantity $S_R - I$ the \emph{Markov gap}.

%%%%%%%%%%%
\subsection{Cross-section boundaries and bulk entanglement}
\label{subsec:bulk-gap}

We will now argue on quite general grounds --- albeit rather heuristically --- that in holographic states, the fidelities of Markov recovery processes on the canonical purification are controlled by the number of boundaries in the entanglement wedge cross-section. Put simply: the more boundaries there are in the entanglement wedge cross-section, the harder it is for a Markov recovery process to accurately reproduce the chain $A \rightarrow B \rightarrow B^*$ (or $B \rightarrow A \rightarrow A^*$). This discussion serves as a prelude to the concrete analysis of section \ref{sec:geometric-proof}, where we will prove our claimed inequality \eqref{eq:big-technical-claim} lower-bounding $S_R - I$ in terms of the number of cross-section boundaries for time-symmetric states in AdS$_3$ gravity.

Let us examine in some detail a particular example of a canonical purification. Let $\rho_{AB}$ be the density matrix of two equal-time intervals in the AdS$_3$ vacuum, with the intervals being large enough that the entanglement wedge of $\rho_{AB}$ is connected. In figure \ref{fig:intervals-jagged-surfaces}, we have sketched a static time-slice of the canonical purification of $\rho_{AB}$. As explained in the introduction, the bulk dual of a canonical purification has, as initial data, the geometry formed by gluing together two copies of the entanglement wedge along their spacelike boundaries. For the specific case of two equal-time intervals in the AdS$_3$ vacuum, the canonical purification is a two-boundary wormhole.

\begin{figure}
    \centering
    \includegraphics{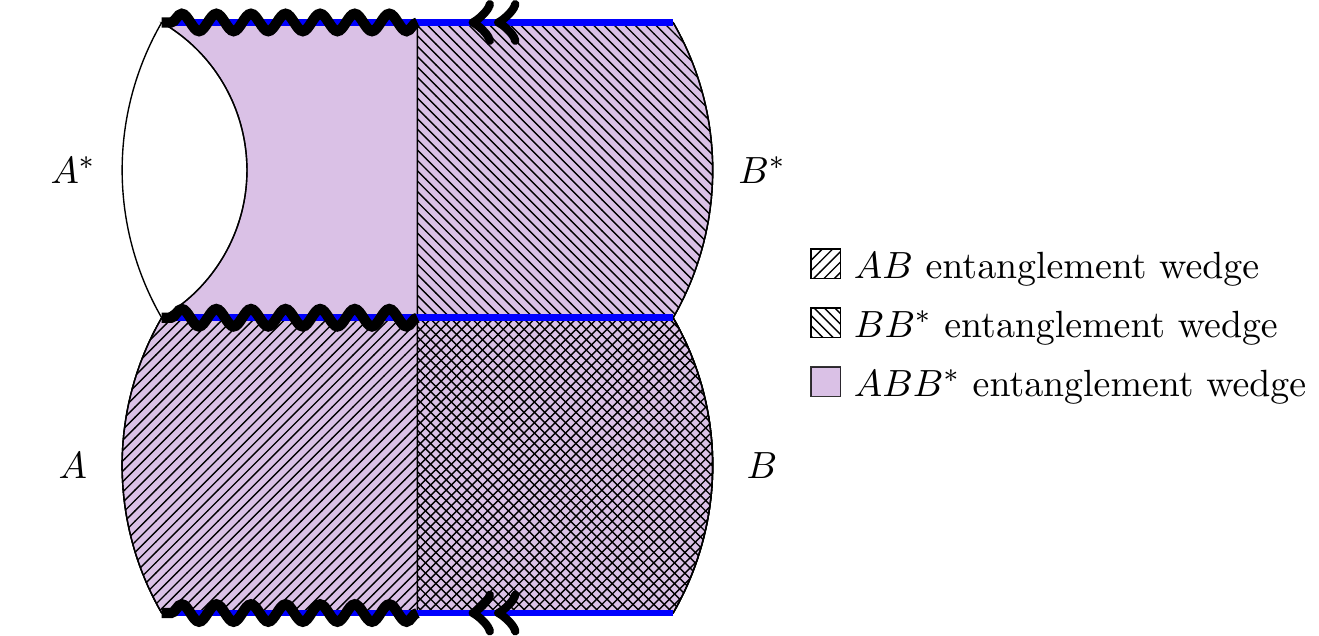}
    \caption{A sketch of the canonical purification for the density matrix of two intervals in the AdS$_3$ vacuum; the top and bottom lines are identified, making the entire geometry a two-boundary wormhole. The $ABB^*$ entanglement wedge is shaded in violet, and the $AB$ and $BB^*$ entanglement wedges are indicated with crosshatching. The blue lines are minimal surfaces for $AB$ and $A^* B^*$. Jagged curves have been imposed over two minimal surfaces whose tubular neighborhoods are visible to $ABB^*$ but not to $AB$ or $BB^*$; we argue in the text of the paper that these surfaces must be supported by boundary entanglement contributing to $I(A:B^*|B).$}
    \label{fig:intervals-jagged-surfaces}
\end{figure}

We will focus on the Markov chain $A \rightarrow B \rightarrow B^*$, though we could equally well study the chain $B \rightarrow A \rightarrow A^*$ with similar conclusions. We would like to know, abstractly, how hard it is to recover the 3-party state $\rho_{ABB^*}$ from the two-party state $\rho_{AB}$ with a Markov map $\mathcal{R}_{B \rightarrow B B^*}.$ The bulk regions dual to the density matrices $\rho_{AB}$, $\rho_{BB^*}$, and $\rho_{ABB^*}$ are indicated with shading and crosshatching in figure \ref{fig:intervals-jagged-surfaces}. In particular, any sufficiently small tubular neighborhood of the ``jagged'' surfaces indicated in figure \ref{fig:intervals-jagged-surfaces} with thick, wavy lines is contained in the entanglement wedge of $\rho_{ABB^*}$, but no tubular neighborhood is contained entirely within the entanglement wedges of either $\rho_{AB}$ or $\rho_{BB^*}$.

There is by now very good reason to believe that geometric connections in bulk states are sourced, in a meaningful sense, by large amounts of boundary entanglement. This heuristic, originally advocated by Van Raamsdonk in \cite{van2010building} using the Ryu-Takayanagi formula, is the fundamental principle underlying many important developments such as ``ER = EPR'' \cite{maldacena2013cool} and entanglement wedge reconstruction behind black hole horizons \cite{islands1, islands2, RWW}.\footnote{The ``geometric connection'' appearing in the behind-the-horizon reconstruction problem is most apparent in the doubly holographic model of entanglement islands introduced in \cite{AMMZ}. See also \cite{Anderson:2021vof} for a singly holographic realization.} If we apply this heuristic to the situation elaborated in the preceding paragraph, it seems natural to conclude that the boundary entanglement sourcing the geometric connections across the jagged surfaces of figure \ref{fig:intervals-jagged-surfaces} is visible to $\rho_{ABB^*}$, but not to $\rho_{AB}$ or $\rho_{BB^*}$. This suggests that we could never hope to reproduce that entanglement by acting on $\rho_{AB}$ with a Markov channel $\mathcal{R}_{B \rightarrow BB^*}$ --- the entanglement sourcing the smooth geometries of the jagged surfaces isn't already present in $\rho_{AB}$, and it can't be added to the state with a channel that doesn't access the $A$ subsystem. \textit{The very existence of the jagged surfaces in figure \ref{fig:intervals-jagged-surfaces} precludes a perfect Markov recovery $\rho_{ABB^*} = \mathcal{R}_{B \rightarrow BB^*}(\rho_{AB}).$} This, in turn, guarantees a nontrivial Markov gap by inequality \eqref{eq:BBstar-Markov}.

%From the preceding discussion, we might be led to believe that each boundary in the entanglement wedge cross-section necessarily makes some contribution to the Markov gap, because every boundary in the entanglement wedge cross-section is attached to a ``jagged surface'' supported by entanglement contributing to $I(A:B^*|B)$. This is \emph{almost} what we believe to be true, but we still need to argue, as we will now do, that the entanglement supporting the ``jagged surfaces'' is less fundamental than the entanglement supporting the geometric structure of a ``corner'' formed by the intersection of a jagged surface and an entanglement wedge cross-section. 

The preceding discussion suggests that each boundary in the entanglement wedge cross-section necessarily makes some contribution to the Markov gap, because every boundary in the entanglement wedge cross-section is attached to a ``jagged surface'' supported by entanglement contributing to $I(A:B^*|B)$. In our detailed analysis below, we will focus our attention on ``corners,'' geometric structures formed by the intersection of a jagged surface and an entanglement wedge cross-section. As we will see, each corner is supported by entanglement making an irreducible contribution to the Markov gap, even in the limit that the length of the jagged surface itself vanishes. In that sense, the contribution of these corners to the Markov gap can be viewed as more fundamental than the entanglement supporting the jagged surface itself.

To argue for this intuitive principle, we will examine a different canonical purification than the one already considered. Let $\rho_{AB}$ be the density matrix of two asymptotic boundaries of a three-boundary wormhole geometry, with the third asymptotic boundary being sufficiently small that the $\rho_{AB}$ entanglement wedge is connected. A static time-slice of this geometry and a static time-slice of its canonical purification are sketched in figure \ref{fig:wormhole-jagged-surfaces}.

\begin{figure}
    \centering
    \includegraphics{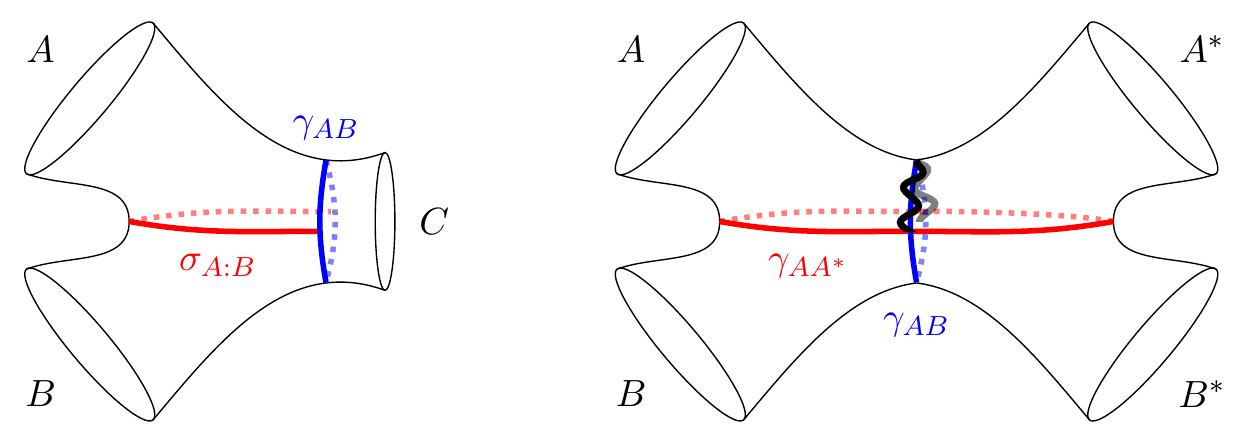}
    \caption{The left side of this figure shows a three boundary wormhole with $\rho_{AB}$ having a connected entanglement wedge; the $AB$ extremal surface is marked $\gamma_{AB}$, and the entanglement wedge cross-section is labeled $\sigma_{A:B}$. The right side shows a time-slice of the spacetime corresponding to the canonical purification. On the canonical purification, we have marked with a jagged curve the portion of $\gamma_{AB}$ whose tubular neighborhoods are visible to $\rho_{ABB^*}$ but not to $\rho_{AB}$ or $\rho_{BB^*}.$}
    \label{fig:wormhole-jagged-surfaces}
\end{figure}

There are two important features of this canonical purification that differ from the one sketched in figure \ref{fig:intervals-jagged-surfaces}. First, note that there is only \emph{one} jagged surface, as opposed to the two that appeared in the canonical purification of two intervals. So we observe, immediately, that boundaries of the entanglement wedge cross-section are not necessarily in one-to-one correspondence with jagged surfaces. Second, in this case, the jagged surface is compact --- and, indeed, one can tweak the moduli of the three-boundary wormhole so that the area of the jagged surface is arbitrarily small.

Still, the very existence of this surface, no matter how small it might be, requires the presence of some boundary entanglement contributing to $I(A:B^* | B).$ We can therefore interpret inequality \eqref{eq:big-technical-claim}, which asserts that at least in certain AdS$_3$ states every boundary of the entanglement wedge cross-section gives a finite contribution to the Markov gap, as suggesting that the geometric connection through a \emph{corner} where a jagged surface intersects an entanglement wedge cross-section must be supported by a chunk of boundary entanglement that is, at least in some coarse sense, quantified by the lower bound on the Markov gap.

We emphasize that thus far, everything we have discussed in this subsection is entirely heuristic. The rest of the paper is dedicated to presenting concrete, technical results that elaborate and reinforce the ideas presented in this subsection. The first of these, presented in the following section, is a proof of the inequality \eqref{eq:big-technical-claim} whose significance was explained in the introduction. We will also take up the question of Markov recovery maps again explicitly in section \ref{sec:recovery-models}, where we construct some models of Markov recovery processes using fixed area states and compute the relevant fidelities using the gravitational path integral.

%%%%%%%%%%%%%%%%%%%%%%%
\section{Lower bounds in pure AdS$_3$ gravity}
\label{sec:geometric-proof}

In holographic theories of quantum gravity, the entanglement entropies of time-symmetric states with no bulk matter can be computed using the Ryu-Takayanagi formula.\footnote{This was argued in great generality by Lewkowycz and Maldacena in \cite{LM}. The main weakness of the argument is that it assumes the gravitational path integrals that compute the integer R\'{e}nyi entropies are dominated by a single family of replica-symmetric saddles. Luckily, this is only expected to fail at ``phase transitions'' where multiple families of saddles (and therefore multiple minimal surfaces) vie for dominance. See \cite{competing-saddles-1, competing-saddles-2} for some calculations in this regime.} The reflected entropies of such states can be computed using the Dutta-Faulkner formula \eqref{eq:dutta-faulkner}. The Markov gap $S_R - I$ is then computable as an area difference between locally minimal surfaces that satisfy certain  global homology constraints, with all of the surfaces lying in a single spacelike geometry.\footnote{In the introduction, we referred to these surfaces as ``KRT surfaces'' and ``RT surfaces'', a terminology that we will take up again in subsection \ref{subsec:two-special-cases}.} Obtaining universal lower bounds for the Markov gap, like the one we claimed in the introduction (inequality \eqref{eq:big-technical-claim}), is then reduced to an exercise in Riemannian geometry.

The problem we have set for ourselves becomes even simpler if we restrict our attention to AdS gravity in three bulk dimensions. Time-symmetric slices of asymptotically AdS$_3$ spacetimes are hyperbolic $2$-manifolds, about which there is an extensive mathematical literature. The minimal surfaces whose areas come into the calculation of $S_R - I$ are all geodesics --- or geodesic segments --- that lie on these hyperbolic $2$-manifolds. It is this considerable simplification that will allow us to prove the universal lower bound claimed in the introduction:\footnote{We remind the reader that the ``\# of cross-section boundaries'' appearing in this inequality is the number of endpoints of the cross section.}
\begin{equation} \label{eq:AdS3-bound-2}
    S_R(A:B) - I(A:B) \geq \frac{\log(2) \ell_{\text{AdS}}}{2 G_N} \times (\text{\# of cross-section boundaries}) + o\left(\frac{1}{G_N}\right).
\end{equation}

In subsection \ref{subsec:pentagons}, we introduce a fundamental trigonometric identity relating the lengths of the sides of a right-angled hyperbolic pentagon. In subsection \ref{subsec:two-special-cases}, we show how the pentagon identity can be used to prove inequality \eqref{eq:AdS3-bound-2} in two special cases: (i) $A$ and $B$ are equal-time intervals in the AdS$_3$ vacuum, and (ii) $A$ and $B$ are asymptotic boundaries of a genus-zero, three-boundary wormhole. In subsection \ref{subsec:general-proof}, we give the proof of inequality \eqref{eq:AdS3-bound-2} in complete generality for all time-symmetric states in pure AdS$_3$ gravity.

Generalizations for non-time-symmetric states, theories with matter, and higher dimensions are addressed in section \ref{sec:generalizations}.

%%%%%%%%%%%
\subsection{Right angled hyperbolic pentagons}
\label{subsec:pentagons}

We will take as our model of hyperbolic space the Poincar\'{e} disk, with metric
\begin{equation} \label{eq:poincare-disk}
    ds^2 = \frac{4}{(1-r^2)^2} \left( dr^2 + r^2 d\theta^2 \right)
\end{equation}
and coordinate ranges $r \in [0, 1), \theta \in [0, 2 \pi).$ In choosing this model, we have set the radius of curvature to one; factors of the radius of curvature in all of our expressions can be restored by dimensional analysis. A \emph{right-angled hyperbolic pentagon} --- see figure \ref{fig:right-angled-pentagon} for an example --- is a compact, convex set in the Poincar\'{e} disk bounded by five geodesic segments that meet sequentially at right angles. We will also use the term ``right-angled hyperbolic pentagon'' to refer to a pentagon on any Riemannian manifold that is isometric to a right-angled pentagon in the Poincar\'{e} disk.

\begin{figure}
    \centering
    \includegraphics{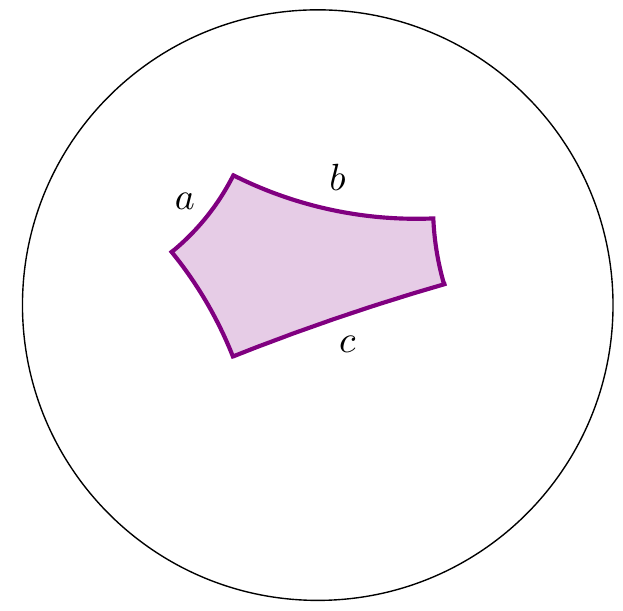}
    \caption{An example of a right-angled pentagon in the hyperbolic disk. Two adjacent sides have been labeled $a$ and $b$, and the unique non-adjacent side has been labeled $c$.}
    \label{fig:right-angled-pentagon}
\end{figure}

Right-angled hyperbolic pentagons satisfy a law of cosines relating their side lengths. If $a$ and $b$ are two adjacent sides, and $c$ is the unique non-adjacent side (again see figure \ref{fig:right-angled-pentagon}), then the side lengths satisfy the identity\footnote{There are many ways of proving this and other trigonometric identities on hyperbolic manifolds. Our favorite methods come from the field of \emph{hyperbolic line geometry}, in which geodesics in three-dimensional hyperbolic space are identified with the SL$_2(\mathbb{C})$ matrices that describe $\pi$-rotations about those geodesics, and geometric relations among geodesics are expressed as algebraic relations among the corresponding matrices. A very nice elaboration of this theory can be found in \cite{fenchel}. Our identity \eqref{eq:law-of-cosines} is given there as equation (2) on page 82, with a sign difference due to an orientation convention adopted therein.}
\begin{equation} \label{eq:law-of-cosines}
    \sinh(a) \sinh(b) = \cosh(c).
\end{equation}
Using this identity, we will compute a universal lower bound on the quantity $a + b - c$ for a right-angled hyperbolic pentagon; in the following subsections, this lower bound will be used to bound the Markov gap.

We may write $a + b - c$ using equation \eqref{eq:law-of-cosines} as
\begin{equation}
    a + b - c
        = a + b - \arccosh[\sinh(a) \sinh(b)].
\end{equation}
We'd like to minimize this quantity over all permitted values of $a$ and $b$. It will actually be conceptually easier to write the identity in terms of $u \equiv \sinh(a)$ and $v \equiv \sinh(b)$,
\begin{equation} \label{eq:apbmc-early}
    a + b - c = \arcsinh(u) + \arcsinh(v) - \arccosh(u v),
\end{equation}
and minimize over all permitted values of $u$ and $v$. Equation \eqref{eq:law-of-cosines} requires that the product $u v$ is in the image of $\cosh$, so we must have $u v \geq 1$, and the condition $a, b \geq 0$ restricts $u$ and $v$ to be individually nonnegative; this is the parameter space over which we will minimize equation \eqref{eq:apbmc-early}.

Computing the gradient of the right-hand side of \eqref{eq:apbmc-early}, it is straightforward to show that it vanishes only at $u, v = \pm i,$ which is not in the allowed parameter space. So the minimum of $a + b - c$ must be realized on the boundary of the parameter space, either on the curve $v = 1/u$ or in a limit as $u$ or $v$ becomes large. Setting $v = 1/u$, we have
\begin{equation}
    a + b - c = \arcsinh\left(\frac{1}{u}\right) + \arcsinh(u),
\end{equation}
which is minimized at $u=1$ with minimum value $2 \log(1 + \sqrt{2}).$ In the limit as one of the parameters becomes large --- say $u$ without loss of generality, since the expression is symmetric in $u$ and $v$ --- we have
\begin{equation}
    a + b - c = \arcsinh(v) - \log(v),
\end{equation}
which is minimized in the limit $v \rightarrow \infty$ with minimum $\log(2).$

As the two candidate minima are $2 \log(1 + \sqrt{2})$ and $\log(2)$, and $\log(2)$ is the smaller of the two, we conclude with the universal lower bound
\begin{equation} \label{eq:apbmc}
    a + b - c \geq \log(2).
\end{equation}

%%%%%%%%%%%
\subsection{Intervals and asymptotic regions}
\label{subsec:two-special-cases}

We will now use the pentagon inequality \eqref{eq:apbmc} to derive the Markov gap bound \eqref{eq:AdS3-bound-2} for two special cases, as a warmup for the general proof in subsection \ref{subsec:general-proof}.

Consider first the case where $A$ and $B$ are two equal-time intervals in the AdS$_3$ vacuum. We will assume that $A$ and $B$ are sufficiently large that their entanglement wedge is in the connected phase; otherwise, $S_R(A:B)$ and $I(A:B)$ both vanish at order $1/G_N$, there is no entanglement wedge cross-section and thus no cross-section boundary, and inequality \eqref{eq:AdS3-bound-2} is trivial. The static time slice is isometric to the Poincar\'{e} disk, and the minimal surfaces are straightforward to compute. We have sketched them in figure \ref{fig:minimal-surface-intervals}, where $\gamma_A$ and $\gamma_B$ are the Ryu-Takayanagi surfaces of $A$ and $B$, $\gamma_{AB}$ is the Ryu-Takayanagi surface of $AB$, and $\sigma_{A:B}$ is the entanglement wedge cross-section.

\begin{figure}
    \centering
    \makebox[\textwidth][c]{
	\subfloat[\label{fig:minimal-surface-intervals}]{
	    \includegraphics[scale=1]{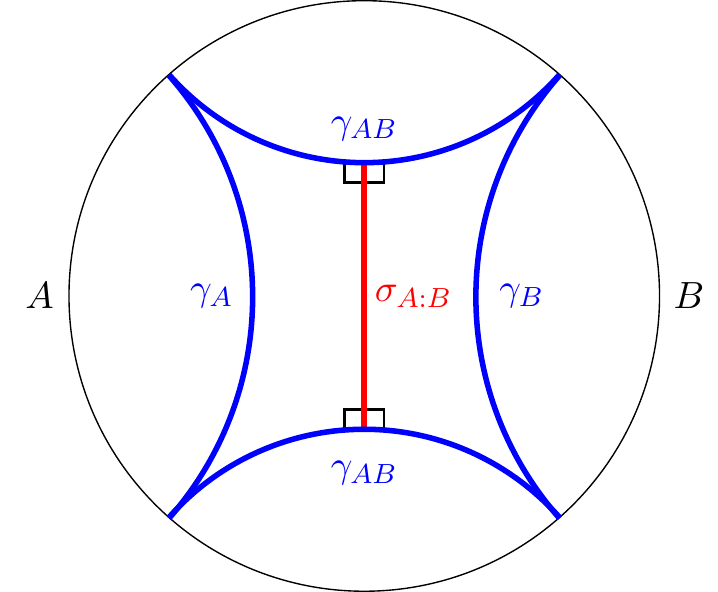}
	}
	\hspace{0.2em}
	\subfloat[\label{fig:krt-rt-intervals}]{
    	\includegraphics[scale=1]{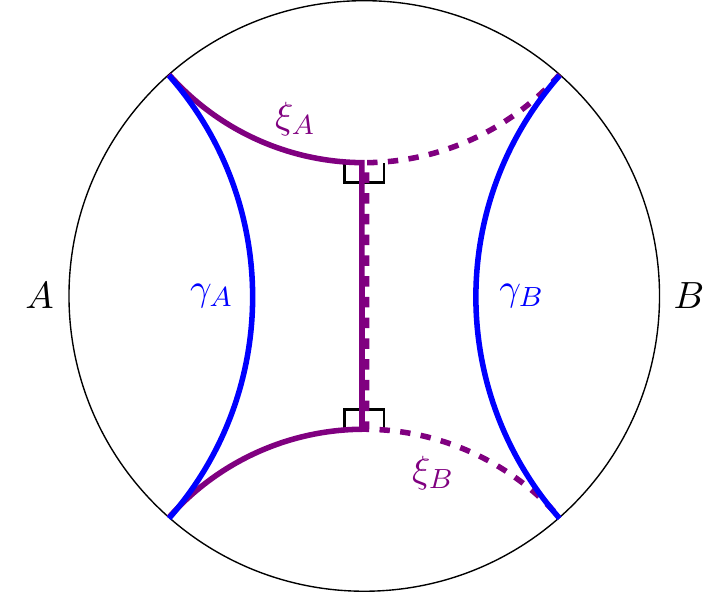}
	}
	}
    \caption{(a) The minimal surfaces relevant for computing $S_R - I$ for two equal-time intervals in the connected phase of the AdS$_3$ vacuum.  (b) The same configuration, but where the disjoint union $\gamma_{AB} \vee \sigma_{A:B} \vee \sigma_{A:B}$ has been redrawn in terms of two ``KRT'' surfaces $\xi_A$ and $\xi_B$. The surface $\xi_A$ has been drawn with a solid curve, while $\xi_B$ has been drawn with a dashed curve.}
\end{figure}

As discussed in the introduction, portions of $\gamma_{AB}$ and $\sigma_{A:B}$ can be combined to form candidate Ryu-Takayanagi surfaces for $A$ and $B$ that contain right-angled ``kinks.'' We call these surfaces KRT surfaces --- for ``kinked Ryu-Takayanagi surface'' --- and denote them by $\KRT(A)$ or by $\xi_{A}$ (respectively $\KRT(B)$ and $\xi_{B}$). These are sketched in figure \ref{fig:krt-rt-intervals}.

Using the Ryu-Takayanagi formula and the Dutta-Faulkner formula, the Markov gap can be written as
\begin{equation} \label{eq:FRT-m-RT-2-intervals}
    S_R(A:B) - I(A:B) = \frac{\area(\xi_{A}) - \area(\gamma_{A})}{4 G_N} + \frac{\area(\xi_{B}) - \area(\gamma_{B})}{4 G_N} + o\left(\frac{1}{G_N}\right).
\end{equation}
Since $\xi_A$ is homologous to $\gamma_A$, and $\gamma_A$ is minimal within its homology class by the assumptions of the Ryu-Takayanagi formula, each term on the right-hand side of \eqref{eq:FRT-m-RT-2-intervals} is individually nonnegative. We can obtain a tighter lower bound on each term by tiling the homology regions between the KRT surfaces and the RT surfaces with right-angled pentagons.

\begin{figure}
    \centering
    \includegraphics{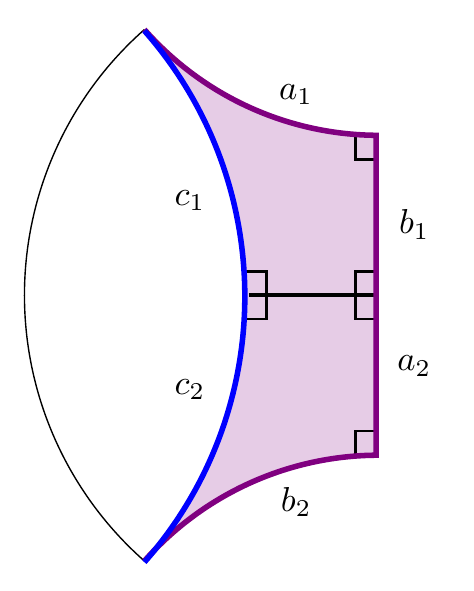}
    \caption{The homology region between the RT surface $\gamma_A$ and corresponding KRT surface $\xi_A$, previously depicted in figure \ref{fig:krt-rt-intervals}, has been tiled by two degenerate right-angled pentagons by drawing the minimal surface connecting $\gamma_A$ to the entanglement wedge cross-section. We have $a_j + b_j - c_j \geq \log(2)$ for $j=1,2.$}
    \label{fig:pentagons-intervals}
\end{figure}

In figure \ref{fig:pentagons-intervals}, we have sketched the homology region between $\xi_A$ and $\gamma_A$, and decomposed it into the union of two right-angled pentagons. This decomposition is constructed by drawing the minimal curve between $\gamma_A$ and $\sigma_{A:B}$, which meets those curves at right angles and subdivides the homology region into two tiles. It may not be obvious that the tiles in this decomposition are right-angled pentagons, as they appear to have only four sides. However, each tile has a \emph{degenerate} side ``at infinity''; the tiles are degenerate right-angled pentagons that can be obtained in a limit as one side of an ordinary right-angled pentagon shrinks and goes off to infinity to become a single vertex. The bound \eqref{eq:apbmc} still holds in this limit, where $a$ and $c$ (or $b$ and $c$) are the two sides that meet at infinity.\footnote{Even though the two sides of the pentagon that meet at infinity each have infinite length, the divergences in the lengths cancel and the difference in their areas is well-defined and independent of cutoff procedure. This is a special case of the general result proved in \cite{sorce-cutoffs}. In particular, choosing a cutoff procedure where we regulate the degenerate pentagon by taking it to be a limit of non-degenerate pentagons, one can show that degenerate right-angled pentagons satisfy inequality \eqref{eq:apbmc}.}

The difference in area between $\xi_A$ and $\gamma_A$ can be written as the sum over one $(a + b - c)$ term for each pentagon in the tiling. Using the labels of figure \ref{fig:pentagons-intervals}, the identity is
\begin{equation}
    \area(\xi_A) - \area(\gamma_A) = (a_1 + b_1 - c_1) + (a_2 + b_2 - c_2).
\end{equation}
Applying inequality \eqref{eq:apbmc} gives
\begin{equation}
    \area(\xi_A) - \area(\gamma_A) \geq 2 \log(2).
\end{equation}
Performing the same calculation for $(\area(\xi_B) - \area(\gamma_B))$ and plugging the lower bounds back into equation \eqref{eq:FRT-m-RT-2-intervals}, we obtain the inequality
\begin{equation}
    S_R(A:B) - I(A:B)
        \geq \frac{\log(2)}{G_N}.
\end{equation}
This matches our claimed general inequality \eqref{eq:AdS3-bound-2}, because for the case of two equal-time intervals in the AdS$_3$ vacuum, the entanglement wedge cross-section has two boundaries --- one where it meets each component of $\gamma_{AB}$. Note that this minimum is saturated only in the limit where the two intervals $A$ and $B$ become large; this follows from the discussion of subsection \ref{subsec:pentagons}.

As a second example, consider the case where $A$ and $B$ are two asymptotic boundaries of a genus-zero, three-boundary wormhole. We will assume that the horizon of the third boundary is sufficiently small that the entanglement wedge of $AB$ is in the connected phase; otherwise, as in the case of two intervals, there is no entanglement wedge cross-section and inequality \eqref{eq:AdS3-bound-2} is trivial. The minimal surfaces for this setup are sketched in figure \ref{fig:minimal-surfaces-wormhole} using the same notation as for the two-interval case in figure \ref{fig:minimal-surface-intervals}.

\begin{figure}
    \centering
    \includegraphics{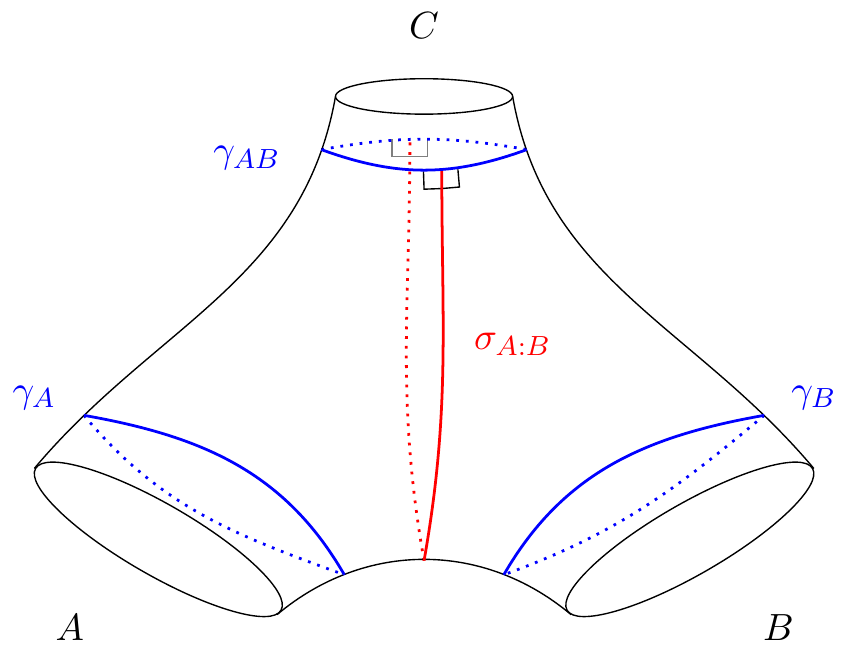}
    \caption{A three-boundary wormhole geometry with boundaries $A,$ $B$, and $C$ such that the state on $AB$ is in its connected phase. The $A$ and $B$ minimal surfaces are labeled $\gamma_{A}$ and $\gamma_{B}$, the $AB$ minimal surface is labeled $\gamma_{AB}$, and the entanglement wedge cross-section is labeled $\sigma_{A:B}.$}
    \label{fig:minimal-surfaces-wormhole}
\end{figure}

In this case, the surfaces $\gamma_{A}, \gamma_{B},$ and $\gamma_{AB}$ bound a ``pair of pants'' geometry, which is subdivided into two pieces by the cross-section $\sigma_{A:B}$. These two pieces are the homology regions bounded by the KRT and RT surfaces for $A$ and $B$; they are sketched in figure \ref{fig:krt-rt-wormhole}. We can further subdivide each of those pieces into two right-angled hyperbolic pentagons by drawing the minimal geodesics between $\gamma_{AB}$ and $\gamma_{A}$ (or $\gamma_{B}$) and between $\sigma_{A:B}$ and $\gamma_{A}$ (or $\gamma_{B}$). This pentagonal tiling of the homology regions is also sketched in figure \ref{fig:krt-rt-wormhole}.\footnote{The fact that these minimal geodesics exist is an elementary feature of the pair of pants geometry. In the general proof of subsection \ref{subsec:general-proof}, we will prove the existence of pentagonal tilings for a much more general class of homology regions; the reader wondering about the existence of pentagonal tilings even for the three-boundary wormhole is encouraged to look ahead to the general proof.} As in the two-interval case discussed previously, the quantity $(\area(\xi_{A}) - \area(\gamma_A))$ can be written as a sum of two $(a + b - c)$ terms for the right-angled pentagons tiling the homology region between $\KRT(A)$ and $\RT(A)$; similarly for the quantity $(\area(\xi_{B}) - \area(\gamma_B)).$ Even though the three-boundary wormhole is not globally isometric to the Poincar\'{e} disk, it is \emph{locally} isometric to the Poincar\'{e} disk, and so each pentagon in the tiling of figure \ref{fig:krt-rt-wormhole} is a right-angled hyperbolic pentagon. Once again applying the pentagon inequality \eqref{eq:apbmc}, we obtain the bound
\begin{equation}
    S_R(A:B) - I(A:B)
        \geq \frac{\log(2)}{G_N},
\end{equation}
which matches the general inequality \eqref{eq:AdS3-bound-2}, since the entanglement wedge cross-section has two boundaries (see again figure \ref{fig:minimal-surfaces-wormhole}).

\begin{figure}
    \centering
    \includegraphics{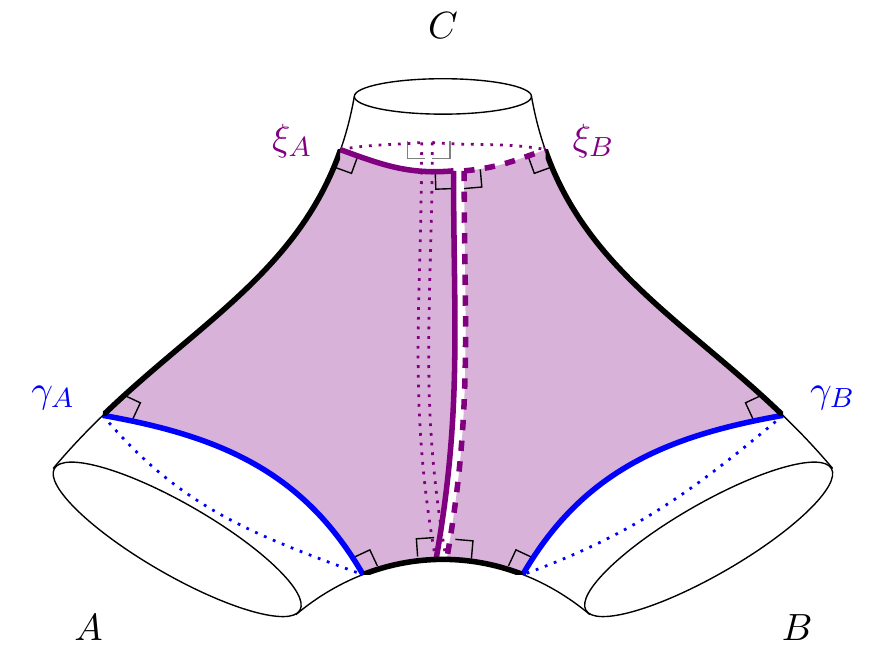}
    \caption{The KRT and RT surfaces for the three-boundary wormhole originally sketched in figure \ref{fig:minimal-surfaces-wormhole}. We have also sketched the tiling of each KRT-RT homology region by right-angled pentagons, though we have only shaded the pentagons appearing on the front-facing side of the wormhole for the sake of visual clarity.} 
    \label{fig:krt-rt-wormhole}
\end{figure}

The general proof, for arbitrary global states and arbitrary subregions $A$ and $B$ within those states, is presented in the following subsection. The general idea is to show that each component of a KRT surface can be homotopically deformed to a geodesic such that the homotopy region can be tiled with right-angled hyperbolic pentagons, with each pentagon in the tiling corresponding to a single kink in the KRT surface. The homotopy tilings will be constructed so that the pentagon inequality \eqref{eq:apbmc} guarantees that the area difference between a KRT component and its homotopic geodesic is lower-bounded by $\log(2)$ times the number of kinks in the KRT component. After every component of a KRT surface has been homotopically deformed to a geodesic, the resulting ``deformed'' surface is smooth and is still in the same homology class as the ``true'' RT surface, and thus has area no greater than the true RT surface. Combining these observations will give us
\begin{equation}
    \area(\KRT) - \area(\RT) \geq \log(2) \times (\text{\# of right-angled kinks}),
\end{equation}
which, upon relating the total number of KRT kinks to the total number of cross-section boundaries, will reproduce inequality \eqref{eq:AdS3-bound-2}.

This sketch will be explained in much more detail over the course of the proof.

%%%%%%%%%%%
\subsection{The general proof}
\label{subsec:general-proof}

The entanglement wedge cross-section was originally defined in \cite{EOP1,EOP2} as a codimension-$2$ surface that splits the entanglement wedge into two homology regions, such that one homology region belongs to each of the two boundary regions $A$ and $B$. It will be useful, for our proof, to reframe this as a statement about minimal surfaces. Formally, we say that the entanglement wedge cross-section $\sigma_{A:B}$ induces a bipartition of the minimal surface $\gamma_{AB}$ into two pieces --- call them $\gamma_{AB}^{(A)}$ and $\gamma_{AB}^{(B)}$ --- such that the following homology equivalences hold:
\begin{align}
    \gamma_A & \sim_{\text{homology}} \gamma_{AB}^{(A)} \cup \sigma_{A:B}, \\
    \gamma_B & \sim_{\text{homology}} \gamma_{AB}^{(B)} \cup \sigma_{A:B}.
\end{align}
The surfaces on the right-hand sides of these equations are what we have been calling the ``KRT surfaces'' $\xi_A$ and $\xi_B$.

On a general hyperbolic $2$-manifold, for general boundary subregions $A$ and $B$, what can we say about the KRT surfaces for $A$ and $B$? First, trivially, each KRT surface is homologous to the corresponding RT surface. Second, each KRT surface is piecewise smooth, with each smooth portion of the surface being a geodesic segment. Third, the only non-smoothness in a KRT surface comes in the form of right-angled kinks. This last claim follows from the fact that $\sigma_{A:B}$ and $\gamma_{AB}$ are individually smooth, so the only non-smoothness in a KRT surface comes from an intersection between $\sigma_{A:B}$ and $\gamma_{AB}$; such intersections must be orthogonal, by the assumption that $\sigma_{A:B}$ has minimal length.\footnote{If the entanglement wedge cross-section did not meet $\gamma_{AB}$ at a right angle, then it would be possible to deform the cross-section to be ``more orthogonal,'' making its area smaller.} Crucially, the right angles of these kinks always face ``inward'' toward the homology region; i.e., the homology region between a KRT surface and its corresponding RT surface is convex. Finally, the KRT surface is simple, i.e., it does not contain any self-intersections aside from the right-angle intersections of its geodesic segments.

Our goal is to use the properties listed in the preceding paragraph to prove that the area difference between a KRT surface and the corresponding RT surface is lower-bounded by $\log(2)$ times the number of kinks in the KRT surface. Because the KRT surface is non-self-intersecting, each of its components must be either (a) a kinked geodesic with two endpoints at infinity (as in figure~\ref{fig:pentagons-intervals}), (b) a kinked closed geodesic (as in figure~\ref{fig:krt-rt-wormhole}), (c) a smooth geodesic with two endpoints at infinity, or (d) a smooth closed geodesic. We label these four types of surfaces ``KA'' (for ``kinked asymptotic''), ``KL'' (for ``kinked loop''), ``A'' (for ``asymptotic''), and ``L'' (for ``loop''). The component decomposition of a KRT surface can be written schematically as
\begin{equation}
    \KRT = \cup_{\alpha} \KA_{\alpha} \cup_{\beta} \KL_{\beta} \cup_{j} \A_{j} \cup_{k} \L_{k}.
\end{equation}

We will soon show that each kinked asymptotic geodesic $\KA_{\alpha}$ is homotopic to a smooth asymptotic geodesic $\A_{\alpha}$ with the same endpoints whose area is smaller by at least $\log(2)$ times the number of kinks in $\KA_{\alpha}$; this area difference will follow from a pentagonal tiling of the homotopy region, where each kink in $\KA_{\alpha}$ becomes a vertex of a right-angled pentagon, and the number of pentagons is equal to the number of kinks. We will show, similarly, that each kinked geodesic loop $\KL_{\beta}$ is homotopic to a smooth geodesic loop $\L_{\beta}$ with area smaller by at least $\log(2)$ times the number of kinks in $\KL_{\beta}.$ This will give us
\begin{equation}
    \KRT \sim_{\text{homotopy}} \cup_{\alpha} \A_{\alpha} \cup_{\beta} \L_{\beta} \cup_{j} \A_{j} \cup_{k} \L_{k}.
\end{equation}
But because homotopy is stronger than homology, this will imply
\begin{equation}
    \KRT \sim_{\text{homology}} \cup_{\alpha} \A_{\alpha} \cup_{\beta} \L_{\beta} \cup_{j} \A_{j} \cup_{k} \L_{k}.
\end{equation}
Because homology is an equivalence relation, and because we have $\KRT \sim_{\text{homology}} \RT$, we obtain
\begin{equation}
    \RT \sim_{\text{homology}} \cup_{\alpha} \A_{\alpha} \cup_{\beta} \L_{\beta} \cup_{j} \A_{j} \cup_{k} \L_{k}.
\end{equation}
The fact that the RT surface is minimal in its homology class gives
\begin{equation}
    \area(\RT) \leq \area\left[ \cup_{\alpha} \A_{\alpha} \cup_{\beta} \L_{\beta} \cup_{j} \A_{j} \cup_{k} \L_{k} \right]
    \leq \area(\KRT) - \log(2) \times (\text{total \# of kinks}).
\end{equation}
From this formula immediately follows the universal bound \eqref{eq:AdS3-bound-2}, because the total number of kinks in each KRT surface ($\KRT(A)$ or $\KRT(B)$) is equal to the number of boundaries of the entanglement wedge cross-section. We therefore have
\begin{align}
    S_R(A:B) - I(A:B)
        & = \frac{\KRT(A) - \RT(A)}{4 G_N} + \frac{\KRT(B) - \RT(B)}{4 G_N} + o\left(\frac{1}{G_N}\right) \nonumber \\
        & \geq \frac{\log(2) \times (\text{\# of $A$ kinks})}{4 G_N} + \frac{\log(2) \times (\text{\# of $B$ kinks})}{4 G_N} + o\left(\frac{1}{G_N}\right) \nonumber \\
        & = \frac{2 \log(2) \times (\text{\# of cross-section boundaries})}{4 G_N} + o\left(\frac{1}{G_N}\right),
\end{align}
which is exactly the bound claimed in \eqref{eq:AdS3-bound-2}.

We now complete the proof of inequality \eqref{eq:AdS3-bound-2} by proving three lemmas. In the first lemma, we prove that any kinked asymptotic geodesic in the Poincar\'{e} disk has the aforementioned ``homotopy shrinking property,'' i.e., that it is homotopic to a smooth asymptotic geodesic whose area is smaller by at least $\log(2)$ times the number of kinks. In the second and third lemmas, we use the fact that every hyperbolic $2$-manifold is universally covered by the Poincar\'{e} disk to prove the analogous statements for kinked asymptotic geodesics and kinked geodesic loops in arbitrary hyperbolic $2$-manifolds.

\begin{lemma}
    \label{lem:poincare-disk}
    Let $\tilde{\KA}$ be a non-self-intersecting, piecewise-geodesic curve in the Poincar\'{e} disk with the following properties: (see figure \ref{fig:krt-poincare-disk} for a sketch)
        \begin{enumerate}[(i)]
            \item Each geodesic segment meets the next segment at a right angle.
            \item All right angles ``point in the same direction'' in the sense that an ant crawling along $\tilde{\KA}$ will have to turn in the same direction each time it encounters a right angle.
            \item $\tilde{\KA}$ has well-defined endpoints at infinity in the sense that there is a smooth asymptotic geodesic $\tilde{\A}(t)$ and a parametrization $t$ of $\tilde{\KA}$ such that the hyperbolic distance $|\tilde{\A}(t) - \tilde{\KA}(t)|$ remains bounded in the limits $t \rightarrow \pm \infty.$
        \end{enumerate}
        
        Then:
        \begin{enumerate}[(1)]
            \item The geodesic $\tilde{\A}$ does not intersect $\tilde{\KA}$, and the homotopy region between the two is convex.
            \item The homotopy region between $\tilde{\KA}$ and $\tilde{\A}$ can be tiled by hyperbolic right-angled pentagons, with one pentagon for each kink in $\tilde{\KA}$.
            \item The area difference between $\tilde{\KA}$ and $\tilde{\A}$ is lower-bounded by $\log(2)$ times the number of kinks in $\tilde{\KA}.$
        \end{enumerate}
\end{lemma}

\begin{figure}
    \centering
    \makebox[\textwidth][c]{
	\subfloat[\label{fig:krt-poincare-disk}]{
	    \includegraphics{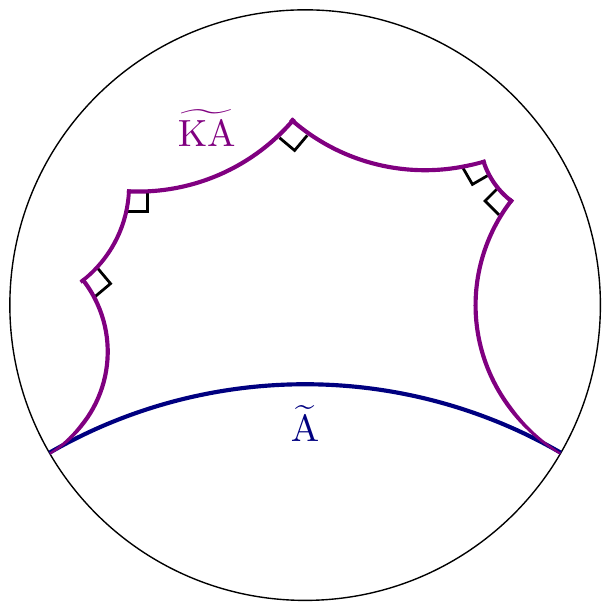}
	}
	\hspace{0.4em}
	\subfloat[\label{fig:phi-map-poincare-disk}]{
    	\includegraphics{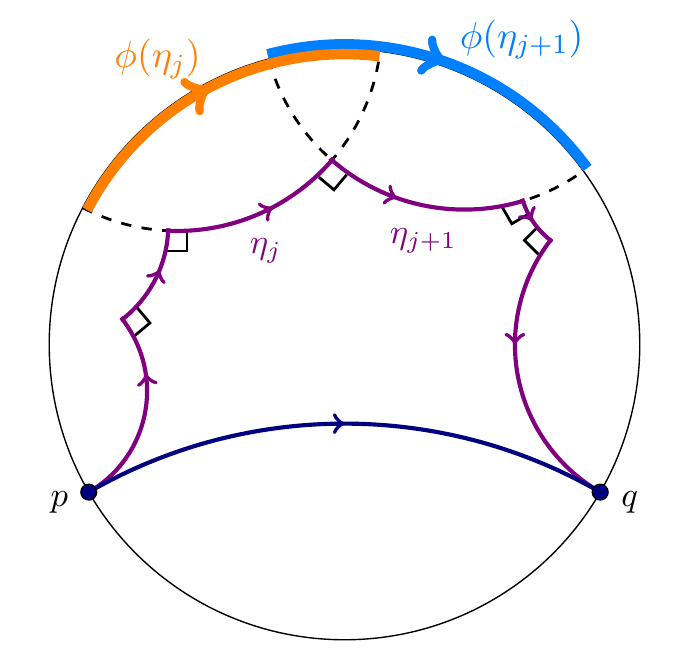}
	}
	}
    \caption{(a) $\tilde{\KA}$ is a piecewise-geodesic curve in the Poincar\'{e} disk satisfying properties (i)-(iii) of lemma \ref{lem:poincare-disk}. $\tilde{\A}$ is the unique smooth geodesic that has the same boundary endpoints as $\tilde{\KA}.$ The curve $\tilde{\KA}$ we have drawn here has finitely many geodesic segments, but the number of geodesic segments could in principle be infinite. (b) A choice of orientation on $\tilde{\KA}$, together with a visual representation of the map $\phi$ from geodesic segments to boundary segments that is defined in the proof of lemma \ref{lem:poincare-disk}. The boundary points of $\tilde{\A}$ have been labeled $p$ and $q$.}
\end{figure}

\begin{proof}

    We will break the proof up in terms of claims (1)-(3) of the lemma statement.
    
    \begin{enumerate}[(1)]
        \item First, let's assign $\tilde{\KA}$ an orientation so that an ant crawling along $\tilde{\KA}$ will have to turn right each time it encounters a kink. This orientation is sketched in figure \ref{fig:phi-map-poincare-disk}. We define a map $\phi$ from geodesic segments of $\tilde{\KA}$ to oriented segments of the boundary of the Poincar\'{e} disk as follows. Each geodesic segment can be extended to a full, asymptotic geodesic that has a ``future'' and ``past'' endpoint with respect to the chosen orientation of $\tilde{\KA}.$ If $\eta$ labels a geodesic segment on $\tilde{\KA}$, then $\phi(\eta)$ will be the clockwise-oriented boundary segment that starts at the ``past'' endpoint of the extension of $\eta$ and ends at the ``future'' endpoint of the extension of $\eta$. This map is also sketched in figure \ref{fig:phi-map-poincare-disk}.
        
        The geodesic segments of $\tilde{\KA}$ are naturally ordered with respect to its orientation; let's call the ordered, possibly finite sequence of segments $\{\eta_j\}.$ The map $\phi$ sends this sequence to an ordered sequence of clockwise-oriented boundary segments $\{\phi(\eta_j)\}.$ Any two consecutive segments $\phi(\eta_j)$ and $\phi(\eta_{j+1})$ will have an intersection bounded by the clockwise endpoint of $\phi(\eta_j)$ and the counterclockwise endpoint of $\phi(\eta_{j+1})$ --- again consult figure \ref{fig:phi-map-poincare-disk}.
        
        Because $\tilde{\KA}$ has a well-defined future endpoint $q$ --- the future endpoint of the smooth geodesic $\tilde{\A}(t)$ --- the clockwise endpoints of the segments $\phi(\eta_j)$ must tend to $q$ in the limit $j \rightarrow \infty$.\footnote{If $\tilde{\KA}$ only has finitely many segments, then the clockwise endpoint of the last $\phi(\eta_j)$ must be $q$.} Similarly, because $\tilde{\KA}$ has a well-defined past endpoint $p$, the counterclockwise endpoints of the segments $\phi(\eta_j)$ must tend to $p$ in the limit $j \rightarrow -\infty$. From this fact, and the fact that $\tilde{\KA}$ is non-self-intersecting, we may conclude that each boundary segment $\phi(\eta_j)$ lies entirely in the clockwise-oriented boundary segment that starts at $p$ and ends at $q$. Because the boundary segments $\phi(\eta_j)$ translate clockwise as $j$ increases, if the clockwise endpoint of $\phi(\eta_j)$ ever \emph{passes} the endpoint $q$, then the segments will need to do a full loop around the boundary to be consistent with the $j\rightarrow \infty$ limit; such a loop would necessarily cause $\tilde{\KA}$ to self-intersect in the bulk. An analogous line of reasoning holds for the counterclockwise endpoints. 
        
        The fact that the boundary segments $\phi(\eta_j)$ all lie within the clockwise boundary segment $p \rightarrow q$ is equivalent to claim (1), that $\tilde{\KA}$ does not intersect $\tilde{\A}$ --- since two geodesics in the Poincar\'{e} disk only intersect if their boundary points ``alternate'' --- and that the homotopy region is convex --- because each right angle must be oriented ``inwards'' toward the geodesic $\tilde{\A}$, as in figure \ref{fig:krt-poincare-disk}.
        
        \item 
        For any two non-intersecting asymptotic geodesics in the Poincar\'{e} disk, there is a unique third geodesic that intersects each of the first two orthogonally. If the original two geodesics have a common endpoint at the boundary, then the mutually intersecting geodesic is just the degenerate point at infinity.
        
        Let $\eta_{j}$ be one of the geodesic segments of $\tilde{\KA}$. Let $\rho_j$ be the unique geodesic segment that connects the full geodesic extension of $\eta_j$ to the full geodesic $\tilde{\A}$ orthogonally. We will show that the intersection of $\rho_j$ with the geodesic extension of $\eta_j$ must lie within the segment $\eta_j$ itself.
        
        \begin{figure}
        \centering
            \includegraphics{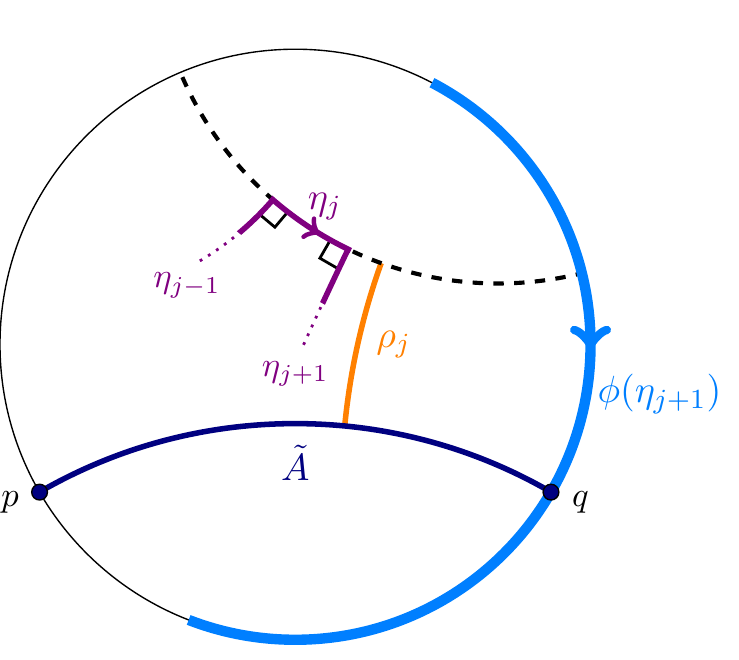}
            \caption{Here $\tilde{\A}$ is the smooth geodesic connecting the endpoints of $\tilde{\KA}$, and $\eta_{j-1},$ $\eta_{j},$ $\eta_{j+1}$ are segments of $\tilde{\KA}.$ The curve $\rho_j$ is the unique geodesic segment orthogonally intersecting both $\tilde{\A}$ and the geodesic extension of $\eta_j$. If we assume that $\rho_j$ is in the ``future'' of the segment $\eta_j$, as sketched here, then $\phi(\eta_{j+1})$ is not contained within the clockwise boundary segment connecting $p$ and $q$; this contradicts part (1) of the proof of lemma \ref{lem:poincare-disk}. An analogous contradiction arises if $\rho_j$ is in the ``past'' of $\eta_j$, which proves that $\rho_j$ intersects the segment $\eta_j$.}
            \label{fig:connecting-lines-poincare-disk}
        \end{figure}
        
        Suppose, toward contradiction, that $\rho_j$ intersects the geodesic extension of $\eta_j$ in a part of the geodesic extension that does not lie in the segment $\eta_j$. Suppose further, without loss of generality, that the point where $\rho_j$ intersects the geodesic extension lies in the ``future'' of the segment $\eta_j$ with respect to the orientation chosen in part (1) of this proof. This is sketched in figure \ref{fig:connecting-lines-poincare-disk}.
        
        The segment $\eta_{j+1}$, since it stems off of $\eta_j$ \emph{before} reaching $\rho_j$, must have the property that the boundary segment $\phi(\eta_{j+1})$ extends beyond the boundary segment lying clockwise between $p$ and $q$. See figure \ref{fig:connecting-lines-poincare-disk} for a sketch. This contradicts part (1) of this proof; we conclude that $\rho_j$ must intersect $\eta_j$.
        
        Tiling the homotopy region between $\tilde{\KA}$ and $\tilde{\A}$ is now simple. We draw all of the geodesic segments $\{\rho_j\}$ that lie orthogonally between $\tilde{\A}$ and the $\tilde{\KA}$ segments $\{\eta_j\}$; the tiles divided by these geodesics, sketched in figure \ref{fig:pentagon-tiles-poincare-disk}, are all hyperbolic right-angled pentagons, possibly with the last tiles in the sequence being degenerate right-angled pentagons if $\{\eta_j\}$ is a finite or half-finite sequence.
        
        \begin{figure}
            \centering
            \includegraphics[scale=1.3]{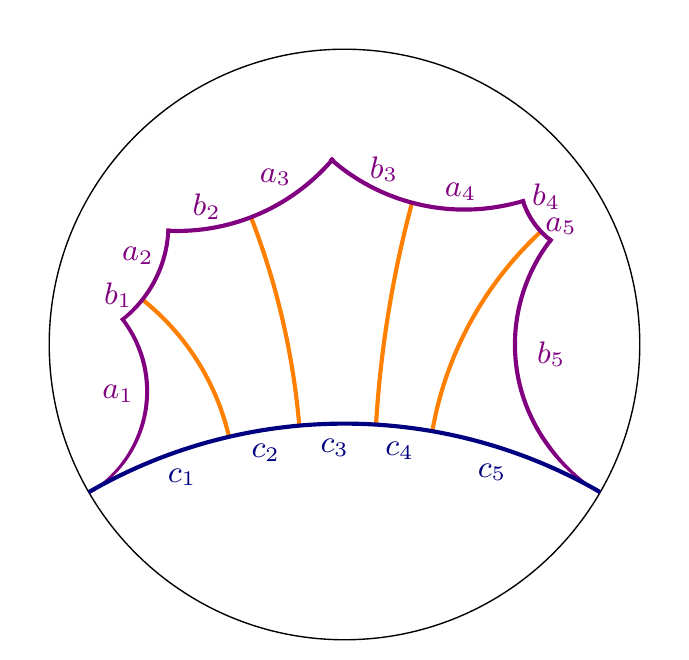}
            \caption{A pentagonal tiling of the homotopy region between the two curves originally sketched in figure \ref{fig:krt-poincare-disk}, formed by drawing the orthogonal geodesic segments joining each segment of $\tilde{\KA}$ to $\tilde{\A}$. Every intersection in this figure is right-angled, except for the two intersections at infinity.}
            \label{fig:pentagon-tiles-poincare-disk}
        \end{figure}
        
        \item 
        Each pentagon in the tiling we have just constructed has, as one of its vertices, a kink of $\tilde{\KA}$. The side of the pentagon opposite that vertex is a segment of the geodesic $\tilde{\A}$. If we label the two sides of the pentagon adjacent to the kink $a_j$ and $b_{j}$, and label the opposite side $c_j$ --- see again figure \ref{fig:pentagon-tiles-poincare-disk} --- then we have
        \begin{equation}
            \area(\tilde{\KA}) - \area(\tilde{\A})
                = \sum_{j} (a_j + b_j - c_j).
        \end{equation}
        Applying inequality \eqref{eq:apbmc}, we obtain the desired bound
    \begin{equation}
            \area(\tilde{\KA}) - \area(\tilde{\A})
                \geq \log(2) \times (\text{\# of kinks}).
        \end{equation}
        
    \end{enumerate}
\end{proof}

\begin{lemma}
    \label{lem:homotopy-shrinking-asymptotic}
    Let $\Sigma$ be a complete hyperbolic $2$-manifold without cusps, and let $\KA$ be a non-self-intersecting, piecewise-geodesic curve on $\Sigma$ satisfying properties (i)-(iii) of lemma \ref{lem:poincare-disk}. To condition (iii) we add the requirement that the geodesic defining the endpoints of $\KA$ does not have a limit cycle, i.e., both ends of $\KA$ genuinely go off ``to infinity.''
    
    Then:
    \begin{enumerate}[(1)]
        \item $\KA$ is homotopic to a geodesic $\A$ with the same asymptotic endpoints as $\KA$.
        \item The homotopy region can be tiled by right-angled hyperbolic pentagons, with each kink in $\KA$ being a vertex of its own pentagon.
        \item The area of $\KA$ exceeds the area of $\A$ by at least $\log(2)$ times the number of kinks in $\KA$.
    \end{enumerate}
\end{lemma}

\begin{proof}
    It is a fundamental theorem in hyperbolic geometry that every complete hyperbolic $2$-manifold is universally covered by the Poincar\'{e} disk. Let $\Pi : \mathbb{H}_{2} \rightarrow \Sigma$ be a covering map, fix a point $p$ on $\KA$, and choose a point $\tilde{p}$ in the fiber $\Pi^{-1}(p)$. It is a basic fact in the theory of covering spaces that there is a unique lift of $\KA$ to the Poincar\'{e} disk that passes through $\tilde{p}$. This lift can be constructed by pulling a small neighborhood of $p$ in $\KA$ back to a small neighborhood of $\tilde{p}$ using $\Pi^{-1}$, after which the lift extends uniquely away from the pulled-back segment. The fact that no other lifts of $\KA$ can pass through $\tilde{p}$ follows from the fact that $\KA$ is non-self-intersecting, together with the fact that $\Pi$ is a local homeomorphism.
    
    Let us denote this lift by $\tilde{\KA}.$ By assumption, $\KA$ satisfies conditions (i)-(iii) of lemma \ref{lem:poincare-disk}, and all three of these conditions are preserved under the lift.\footnote{\label{footnote:multi-geodesic-caveat}There is one caveat here worth emphasizing: while condition (iii) being satisfied for $\KA$ guarantees that condition (iii) is satisfied for $\tilde{\KA}$ --- one can show this by defining the endpoints of $\tilde{\KA}$ using local lifts of the $\KA$ ``endpoint-defining geodesic'' in a neighborhood of the asymptotic boundary --- it will generally not be the case that the geodesic defining the endpoints of $\tilde{\KA}$ is a lift of the geodesic we have chosen arbitrarily to define the endpoints of $\KA$. There is only one geodesic in $\mathbb{H}_{2}$ with the same endpoints as $\tilde{\KA}$, while there are many geodesics in $\mathbb{H}_{2}$ whose projections to $\Sigma$ have the same endpoints as $\KA$. This is because geodesics ending at points that are \textit{images} of the $\KA$ endpoints under the quotient will map to geodesics with the same endpoints as $\tilde{\KA}$.} So $\tilde{\KA}$ satisfies conditions (i)-(iii) of lemma \ref{lem:poincare-disk}, and we may conclude that there is a geodesic $\tilde{\A}$ in $\mathbb{H}_{2}$ with the same endpoints as $\tilde{\KA}$, that is homotopic to $\tilde{\KA}$ with the homotopy region tiled by right-angled hyperbolic pentagons, such that the area of $\tilde{\KA}$ exceeds the area of $\tilde{\A}$ by $\log(2)$ times the number of kinks in $\tilde{\KA}$. Homotopies in the universal cover $\mathbb{H}_{2}$ are preserved under projection down to $\Sigma$, so $\KA$ is homotopic to $\A \equiv \Pi(\tilde{\A})$.
    
    We now show that the length of $\KA$ on $\Sigma$ is the same as the length of $\tilde{\KA}$ on $\mathbb{H}_{2}$. The covering map $\Pi$ is a local isometry, so this could only fail to be true if $\Pi$ were non-injective on $\tilde{\KA}$. But if $\Pi$ were non-injective on $\tilde{\KA}$, then the segment of $\tilde{\KA}$ lying between any two points with the same $\Pi$-image would be mapped to a loop in $\KA$; since $\KA$ was assumed to be non-self-intersecting, this cannot be the case and thus we have $\area(\KA) = \area(\tilde{\KA}).$
    
    One can also show that $\Pi$ must be injective on $\tilde{\A}$. Every non-compact geodesic on a cuspless hyperbolic $2$-manifold is non-self-intersecting.\footnote{If a non-compact geodesic had a nontrivial loop, then the element of the fundamental group corresponding to that loop would induce an infinite-order isometry of the covering space preserving the lift of the geodesic; the quotient of the lift by the group of covering transformations, which ought to be isomorphic to the original geodesic, would then have to be compact. More details on this ``fundamental group $\leftrightarrow$ covering transformation'' correspondence are given in the proof of lemma \ref{lem:homotopy-shrinking-loop}.} If $\A$ were compact, then it could not be homotopic to the asymptotic geodesic $\KA$. So $\A$ must be non-self-intersecting, which implies that $\Pi$ is injective on $\tilde{\A}$, which implies $\area(\A) = \area(\tilde{\A})$. Putting this together gives
    \begin{equation}
        \area(\KA) - \area(\A) = \area(\tilde{\KA}) - \area(\tilde{\A}) \geq \log(2) \times \text{\# kinks},
    \end{equation}
    so we may conclude that claims (1)-(3) in the lemma statement hold with $\A \equiv \Pi(\tilde{\A}).$\footnote{In the statement of the lemma, we assumed the existence of a geodesic on $\Sigma$ with the same asymptotic endpoints as $\KA$ (condition iii). This was just a formal way of saying that $\KA$ has well-defined endpoints at infinity. As we emphasized in footnote \ref{footnote:multi-geodesic-caveat}, that geodesic has no relation to $\A$, which is special in that it not only has the same boundary endpoints as $\KA$, but is homotopic to $\KA$ as well.}
\end{proof}

\begin{lemma}
    \label{lem:homotopy-shrinking-loop}
    Let $\Sigma$ be a complete hyperbolic $2$-manifold without cusps, and let $\KL$ be a non-self-intersecting, piecewise-geodesic loop on $\Sigma$ satisfying conditions (i) and (ii) of lemma \ref{lem:poincare-disk}.
    
    Then:
    \begin{enumerate}[(1)]
        \item $\KL$ is homotopic to a unique geodesic loop $\L$.
        \item The homotopy region can be tiled by right-angled hyperbolic pentagons, with each kink in $\KL$ being a vertex of its own pentagon.
        \item The area of $\KL$ exceeds the area of $\L$ by at least $\log(2)$ times the number of kinks in $\KL$.
    \end{enumerate}
\end{lemma}
\begin{proof}
    Claim (1) in this lemma is a special case of a classic theorem in hyperbolic geometry, which states that \emph{every} homotopy class of non-self-intersecting loops on a hyperbolic manifold has a unique geodesic representative. Our proof of the present lemma works by following the proof of that theorem --- specifically the proof presented in section 1.2 of \cite{farb2011primer} --- and adding details as needed to prove claims (2) and (3).
    
    As in the previous lemma, let $\Pi : \mathbb{H}_{2} \rightarrow \Sigma$ be a covering of $\Sigma$ by the Poincar\'{e} disk. Fix a point $p$ in the loop $\KL$, and let $\tilde{p}$ be a point in the fiber $\Pi^{-1}(p).$ There is a unique lift of $\KL$ that passes through $\tilde{p}$, which we call $\tilde{\KL}$. Because conditions (i) and (ii) of lemma \ref{lem:poincare-disk} are preserved under lifting, $\tilde{\KL}$ is a kinked geodesic in the hyperbolic disk. Because $\KL$ is non-self-intersecting, so is $\tilde{\KL}.$ The lifting procedure is sketched in figure \ref{fig:KL-lift}.
    
    \begin{figure}
        \centering
        \includegraphics{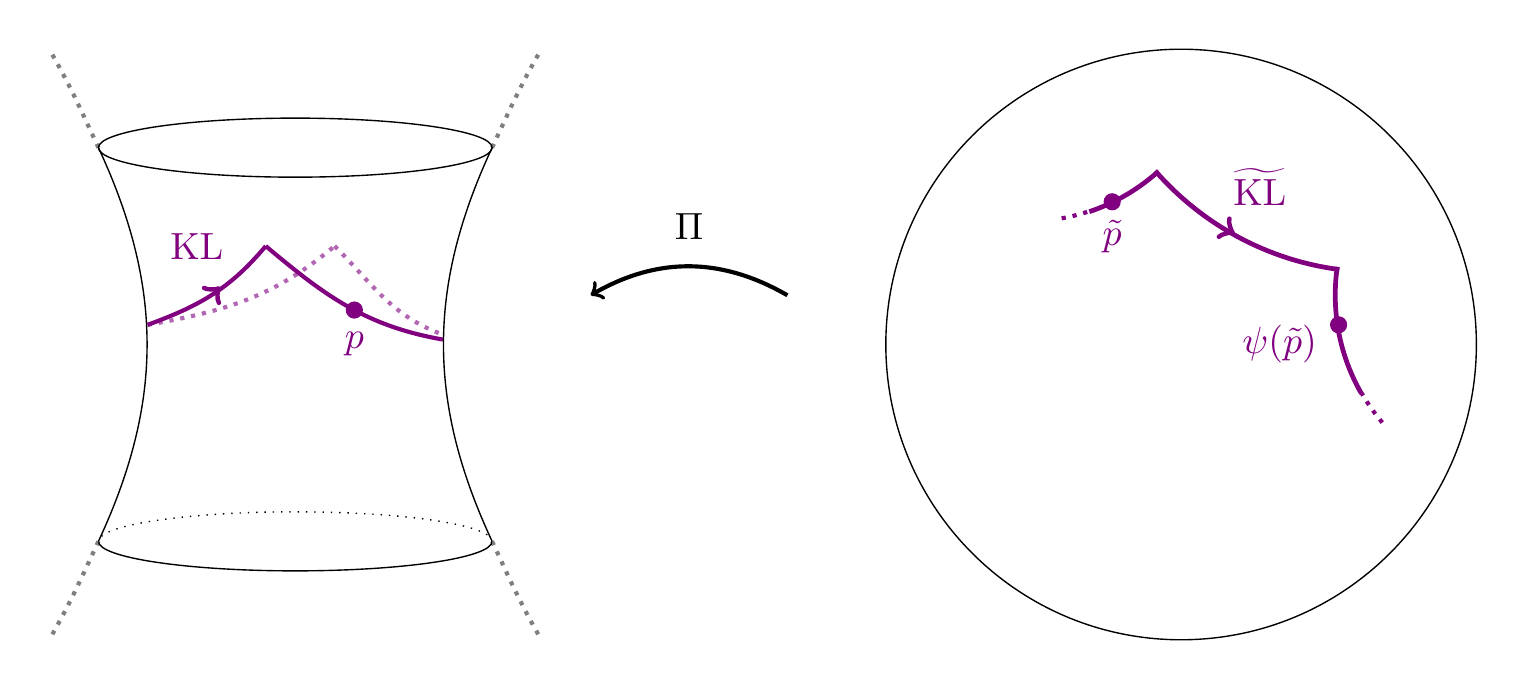}
        \caption{The left-hand side of this figure shows a piece of a hyperbolic $2$-manifold containing a piecewise-geodesic loop $\KL$ satisfying the conditions of lemma \ref{lem:homotopy-shrinking-loop}. The right-hand side of the figure shows a portion of the unique lift $\tilde{\KL}$ passing through the point $\tilde{p}$ in the universal cover. We have also indicated the action of the covering isometry $\psi$ induced by the loop $\KL$, which maps $\tilde{\KL}$ to itself.}
        \label{fig:KL-lift}
    \end{figure}
    
    There is a classic correspondence in algebraic topology between elements of the fundamental group of a space and automorphisms of its universal cover. An automorphism of a universal cover $\Pi : \tilde{M} \rightarrow M$ is a map $\psi$ from $\tilde{M}$ to itself that satisfies $\Pi \circ \psi = \Pi.$ These automorphisms are often called \emph{deck transformations}. It is a theorem in algebraic topology that any deck transformation of a universal cover is completely determined by its action on a single point.\footnote{This follows from the fact that deck transformations are by definition lifts of the covering map, together with the theorem that any two lifts of a map whose domain is connected must agree everywhere if they agree at a single point. This is proved as proposition 1.34 in \cite{hatcher2001algebraic}.} So after choosing a point $p$ in the base space and a point $\tilde{p}$ in its fiber $\Pi^{-1}(p)$, a loop $\gamma$ based at $p$ determines the deck transformation that sends $\tilde{p}$ to the endpoint of the unique path-lift\footnote{By path-lift we mean we lift only a single copy of $\gamma$ starting and ending at $p$, so that $\tilde{\gamma}$ is a compact curve starting at $\tilde{p}$ and ending at some other point in the fiber $\Pi^{-1}(p)$.} $\tilde{\gamma}$ that starts at $\tilde{p}.$
    
    If we assign $\KL$ an orientation, then we can think of it as an element of the fundamental group $\pi_{1}(\Sigma; p).$ Following the algorithm of the preceding paragraph, the loop $\KL$ passing through $p$, together with our choice of point $\tilde{p} \in \Pi^{-1}(p)$ and our choice of orientation, determines a deck transformation $\psi : \mathbb{H}_2 \rightarrow \mathbb{H}_{2}$. This deck transformation necessarily maps the lift $\tilde{\KL}$ to itself. The action of $\psi$ on a sample lift is sketched in figure \ref{fig:KL-lift}.
    
    Because the covering map $\Pi : \mathbb{H}_{2} \rightarrow \Sigma$ is not only a local homeomorphism but also a local isometry, the deck transformation $\psi$ must be an isometry of $\mathbb{H}_{2}$ to itself. The isometries of $\mathbb{H}_{2}$ are categorized into three types: parabolic, elliptic, and hyperbolic.\footnote{The hyperbolic isometries are sometimes called ``loxodromic,'' especially in the analogous classification of isometries of $\mathbb{H}_{3}.$} See chapter 1 of \cite{marden2007outer} for an introduction to this classification. The assumption that $\Sigma$ is smooth rules out the existence of elliptic deck transformations, because these have fixed points in the universal cover that lead to conical defects in the quotient. The assumption that $\Sigma$ has no cusps rules out the existence of parabolic deck transformations, because parabolic deck transformations have no lower bound on the distance between a point and its image; taking the quotient of $\mathbb{H}_{2}$ by a parabolic deck transformation always creates a cusp. We conclude that $\psi$ must be a hyperbolic map from $\mathbb{H}_{2}$ to itself.

    Hyperbolic isometries of the Poincar\'{e} disk have a geodesic axis that is mapped to itself, with one boundary endpoint acting as a source and the other as a sink. The action of a particular hyperbolic isometry is sketched in figure \ref{fig:hyperbolic-isometry}. We denote the geodesic axis of $\psi$ by $\tilde{\L}.$
    
    \begin{figure}
        \centering
        \includegraphics{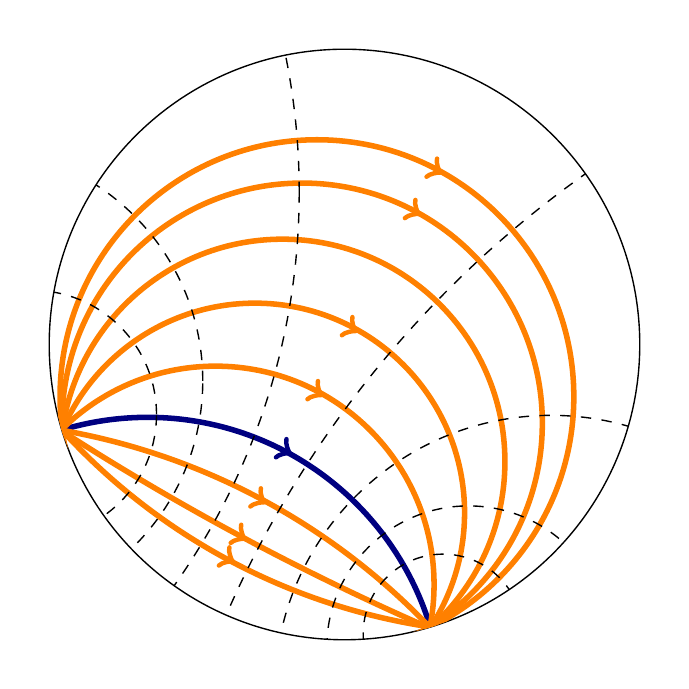}
        \caption{An example of a hyperbolic isometry acting on the Poincar\'{e} disk. The blue curve is a geodesic and is called the ``axis'' of the isometry. The orange curves are invariant curves of the isometry; they are mapped into themselves under a ``push'' whose direction is indicated by the arrows. The dashed lines are mapped one into the next by the isometry. The invariant curves are spaced apart by about half an AdS radius, as are the dashed lines.}
        \label{fig:hyperbolic-isometry}
    \end{figure}
    
    The lift $\tilde{\KL}$ must have well-defined boundary endpoints given by the endpoints to $\tilde{\L}$. This follows from the fact that both are preserved by the isometry; we can pick a fundamental domain for the action of $\psi$ on $\tilde{\L}$, and parametrize $\tilde{\L}$ with a parameter $t$ that gives equal time to each image of the fundamental domain. If we parametrize $\tilde{\KL}$ similarly, then the distance $|\tilde{\KL}(t) - \tilde{\L}(t)|$ is bounded in the limits $t \rightarrow \pm \infty.$ This is because the fundamental domains of $\tilde{\KL}$ and $\tilde{\L}$, being compact, are a bounded distance apart from one another; since $\psi$ is an isometry, that bounded distance is preserved in the limits $t \rightarrow \pm \infty.$
    
    So $\tilde{\KL}$ is a non-self-intersecting, kinked geodesic satisfying conditions (i)-(iii) of lemma \ref{lem:poincare-disk} --- we assumed conditions (i) and (ii), and proved condition (iii) of lemma \ref{lem:poincare-disk} in the preceding paragraph. We can then apply lemma \ref{lem:poincare-disk} to show that $\tilde{\KL}$ is homotopic to $\tilde{\L}$ and that the homotopy region is tiled by right-angled pentagons. If $\rho$ is any of the geodesic segments connecting $\tilde{\KL}$ to $\tilde{\L}$, then the portion of the tiling lying between $\rho$ and its image $\psi(\rho)$ is a fundamental domain for the tiling with respect to $\psi$; this tiling projects down to a tiling of the homotopy region between $\KL$ and $\L \equiv \Pi(\L)$, which proves the lemma. This ``fundamental domain tiling'' argument is sketched in figure \ref{fig:fundamental-domain-tiling}.
    
    \begin{figure}
        \centering
        \includegraphics{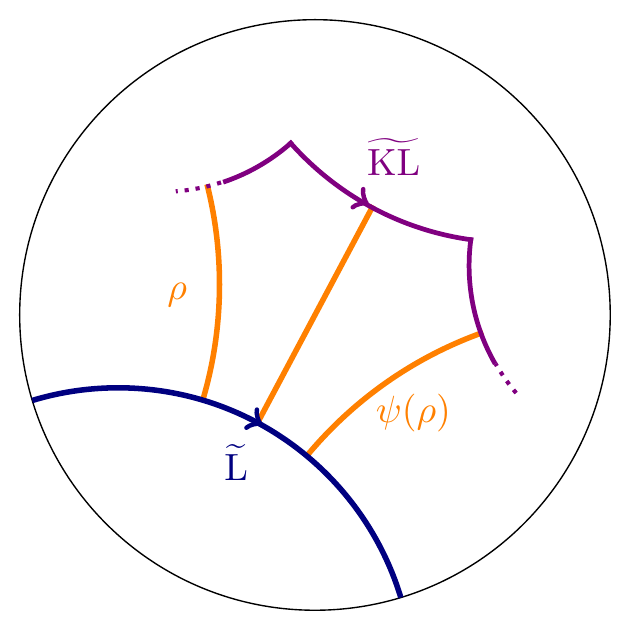}
        \caption{As explained in the proof of lemma \ref{lem:homotopy-shrinking-loop}, the lift $\tilde{\KL}$ gives rise to an isometry $\psi$ with a geodesic axis $\tilde{\L}$. Lemma \ref{lem:poincare-disk} tells us the homotopy region between $\tilde{\KL}$ and $\tilde{\L}$ can be tiled by right-angled pentagons by drawing the orthogonal geodesics between $\tilde{\L}$ and each segment of $\tilde{\KL}$. Here, we have marked one of the geodesics $\rho$ and its isometry image $\psi(\rho)$. The portion of the tiling between these two curves is a fundamental domain for the tiling under the action of $\psi.$}
        \label{fig:fundamental-domain-tiling}
    \end{figure}
    
\end{proof}

%%%%%%%%%%%%%%%%%%%%%%%
\section{Bulk matter and higher dimensions}
\label{sec:generalizations}

It is natural to ask whether inequality \eqref{eq:AdS3-bound-2}, derived in the previous section for time-symmetric states in pure AdS$_3$ gravity, has an analogue in more general theories of gravity. In this section, we take up that question on a few fronts.

In subsection \ref{subsec:spherical-matter}, we give a simple calculation showing that in \emph{any} asymptotically AdS$_3$ spacetime with a spherically symmetric moment of time symmetry, the area contribution to the Markov gap for two antipodal, equal-size intervals in the moment of time symmetry approaches the universal value $[S_R - I]_{\text{area}} = \log(2) / G_N$ in the limit as the intervals become large. Because the Markov gap is saturated for two vacuum intervals only in this limit (cf. section \ref{subsec:two-special-cases}), this implies that inequality \eqref{eq:AdS3-bound-2} holds at the classical level for arbitrary spherically symmetric perturbations to the state of two intervals in the vacuum, to all orders in perturbation theory. In subsection \ref{subsec:quantum-matter}, we show that the inequality holds even if the perturbation contains quantum matter, up to some mild assumptions about limits of bulk entanglement entropies; we also comment on some general features of the classical and quantum contributions to perturbations of the Markov gap. In subsection \ref{subsec:higher-dimensions}, we discuss a possible generalization of \eqref{eq:AdS3-bound-2} to higher dimensions, where the number of cross-section boundaries is replaced by the codimension-$3$ area of the cross-section boundary.

We do not address the issue of non-time-symmetric states here. However, we suspect that the techniques that must be developed to generalize inequality \eqref{eq:AdS3-bound-2} to states with classical matter and higher dimensions will naturally lend themselves to bounding the Markov gap in general, non-time-symmetric states; we comment on this research direction further in the discussion (section \ref{subsec:future-generalizations}).

%%%%%%%%%%%
\subsection{Spherically symmetric matter in three dimensions}
\label{subsec:spherical-matter}

The setting for this subsection will be an arbitrary asymptotically AdS$_3$ spacetime with a spherically symmetric moment of time symmetry.\footnote{When we say ``spherically symmetric,'' we will mean that the spacetime possesses not only a rotational symmetry along an angle $\theta$ but also a $\theta \mapsto -\theta$ reflection symmetry.} The metric on the moment of time symmetry can be written
\begin{equation} \label{eq:spher-sym-spacetime}
    ds^2 = f(r) dr^2 + r^2 d\theta^2
\end{equation}
with $r \in [0, \infty)$. Vacuum AdS$_3$ is the spacetime with $f(r) = 1 / (1+r^2).$ The requirement that the spacetime be asymptotically AdS$_3$ imposes that in a large-$r$ expansion, $f(r)$ differs from the vacuum value $1 / (1+r^2)$ only at order $O(1/r^4)$.

In any such metric, let us consider the Ryu-Takayanagi surface corresponding to a boundary region of angular extent $\Delta \theta$. Because of the spherical symmetry, the geometry of the surface depends only on the angular extent of the boundary region and not on its position; it will also have a reflection symmetry in $\theta$. The coordinate position of such a Ryu-Takayanagi surface is determined by an equation $\theta = h(r)$; the induced metric on this surface is
\begin{equation}
    ds^2 = (f(r) + r^2 h'(r)^2) dr^2,
\end{equation}
and the induced volume form is
\begin{equation} \label{eq:volume-form}
    \sqrt{g_{\text{RT}}}
        = \sqrt{f(r) + r^2 h'(r)^2}.
\end{equation}
If $\theta = h(r)$ is a minimal surface, then it satisfies the Euler-Lagrange equations of motion for the induced volume form, which are
\begin{equation} \label{eq:intermed-EL}
    \frac{r^2 h'(r)}{\sqrt{f(r) + r^2 h'(r)^2}} = \text{const.}
\end{equation}
Let $r_*$ be the minimal value of $r$ attained by the minimal surface $\theta = h(r).$ At $r=r_*$, we have $h'(r) = \infty$; from this relation, we can fix the constant in \eqref{eq:intermed-EL} and solve for $h'(r)^2$ to obtain the equation of motion
\begin{equation} \label{eq:final-EL}
    h'(r)^2 = \frac{r_*^2 f(r)}{r^2 (r^2 - r_*^2)}.
\end{equation}

The total length of a minimal surface contained between a turning point $r=r_*$ and a radial cutoff $r = 1/\epsilon$ can be obtained by plugging \eqref{eq:final-EL} into the volume form \eqref{eq:volume-form}, multiplying by two (because $\theta$ is two-valued as a function of $r$), and integrating. We obtain the expression
\begin{equation} \label{eq:regulated-length}
    L_{\epsilon}(r_*)
        = 2 \int_{r_*}^{1/\epsilon} dr\, r \sqrt{\frac{f(r)}{r^2 - r_*^2}}.
\end{equation}
To compute the mutual information of two antipodal, equal-size intervals $A$ and $B$, we will need to consider two types of minimal surfaces. The two ``disconnected'' surfaces individually homologous to $A$ and $B$, whose turning point we will label $r_*$, and the ``connected'' surfaces that are homologous only to the union $A \cup B$, whose turning point we will label $r_*^c$.\footnote{Explicit formulas for these turning points in terms of the angular extents of the intervals can be obtained, but we will not need them for the present calculation.} Assuming that the intervals are large enough that the connected surface is globally minimal, the classical contribution to the mutual information is given via the Ryu-Takayanagi formula by
\begin{equation}
    I(A:B)_{\text{area}} = \lim_{\epsilon \rightarrow 0} \frac{2 L_{\epsilon}(r_*) - 2 L_{\epsilon}(r_*^c)}{4 G_N}.
\end{equation}
Plugging in formula \eqref{eq:regulated-length} for the lengths, we may write this as
\begin{equation}
    I(A:B)_{\text{area}} = \frac{1}{G_N} \lim_{\epsilon \rightarrow 0} \left[ \int_{r_*}^{1/\epsilon} dr\, r \sqrt{\frac{f(r)}{r^2 - r_*^2}} - \int_{r_*^c}^{1/\epsilon} dr\, r \sqrt{\frac{f(r)}{r^2 - (r_*^c)^2}} \right].
\end{equation}
Partially combining the integrals gives the expression
\begin{equation}
    I(A:B)_{\text{area}} = \frac{1}{G_N} \left[ \int_{r_*}^{r_*^c} dr\, r \sqrt{\frac{f(r)}{r^2 - r_*^2}} + \int_{r_*^c}^{\infty} dr\, r \sqrt{f(r)} \left( \frac{1}{\sqrt{r^2 - r_*^2}}  - \frac{1}{\sqrt{r^2 - (r_*^c)^2}}  \right)\right],
\end{equation}
where we have removed all $\epsilon$ dependence because both integrals in this expression converge.

Thanks to the spherical symmetry of the metric, the entanglement wedge cross-section is guaranteed to be a line of constant $\theta$ with endpoints at $r = r_*^c$. The length of this surface is given by
\begin{equation}
    L(\sigma_{A:B}) = 2 \int_{0}^{r_*^c} \sqrt{f(r)}.
\end{equation}
Using the Dutta-Faulkner formula $S_{R,\text{area}} = L(\sigma_{A:B}) / 2 G_N$, we may now write the classical contribution to the Markov gap as
\begin{align}
    [S_R(A:B) - I(A:B)]_{\text{area}}
        & = \frac{1}{G_N} \left[ \int_{0}^{r_*^c} \sqrt{f(r)}
            - \int_{r_*}^{r_*^c} dr\, r \sqrt{\frac{f(r)}{r^2 - r_*^2}} \right. \nonumber \\
            & \qquad \left. - \int_{r_*^c}^{\infty} dr\, r \sqrt{f(r)} \left( \frac{1}{\sqrt{r^2 - r_*^2}}  - \frac{1}{\sqrt{r^2 - (r_*^c)^2}}  \right)
        \right].
\end{align}
In the limit where the two intervals become large and collectively take up the entire boundary, spherical symmetry guarantees the limits $r_* \rightarrow 0$ and $r_*^c \rightarrow \infty$. It is straightforward to take the limit $r_* \rightarrow 0$ in the above expression; the first two integrals cancel in this limit, and the third integral simplifies, giving
\begin{equation}
    \lim_{r_* \rightarrow 0} [S_R(A:B) - I(A:B)]_{\text{area}}
        = \frac{1}{G_N} \int_{r_*^c}^{\infty} dr\, r \sqrt{f(r)} \left( \frac{1}{\sqrt{r^2 - (r_*^c)^2}} - \frac{1}{r}  \right).
\end{equation}
To study the $r_*^c \rightarrow \infty$ limit of this integral, it will be convenient to make the substitution $r \mapsto r_*^c \sec{\alpha}$, which removes $r_*^c$ from the limits of integration. Under this substitution, we have
\begin{equation}
    \lim_{r_* \rightarrow 0} [S_R(A:B) - I(A:B)]_{\text{area}}
        = \frac{r_*^c}{G_N} \int_{0}^{\pi/2} d\alpha\, \sqrt{f(r_*^c \sec{\alpha})} \sec{\alpha} (\sec{\alpha} - \tan{\alpha})
\end{equation}
In the $r_*^c \mapsto \infty$ limit, we may approximate $f(r_*^c \sec{\alpha})$ by its large-argument expansion $f(r_*^c \sec{\alpha}) \sim 1/(r_*^c \sec{\alpha})^2$. Making this approximation gives 
\begin{equation}
    \lim_{r_*^c \rightarrow \infty} \lim_{r_* \rightarrow 0} [S_R(A:B) - I(A:B)]_{\text{area}}
        = \frac{1}{G_N} \int_{0}^{\pi/2} d\alpha\, (\sec{\alpha} - \tan{\alpha})
        = \frac{\log(2)}{G_N}.
\end{equation}

The result of this calculation seems quite suggestive. For any asymptotically AdS$_3$ spacetime with a spherically symmetric moment of time symmetry, the classical contribution to the Markov gap of two large intervals approaches the lower bound of inequality \eqref{eq:AdS3-bound-2}. This seems to hint at some underlying universality in the Markov gap for two intervals, at least in holographic conformal field theories.

As indicated in the introduction to this section, this calculation also implies that the bound \eqref{eq:AdS3-bound-2} is respected at the classical level for arbitrary spherically symmetric perturbations to the state of two vacuum intervals, to all orders in perturbation theory. This follows from the fact that \eqref{eq:AdS3-bound-2} is saturated only in the limit $r_* \rightarrow 0, r_*^c \rightarrow \infty$; but in that limit, we have just shown that the Markov gap for any spherically symmetric metric approaches that of the vacuum.

%%%%%%%%%%%
\subsection{Quantum perturbations}
\label{subsec:quantum-matter}

In the presence of quantum matter, the Dutta-Faulkner formula is given by equation \eqref{eq:dutta-faulkner}. In small-$G_N$ perturbation theory, the minimum over quantum extremal surfaces appearing in the holographic entanglement entropy formula can be replaced by the generalized entropy of the classically minimal surface --- in fact, this calculation by Faulkner, Lewkowycz, and Maldacena in \cite{FLM} served as the prelude to the general quantum extremal surface formula \cite{QES}. While it is important to remember that this perturbative approach can miss important nonperturbative contributions coming from entanglement islands \cite{islands1, islands2} or large breakdowns of entanglement wedge reconstruction \cite{akers2019large}, it is still useful in regimes far from any phase transition where multiple quantum extremal surfaces vie for dominance. Analogously, in small-$G_N$ perturbation theory, the Dutta-Faulkner formula can be replaced by
\begin{equation}
    S_R(A:B) \approx \frac{2 \area(\sigma_{A:B})}{4 G_N} + S_{R, \text{bulk}}(\sigma_{A:B}),
\end{equation}
where $\sigma_{A:B}$ is the \emph{classical} entanglement wedge cross-section.

In this regime, the Markov gap can be written
\begin{equation} \label{eq:perturbative-DF}
    S_R(A:B) - I(A:B)
        \approx [S_R(A:B) - I(A:B)]_{\text{area}} + S_{R, \text{bulk}}(\sigma_{A:B})
                - I_{\text{bulk}}(A:B).
\end{equation}
We need to be a little careful about what each term in this equation means. The area term is just the sum of area differences for the appropriate classical KRT and RT surfaces. The reflected entropy term is also fairly straightforward: it is the reflected entropy of the quantum fields in the bulk entanglement wedge $\W(AB)$ subject to the bipartition induced by the cross-section $\sigma_{A:B}.$ The bulk term $I_{\text{bulk}}(A:B)$, however, is \emph{not} the mutual information of $\W(AB)$ subject to that bipartition. It is the sum of the entropies of the regions bounded by $\RT(A)$ and $\RT(B)$, minus the entropy of the region bounded by $\RT(AB)$; it is not actually the mutual information of a state of the bulk quantum fields.

\begin{figure}
    \centering
    \includegraphics{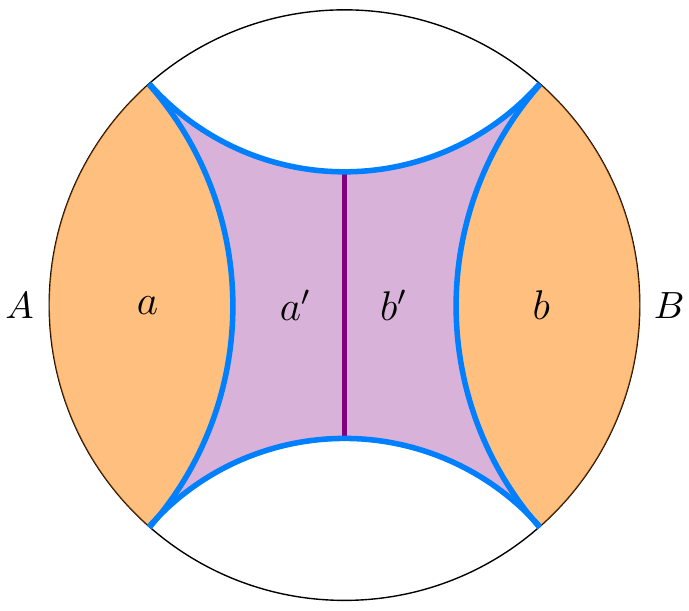}
    \caption{The bulk regions one needs to consider to compute the bulk term in equation \eqref{eq:perturbative-DF}. The region $a$ lies between $A$ and $\RT(A)$, $a'$ lies between $\RT(A)$ and $\KRT(A)$, and analogously for $b$ and $b'.$}
    \label{fig:bunch-of-bulk-regions}
\end{figure}

See figure \ref{fig:bunch-of-bulk-regions} for a representation of the various bulk regions that need to be considered in this calculation. $S_{R,\text{bulk}}(\sigma_{A:B})$ is the reflected entropy of the bulk state on $aa' \cup bb'$. $I_{\text{bulk}}(A:B)$ is the entropy of the bulk state on $a$ plus that of the state on $b$, minus the entropy of the bulk state on $a a' b b'$. Using nonnegativity of the Markov gap for general quantum states gives us
\begin{equation}
    S_{R, \text{bulk}}(\sigma_{A:B})
        = S_R(aa':bb') 
        \geq I(a a' : bb')
        = S(a a') + S(b b') - S(a a' b b').
\end{equation}
Applying this inequality to the bulk term appearing in equation \eqref{eq:perturbative-DF} gives the inequality
\begin{align}
    S_{R, \text{bulk}}(\sigma_{A:B})
                - I_{\text{bulk}}(A:B)
        & = S_{R, \text{bulk}}(\sigma_{A:B})
                - S(a) - S(b) + S(aa' bb')  \nonumber \\
        & \geq S(a a') + S(b b') - S(a) - S(b). \label{eq:bulk-markov-gap}
\end{align}
For general quantum states, the right-hand side of this inequality is not necessarily nonnegative. This may seem puzzling --- after all, as we noted in the introduction in equation \eqref{eq:SRmI-Sgen}, combining the area and bulk terms in equation \eqref{eq:perturbative-DF} with inequality \eqref{eq:bulk-markov-gap} gives a difference in generalized entropies that is guaranteed to be nonnegative when one uses the full, nonperturbative quantum extremal surface formula. At the level of perturbation theory, though, the potential negativity of the bulk contribution to \eqref{eq:perturbative-DF} seems like it could be a problem. However, this problem only arises when the area term in \eqref{eq:perturbative-DF} vanishes --- the perturbative bulk term cannot compete with a nonzero area term, since the area term contributes to the Markov gap at order $1/G_N$ while the perturbative bulk term contributes at order $1$. The potential issue in equation \eqref{eq:perturbative-DF} arising in perturbation theory when the area term vanishes is a manifestation of the principle explained in the first paragraph of this subsection, that perturbative holographic entropy formulas cannot be trusted near entanglement phase transitions.

For the case discussed in subsection \ref{subsec:spherical-matter}, however, where the boundary regions are equal-size antipodal intervals in a spherically symmetric time-slice, it seems reasonable to assume that the right-hand side of inequality \eqref{eq:bulk-markov-gap} will vanish in the limit as the intervals become large. In this limit regions $aa'$ and $a$ both approach a ``half-spacetime'' region bounded in the bulk by a line of constant $\theta$. While they approach this limit in different ways --- for example, the boundary of $aa'$ always has corners while the boundary of $a$ does not --- the bulk entropies appearing in \eqref{eq:bulk-markov-gap} are supposed to be renormalized entropies that capture universal contributions to the entanglement entropy without counting boundary divergences.\footnote{There are some subtleties to keep in mind here, because quantum field theory states on spatial regions with non-smooth boundaries --- such as $aa'$ --- tend to have divergences associated with those corners. Furthermore, higher-derivative corrections to the classical piece of the generalized entropy involve extrinsic curvature terms that are ill-defined at corners. Under a suitable regulation procedure, where the corners are smoothed out enough to regulate the corner divergences and to define the higher-derivative classical entropy, but not enough to appreciably change the area contribution to the entropy, it seems reasonable to hope that the generalized entropy of regions with boundary corners is renormalizable and UV-finite. See section 4.2 of \cite{quantum-maximin} for some discussion of this point.} Barring subtleties in this universality, we suspect that the quantities $S(aa') - S(a)$ and $S(bb') - S(b)$ for antipodal, equal-size intervals in a spherically symmetric state will vanish in the limit as those intervals become large. This establishes the claim made in the introduction to this section, that the bound $S_R - I \geq \log(2) / G_N$ holds in perturbation theory for such ``symmetric two-large-interval states'' even in the presence of quantum matter.

Proving an inequality like \eqref{eq:AdS3-bound-2} for general holographic states with quantum matter seems quite delicate. For states without quantum matter, one could imagine that a clever application of the weak or null energy condition could be used to guarantee inequality \eqref{eq:AdS3-bound-2} in arbitrary states. States with quantum matter are known to violate all the energy conditions satisfied by classical matter, however, and it seems possible that states with quantum matter could cause the area-term contribution to the Markov gap to violate inequality \eqref{eq:AdS3-bound-2}. For the bound to hold in such states, the bulk contribution $S_{R, \text{bulk}} - I_{\text{bulk}}$ would need to be large enough to make up for the deficit in the area term. If this is true, then it implies a rather delicate energy condition in bulk quantum field theories: whenever the energy configuration of the quantum matter is such that $[S_R - I]_{\text{area}}$ dips below the lower bound imposed by equation \eqref{eq:AdS3-bound-2}, the bulk contribution $S_{R, \text{bulk}} - I_{\text{bulk}}$ will need to be larger than the magnitude of that dip.

Should we expect this to be true? Maybe! It certainly seems like a rather intricate condition to impose on bulk quantum fields, especially since, as explained in the discussion surrounding equation \eqref{eq:bulk-markov-gap}, the bulk contribution $S_{R, \text{bulk}} - I_{\text{bulk}}$ is not the Markov gap of any bulk state and does not a priori need to be nonnegative. Determining whether $S_{R, \text{bulk}} - I_{\text{bulk}}$ must be nonnegative due to some other general principle would be an interesting avenue of investigation. Despite the intricacy of these conditions, the classical universality of \eqref{eq:AdS3-bound-2} suggested by the analysis of subsection \ref{subsec:spherical-matter} indicates to us that this inequality really is capturing some interesting universal feature of holographic entanglement. It is our \emph{hope}, then, that the inequality is truly universal or has some universal generalization. However, it is important to note that another classical holographic entropy inequality, the monogamy of mutual information \cite{hayden2013holographic, maximin}, does not generally hold in states with quantum matter \cite{quantum-maximin}.

We comment further on these points in the discussion (section \ref{sec:discussion}).

%%%%%%%%%%%
\subsection{Higher dimensions}
\label{subsec:higher-dimensions}

In inequality \eqref{eq:AdS3-bound-2}, the Markov gap is lower bounded by a universal constant times the number of boundaries in the entanglement wedge cross-section. Because the entanglement wedge cross-section is one-dimensional, its boundary is zero-dimensional, and the only geometric quantity we can associate to it is the number of components. In higher dimensions, the boundary of the entanglement wedge cross-section can be more complicated; it is a codimension-$3$ surface, and one natural generalization of the number of connected components is its codimension-$3$ volume. So we might expect that in $D+1$ bulk dimensions, the generalization of \eqref{eq:AdS3-bound-2} will be something like
\begin{equation} \label{eq:higher-D-inequality}
    S_R(A:B) - I(A:B) \geq \frac{C_{D}}{4 G_N} \times \area(\text{cross-section boundary})
\end{equation}
with $C$ a universal, dimension-dependent constant. If we define the area of a zero-dimensional surface to be its number of components, this bound reproduces \eqref{eq:AdS3-bound-2} in $D=2$ with $C_{2} = 2 \log(2).$

We have attempted to check \eqref{eq:higher-D-inequality} in the case $D>2$ by computing $S_R - I$ numerically for some simple families of states and minimizing the quantity $4 G_N (S_R - I) / \area.$ We hoped that we would find some universal minimum appearing in multiple, a priori unrelated families of states --- this is what happens in the case $D=2$, where both the ``two-interval'' states and the ``two-asymptotic boundaries of a wormhole'' states achieve, as a limit within their families, the minimum value $C_2 = 2 \log(2).$ Unfortunately, this is not the case for the families we have considered in the case $D>2$. This is not evidence \emph{against} the conjectured inequality \eqref{eq:higher-D-inequality}, but it does indicate that the situation above three bulk dimensions is subtle.

The two families we considered in four bulk dimensions were: (i) $A$ and $B$ are two equal-time strips on the boundary of AdS$_{D+1}$ in Poincar\'{e} coordinates; (ii) $A$ and $B$ are two equal-time antipodal caps on the boundary of AdS$_{D+1}$ in global coordinates. While these cases are equivalent in three bulk dimensions, they are not in higher dimensions --- for example, the mutual information between two strips is infinite because they ``touch'' at the point on the boundary sphere that isn't covered by the Poincar\'{e} coordinates (see \cite{fischler2013holographic} for details).

We will not reproduce the calculations here explicitly, because we do not find them particularly instructive. The goal is to compute the quantity
\begin{equation}
    \chi = 4 G_N \frac{S_R(A:B) - I(A:B)}{\area(\text{cross-section boundary})}.
\end{equation}
The mutual information between two strips in arbitrary dimension was computed analytically in \cite{fischler2013holographic}; the mutual information between two caps was computed numerically in \cite{colin2020large}. The methods of those two papers can be adapted quite straightforwardly to compute $S_R(A:B)$ and the codimension-3 area of the cross-section boundary. The final results, in arbitrary dimension $D>2$, are:
\begin{enumerate}
    \item When $A$ and $B$ are vacuum strips in boundary Poincar\'{e} coordinates, $\chi$ monotonically decreases with the size of the strips and monotonically increases with the separation between the strips. Its minimum value, achieved when the separation goes to zero and the strip size goes to infinity, is
    \begin{equation} \label{eq:strip-lower-bound}
        4 G_N \frac{S_R(A:B) - I(A:B)}{\area(\text{cross-section boundary})}
            \geq \frac{2}{D-2} \left[ 1 - \sqrt{\pi} \frac{\Gamma\left(\frac{D}{2(D-1)}\right)}{\Gamma\left(\frac{1}{2(D-1)}\right)}\right].
    \end{equation}
    One can check that the right-hand side goes to $2 \log(2)$ for $D \rightarrow 2.$
    
    \item When $A$ and $B$ are vacuum caps in boundary global coordinates, $\chi$ approaches the right-hand side of \eqref{eq:strip-lower-bound} in the limit as the caps become large and their separation becomes small. This isn't so surprising, because the strips become ``cap-like'' in global coordinates as they increase in size and decrease in separation. However, unlike in the case of strips, in the case of caps $\chi$ can be shown numerically to \emph{decrease} with the separation between the caps --- i.e., the minimal value of $\chi$ for caps is achieved when the caps are as small as possible while still having a connected entanglement wedge, a limit in which there is generally no analytical control of the calculation. (See section 2.2 of~\cite{colin2020large}.) There are, therefore, cap configurations that violate inequality \eqref{eq:strip-lower-bound}, and the minimum value of $\chi$ for antipodal caps has no nice analytic formula.
\end{enumerate}

From these considerations, it is clear that establishing an inequality like \eqref{eq:higher-D-inequality} in arbitrary dimensions is a nuanced undertaking. We mention some techniques that might be helpful for that problem in the discussion (section \ref{subsec:future-generalizations}).

%%%%%%%%%%%%%%%%%%%%%%%
\section{Holographic Markov recoveries and fixed area states}
\label{sec:recovery-models}

In this section, we revisit the Markov recovery processes introduced in section \ref{sec:info-theory}. The essential observation of that section was that the Markov gap $S_R - I$ for a state $\rho_{AB}$ is lower-bounded by a function of the fidelity of Markov recovery processes on the canonical purification. We reproduce the relevant inequalities here:
\begin{align}
    S_R(A:B) - I(A:B) \label{eq:BBstar-Markov-2}
        & \geq - \max_{\mathcal{R}_{B \rightarrow B B^*}} \log F(\rho_{A B B^*}, \mathcal{R}_{B \rightarrow B B^*}(\rho_{AB})). \\
    S_R(A:B) - I(A:B) \label{eq:AAstar-Markov-2}
        & \geq - \max_{\mathcal{R}_{A \rightarrow A A^*}} \log F(\rho_{A A^* B}, \mathcal{R}_{A \rightarrow A A^*}(\rho_{AB})).
\end{align}
Focusing on the first inequality, we argued in section \ref{subsec:bulk-gap} that even the best Markov recovery channel $\mathcal{R}_{B \rightarrow BB^*}$ will be unable to reproduce the boundary entanglement supporting the portion of the quantum extremal surface $\gamma_{AB}$ that lies outside the $B B^*$ entanglement wedge.\footnote{In the language of section \ref{sec:geometric-proof}, this is the non-cross-section portion of the ``kinked Ryu-Takayanagi'' surface $\KRT(A).$} We will call this surface $\J(A)$, referring to section \ref{subsec:bulk-gap}, where we called it a ``jagged surface'' based on the way it was depicted in figures \ref{fig:intervals-jagged-surfaces} and \ref{fig:wormhole-jagged-surfaces}. The failure of the optimal Markov recovery map to produce the entanglement required to support a geometric connection across $\J(A)$ implies an imperfect fidelity of recovery, and therefore a nonzero contribution to the right-hand side of inequality \eqref{eq:BBstar-Markov-2}.

The intuition given in the preceding paragraph is entirely qualitative. It does not give us any quantitative estimate for how the failure of $\mathcal{R}_{B \rightarrow B B^*}$ to create the entanglement supporting $\J(A)$ constrains the fidelity appearing in inequality \eqref{eq:BBstar-Markov-2}. The primary achievement of this section will be to provide such a quantitative estimate. We will argue, in a fixed-area-state toy model of the Markov recovery process, that any Markov recovery map $\mathcal{R}_{B \rightarrow B B^*}$ must have fidelity satisfying the upper bound
\begin{equation} \label{eq:FA-fidelity-bound}
    F(\rho_{ABB^*}, \mathcal{R}_{B \rightarrow B B^*}(\rho_{AB}))
        \leq e^{- \Delta}
                \left[{}_{2} F_1\left(\frac{1}{2}, - \frac{1}{2}; 2; e^{- \Delta} \right) \right]^2,
\end{equation}
with
\begin{equation} \label{eq:entropy-delta}
    \Delta = \frac{\area(\KRT(A)) - \area(\RT(A))}{4 G_N}
\end{equation}
being the entropy difference between the KRT and RT surfaces defined in the introduction. Plugging inequality \eqref{eq:FA-fidelity-bound} in to inequality \eqref{eq:BBstar-Markov-2} gives the lower bound
\begin{equation} \label{eq:final-fidelity-inequality}
    S_R(A:B) - I(A:B)
        \geq \Delta - 2 \log\left[ {}_{2} F_1\left(\frac{1}{2}, - \frac{1}{2}; 2; e^{- \Delta} \right) \right].
\end{equation}
Estimating the right-hand side of \eqref{eq:BBstar-Markov-2} for states with quantum matter --- whose boundary entropies include bulk entropy contributions and thus cannot be modeled as straightforwardly using fixed area states\footnote{However, see \cite{akers-penington} for some entropy calculations in fixed area states with bulk matter.} --- is a more intricate task that we do not undertake in this paper; however, we roughly expect that in suitable toy models, the right-hand side of equation \eqref{eq:entropy-delta} would be replaced by a difference of generalized entropies.

We pause to note that one can check, using $\Delta \geq 0$, that the right-hand side of \eqref{eq:final-fidelity-inequality} is greater than or equal to $\Delta$. If we take this toy model seriously, then we learn that $S_R - I$ is lower-bounded by the entropy difference $S_{\text{gen}}(\KRT(A)) - S_{\text{gen}}(\RT(A))$. At this point, we could declare success; applying the same analysis for the Markov chain $B \rightarrow A \rightarrow A^*$ would give us the two inequalities
\begin{align}
    S_R(A:B) - I(A:B)
        & \geq S_{\text{gen}}(\KRT(A)) - S_{\text{gen}}(\RT(A)). \label{eq:Markov-A-bound} \\
    S_R(A:B) - I(A:B) 
        & \geq S_{\text{gen}}(\KRT(B)) - S_{\text{gen}}(\RT(B)). \label{eq:Markov-B-bound}
\end{align}
If we look back to inequality \eqref{eq:SRmI-Sgen} from the introduction, we could interpret this analysis as teaching us that Markov recovery arguments can reproduce each term of inequality \eqref{eq:SRmI-Sgen} individually. This might, however, seem like a bit of a letdown; after all, since \eqref{eq:SRmI-Sgen} tells us that $S_R - I$ is lower-bounded by the \emph{sum} of these entropy differences, it is stronger than the two inequalities \eqref{eq:Markov-A-bound} and \eqref{eq:Markov-B-bound}. However, we shouldn't expect either of the inequalities \eqref{eq:Markov-A-bound} or \eqref{eq:Markov-B-bound} to be anywhere close to saturation. For one thing, the information inequalities \eqref{eq:BBstar-Markov-2} \& \eqref{eq:AAstar-Markov-2} are not guaranteed to be tight; it is entirely possible, in fact even generic, for the optimal recovery map $\mathcal{R}_{B \rightarrow BB^*}$ to have fidelity greater than $e^{-(S_R - I)}.$ There is also, as we explain in subsection \ref{subsec:tgt-recovery}, generically a significant gap in inequality \eqref{eq:FA-fidelity-bound}. Finally, the second term on the right-hand side of \eqref{eq:final-fidelity-inequality} is never exactly zero. These three sources of error will lead to significant gaps between the left- and right-hand sides of \eqref{eq:Markov-A-bound} and \eqref{eq:Markov-B-bound}, consistent with inequality \eqref{eq:SRmI-Sgen}.

In the following subsections, we give the details of the argument sketched above. In subsection \ref{subsec:tgt-recovery} we describe a particular state $\tilde{\rho}_{ABB^*}$, which we call the ``Too Good to be True'' or TGT state, and give physical arguments for the inequality
\begin{equation}
    F(\rho_{ABB^*}, \mathcal{R}_{B \rightarrow B B^*}(\rho_{AB})) \leq F(\rho_{ABB^*}, \tilde{\rho}_{ABB^*}).
\end{equation}
In subsection \ref{subsec:fixed-area}, we compute the fidelity $F(\rho_{ABB^*}, \tilde{\rho}_{ABB^*})$ for analogous states in fixed-area models, and reproduce inequality \eqref{eq:FA-fidelity-bound}.

%%%%%%%%%%%
\subsection{``Too Good to be True'' recovery}
\label{subsec:tgt-recovery}

The three party reduced state $\rho_{ABB^*}$ has, as its entanglement wedge, the spacelike portion of the canonical purification spacetime bounded in the bulk by $\RT(A^*)$ and on the boundary by $ABB^*$. We sketched a spacelike slice of this entanglement wedge in figure \ref{fig:intervals-jagged-surfaces} for the case where $\rho_{AB}$ is the density matrix of two equal-time intervals in the AdS$_3$ vacuum; we have reproduced a version of that figure with labeling more appropriate for our present purposes as figure \ref{fig:TGT-state}.

\begin{figure}
    \centering
    \includegraphics{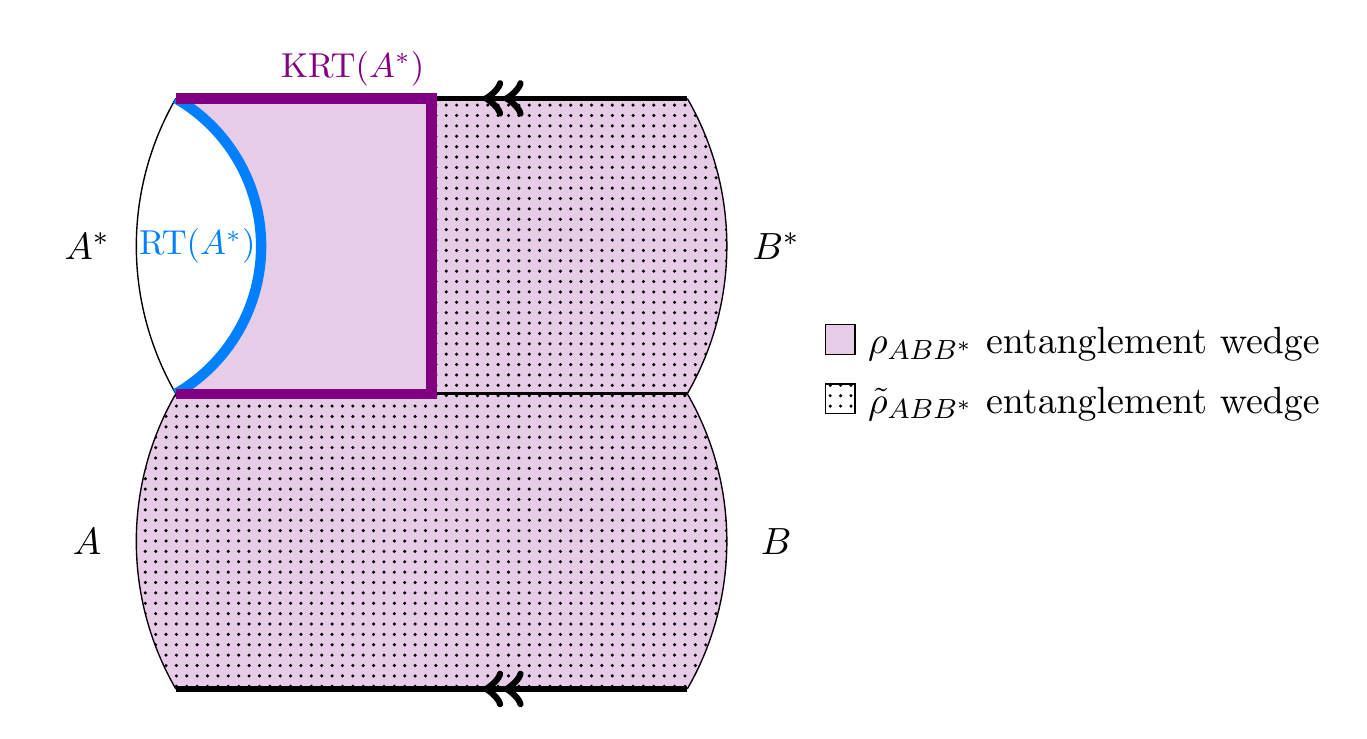}
    \caption{A sketch of the canonical purification for two equal-time intervals in the AdS$_3$ vacuum, including a sketch of two possible states of the system $ABB^*.$ The state $\rho_{ABB^*}$, obtained as a partial trace of the canonical purification, has $\RT(A^*)$ as its entanglement wedge boundary. The TGT state $\tilde{\rho}_{ABB^*}$ has $\KRT(A^*)$ as its entanglement wedge boundary.}
    \label{fig:TGT-state}
\end{figure}

By the ``Too Good to be True'' state, which we will denote $\tilde{\rho}_{ABB^*}$ and sometimes call the TGT state, we will mean a state whose bulk entanglement wedge is the union of the $A B$ and $B B^*$ entanglement wedges. In figure \ref{fig:TGT-state}, this is the bulk region filled with dots. In the language of section \ref{sec:geometric-proof}, this is a state whose entanglement wedge is bounded in the bulk by the KRT surface $\KRT(A^*).$

We pause, for a moment, to reflect on the definition of this state. Within the semiclassical theory defined on the spacetime dual to the canonical purification, the state $\tilde{\rho}_{ABB^*}$ is perfectly well defined --- one simply traces out the bulk QFT degrees of freedom lying between $\RT(A^*)$ and $\KRT(A^*)$. Boundary states, however, are not exactly dual to a single bulk spacetime; the path integrals preparing them admit multiple saddles that contribute to their physics, and tracing out a portion of the dominant saddle is not an operation that has an obvious boundary interpretation. When we talk about the TGT state, we will really mean \emph{any} state for which low-energy correlators and entanglement entropies are accurately reproduced, using the usual AdS/CFT dictionary, by the portion of the canonical purification spacetime lying between $\KRT(A^*)$ and $ABB^*$. In particular, the entropy $S(\tilde{\rho}_{ABB^*})$ will be the generalized entropy of the KRT surface $\KRT(A^*)$. There may be a whole family of states with this property, differing in the details of their subleading saddles, but their physical properties should only differ nonperturbatively in $1/G_N.$ Or perhaps it is not actually possible to construct a TGT state in full AdS/CFT, since the kinks in $\KRT(A^*)$ prevent it from being truly quantum extremal. It \emph{will} be possible to give a precise construction of the TGT state in the fixed-area toy model of the following subsection, which is good enough for our purposes.

Ultimately, the details of the TGT state $\tilde{\rho}_{ABB^*}$ will not be so important. All we will really want is a prescription for computing the fidelity $F(\rho_{ABB^*}, \tilde{\rho}_{ABB^*}).$ The reason for this is that we will assume, on physical grounds, that any Markov recovery map $\mathcal{R}_{B \rightarrow BB^*}$ acting on $\rho_{AB}$ satisfies the inequality
\begin{equation} \label{eq:interm-fidelity-inequality}
    F(\rho_{ABB^*}, \mathcal{R}_{B \rightarrow B B^*}(\rho_{AB})) \leq F(\rho_{ABB^*}, \tilde{\rho}_{ABB^*}).
\end{equation}
This follows, essentially, from the physical arguments made in section \ref{subsec:bulk-gap}. Any $B \rightarrow BB^*$ channel that touches only the $B$ subsystem of $\rho_{AB}$ can \emph{at best} reproduce the $AB$ entanglement wedge, the $BB^*$ entanglement wedge, and whatever shared entanglement is needed to attach those wedges smoothly. In fact, even this is usually too much to ask, which is why we called $\tilde{\rho}_{ABB^*}$ the ``too good to be true'' state.

Indeed, suppose it were possible to obtain the TGT state via a Markov recovery map, i.e., assume we have
\begin{equation}
    \tilde{\rho}_{ABB^*} = \mathcal{R}_{B \rightarrow B B^*}(\rho_{AB})
\end{equation}
for some quantum channel $\mathcal{R}_{B \rightarrow B B^*}.$ It is straightforward to verify the equality
\begin{equation} \label{eq:relent-markov-gap}
    S_R(A:B) - I(A:B) = - S(\rho_{ABB^*}) + S(\rho_{BB^*}) - S(\rho_{AB}||I_A \otimes \rho_{B}),
\end{equation}
where $S(\cdot||\cdot)$ is the relative entropy. Monotonicity of relative entropy under application a quantum channel gives
\begin{equation}
    S(\rho_{AB}||I_A \otimes \rho_{B})
        \leq S(\mathcal{R}_{B \rightarrow BB^*}(\rho_{AB})|| I_A \otimes \mathcal{R}_{B \rightarrow BB^*}(\rho_{B}))
        = S(\tilde{\rho}_{ABB^*} || I_A \otimes \rho_{BB^*}),
\end{equation}
where we have used the fact that the $BB^*$ reduced state of $\tilde{\rho}_{A BB^*}$ agrees with that of $\rho_{ABB^*}.$ Plugging this inequality back into equation \eqref{eq:relent-markov-gap} and writing the relative entropy in terms of entanglement entropies gives
\begin{equation}
    S_R(A:B) - I(A:B) \leq S(\tilde{\rho}_{ABB^*}) - S(\rho_{ABB^*}) = S_{\text{gen}}(\KRT(A^*)) - S_{\text{gen}}(\RT(A^*)).
\end{equation}
Exploiting the $A \leftrightarrow A^*, B \leftrightarrow B^*$ symmetry of the canonical purification, we may rewrite this as
\begin{equation}
    S_R(A:B) - I(A:B) \leq S_{\text{gen}}(\KRT(A)) - S_{\text{gen}}(\RT(A)).
\end{equation}
But, referring back to equation \eqref{eq:SRmI-Sgen} from the introduction, we see that this is a contradiction unless $S_{\text{gen}}(\KRT(B))$ equals $S_{\text{gen}}(\RT(B)).$

We take this to suggest that, except in very special cases, inequality \eqref{eq:interm-fidelity-inequality} really ought to be strict. The gap in that inequality should contribute to a gap in the inequality \eqref{eq:final-fidelity-inequality}; which, as we discussed in the paragraphs following that inequality, will generically be fairly far from equality.

%%%%%%%%%%%
\subsection{Fixed area Markov recovery}
\label{subsec:fixed-area}

We now proceed to computing the fidelity between the canonical purification reduced state $\rho_{ABB^*}$ and the TGT state $\tilde{\rho}_{ABB^*}$ in a toy model where the fidelity can be computed exactly. Our toy model will be made out of fixed area states, introduced in \cite{fixed-area-AR, fixed-area-DHM}, whose gravitational path integral rules were shown in \cite{fixed-area-DHM} to be sufficiently simple that many entropy calculations can be done exactly with fairly mild assumptions.

As we emphasized in the introduction, the surfaces $\RT(A^*)$ and $\KRT(A^*)$ are best thought of as two preferred members of the class of codimension-$2$ spacelike surfaces homologous to $A^*$. The surface $\RT(A^*)$ is singled out because it has minimal generalized entropy among the surfaces that are quantum extremal; the surface $\KRT(A^*)$ is singled out because it contains the entanglement wedge cross-section. So we might rephrase the question of computing $F(\rho_{ABB^*}, \tilde{\rho}_{ABB^*})$ as follows: 
\begin{itemize}
    \item Given a boundary region $R$ and a state $\rho_{R}$ whose bulk entanglement wedge is bounded by the bulk surface $\gamma$, and which contains another covariantly defined surface $\tilde{\gamma}$ with greater generalized entropy, define $\tilde{\rho}_{R}$ to be a state whose entanglement wedge ends at $\tilde{\gamma}$. What is the fidelity $F(\rho_R, \tilde{\rho}_{R})$? 
\end{itemize}
If we substitute $R \rightarrow ABB^*$, $\gamma \rightarrow \RT(A^*)$, and $\tilde{\gamma} \rightarrow \KRT(A^*)$, this is exactly the question we posed in the previous subsection. A general state of this kind is sketched in figure \ref{fig:two-QES-state}.

\begin{figure}
    \centering
    \includegraphics{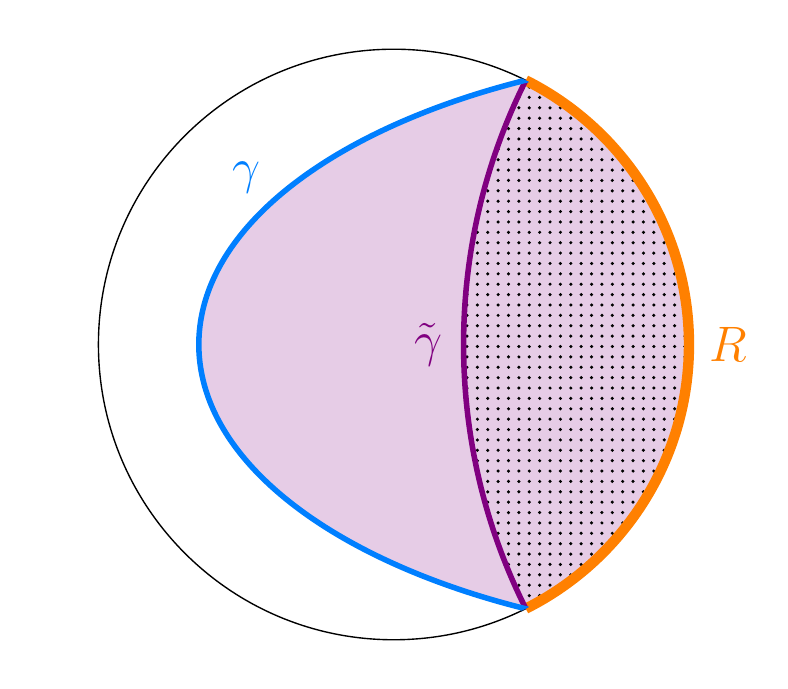}
    \caption{A sketch of a boundary state $\rho_{R}$ whose bulk entanglement wedge contains two covariantly defined surfaces: the minimal quantum extremal surface $\gamma$ and another surface $\tilde{\gamma}$ with greater generalized entropy. The goal of this section is to compute, in a toy model, the quantum fidelity between $\rho_{R}$ and a ``coarser'' state $\tilde{\rho}_{R}$ whose entanglement wedge is only the dotted region.}
    \label{fig:two-QES-state}
\end{figure}

This is an interesting question in general --- presumably the answer tells us something about the boundary degrees of freedom encoding the bulk region between $\gamma$ and $\tilde{\gamma}$. The physical question we are asking is something like, ``how much information do you lose when you trace out that region?'' While we know of no way to compute this quantity for general bulk states, we can compute it exactly in a fixed area state toy model. To apply the fixed-area path integral formalism of \cite{fixed-area-DHM}, we will have to assume that both $\gamma$ and $\tilde{\gamma}$ are classically extremal. This is not exactly the case for the TGT state defined in the previous subsection, even in states with no quantum matter --- while $\tilde{\gamma} = \KRT(A^*)$ is locally extremal, the isolated kinks prevent it from being a truly extremal surface. Verifying that the fixed-area calculation goes through without caveats for extremal surfaces with isolated kinks would be interesting, but would require a careful analysis of boundary conditions at those kinks that is beyond the scope of the present calculation. We will therefore treat $\KRT(A^*)$ as a truly extremal surface. This is not such an extreme assumption, as inequality \eqref{eq:final-fidelity-inequality} is already only supposed to hold in a toy model; the extra assumption that we can treat $\KRT(A^*)$ as extremal should be thought of as an extra assumption on the model.

We now review the bare essentials of fixed area states and their path integrals, originally described in \cite{fixed-area-AR, fixed-area-DHM}; we refer the reader to section 4.2 of \cite{akers-penington} for a very nice, thorough introduction to these techniques. Starting with a static, geometric state in AdS/CFT, a fixed area state is constructed by picking a gauge-invariant surface $\Gamma$\footnote{In fact, it isn't enough to require $\Gamma$ to be gauge invariant; we must also require that it can be sensibly defined in \emph{any} sufficiently smooth spacetime with the right boundary conditions, in order to be able to identify an analogue of $\Gamma$ in spacetimes contributing to the gravitational path integral. The surfaces $\RT(A)$ and $\KRT(A)$ satisfy this property, since they are defined based on certain minimality statements.} and projecting the state onto a narrow band of eigenspaces\footnote{The reason we must project onto a ``narrow band'' of eigenspaces rather than an exact eigenspace is that the eigenvalues of an area operator are continuous.} of the operator that gives the area of $\Gamma$ within a code space of excitations. It was argued in \cite{fixed-area-DHM} that fixed area states can be constructed by taking the gravitational path integral preparing the original state and restricting the integral to be taken over only those metrics for which the area of the designated surface --- again, defined in some gauge-invariant way --- lies within the designated band.\footnote{This discussion is all a bit imprecise, as many of the areas one would like to fix are infinite, and must be somehow regulated. We will not address those subtleties here.} At the level of solutions to the classical equations of motion, fixing the area of an \emph{extremal} surface $\Gamma$ amounts to allowing arbitrary conical defects around that surface; if one approximates a fixed area path integral by summing over saddles, the relevant saddles are spacetimes that are solutions to Einstein's equations everywhere but at the designated surface, where there may be a conical defect.

For our toy model, we will construct a fixed area analogue of $\rho_{R}$ by first taking the state to be classical --- i.e., it contains no quantum matter and therefore boundary entropies are computed by bulk areas, with $\area(\gamma) < \area(\tilde{\gamma})$ --- and then fixing the areas of \emph{both} surfaces $\gamma$ and $\tilde{\gamma}$. We will henceforth abuse notation and call this fixed area state $\rho_{R}$, since from here on out we will work entirely in the toy model. We will define $\tilde{\rho}_{R}$ to be a state where \emph{only} the area of $\tilde{\gamma}$ is fixed. At the level of the saddles contributing to the fixed area path integral we will see shortly that it is sensible to say that $\rho_{R}$ knows about both surfaces $\gamma$ and $\tilde{\gamma}$, while $\tilde{\rho}_{R}$ knows only about the surface $\tilde{\gamma}$; this is the justification for the toy model.

\begin{figure}
    \centering
    \includegraphics{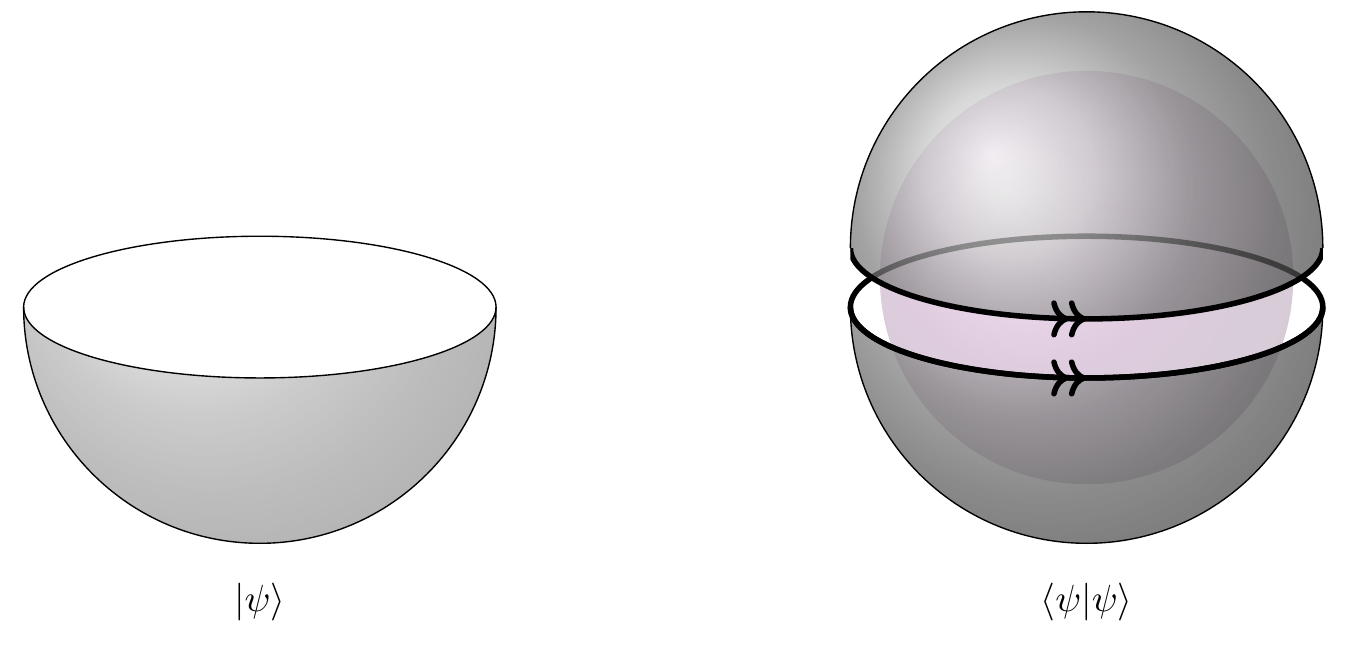}
    \caption{Left-hand side: a Euclidean path integral in the boundary theory preparing a state $\ket{\psi}$. Right-hand side: the Euclidean path integral computing $\braket{\psi}{\psi}$, whose gravitational dual is assumed to be dominated by a single bulk saddle.}
    \label{fig:one-bulk-saddle}
\end{figure}

Now, suppose $\ket{\psi}$ is a state in AdS/CFT that can be prepared by a Euclidean path integral. We will assume that $\ket{\psi}$ is geometric; what this means at a technical level is that the path integral computing $\braket{\psi}{\psi}$ is dominated by a single gravitational saddle point that ``fills in'' the bulk; this is sketched in Figure \ref{fig:one-bulk-saddle}. If $\ket{\psi^{\text{fix}}}$ is a fixed area state constructed from $\ket{\psi}$, and $\sigma^{\text{fix}}_R$ is the corresponding density matrix, then saddles contributing to $\tr((\sigma^{\text{fix}}_R)^n)$ can be constructed by replicating the saddle dominating $\braket{\psi}{\psi}$ $n$ times around the fixed surface. This is sketched in figure \ref{fig:replica-saddles}. The main assumption that goes into the fixed area path integral is that saddles constructed in this way are the only ones that need to be considered to get an accurate answer for quantities like $\tr((\sigma^{\text{fix}}_R)^n).$

\begin{figure}
    \centering
    \includegraphics{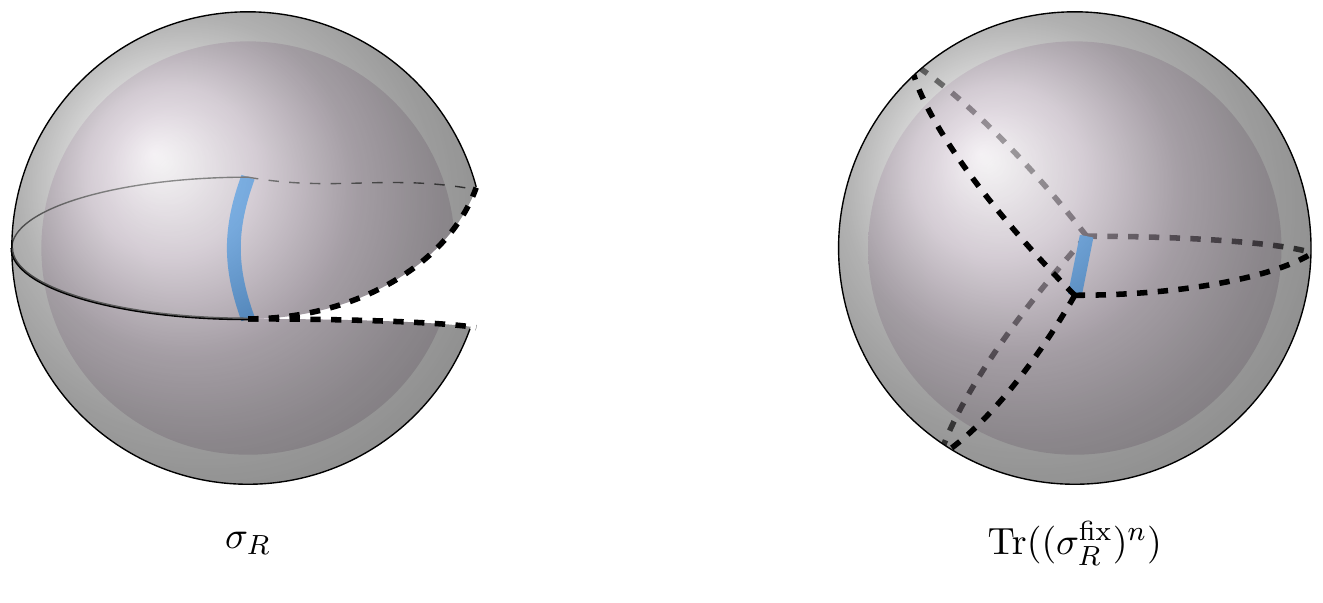}
    \caption{Left-hand side: the path integral preparing $\sigma_R$, whose trace is dominated by a single gravitational saddle containing a surface $\gamma$. Right-hand side: a saddle for $\tr((\sigma_R^{\text{fix}})^3)$, constructed by replicating the single-trace saddle three times around $\gamma$ and thereby giving $\gamma$ a conical defect angle $6 \pi$.}
    \label{fig:replica-saddles}
\end{figure}

We will now describe how to compute the fidelity $F(\rho_{R}, \tilde{\rho}_{R}).$ This is accomplished by computing the quantity
\begin{equation} \label{eq:G-definition}
    G_{m, n}(\rho_{R}, \tilde{\rho}_{R}) = \tr((\tilde{\rho}_{R}^m \rho_{R} \tilde{\rho}_{R}^m)^n)
\end{equation}
for integer values of $m$ and $n$ using the gravitational path integral, analytically continuing to non-integer values of $m$ and $n$, and taking
\begin{equation} \label{eq:fidelity-G}
    F(\rho_R, \tilde{\rho}_R) = G_{\frac{1}{2}, \frac{1}{2}}(\rho_{R}, \tilde{\rho}_{R})^2.
\end{equation}
To compute $G_{m,n}$, we will introduce a useful visual calculus developed in \cite{competing-saddles-2} and \cite{akers-penington}. Every time a factor of $\rho_R$ appears in an expression, we insert a bulk ``pie'' that looks like the one drawn in figure \ref{fig:two-surface-pie}. This is the bulk saddle that dominates the computation of $\tr(\rho_{R})$, with a cut emanating from the fixed-area surface $\gamma$ and ending at $R$. Every time a factor of $\tilde{\rho}_{R}$ appears, we insert a pie like the one shown in figure \ref{fig:one-surface-pie} . To aid in drawing figures, we will ``unwrap'' each pie and draw it as a ``wedge'' --- sketched in figure \ref{fig:wedges}. The path integral for $\tr(\rho_R^n)$, in this language, is sketched in figure \ref{fig:renyi-path-integral}; the segments of cuts neighboring the boundary $R$ are all glued together cyclically, but there is one saddle for each possible gluing of the ``intermediate'' segments of cuts, which are in one-to-one correspondence with permutations in the group $S_n$. Readers who remain confused about this technique are encouraged to consult the more detailed explanation of \cite{akers-penington}.

\begin{figure}
    \centering
    \makebox[\textwidth][c]{
	\subfloat[\label{fig:two-surface-pie}]{
	    \includegraphics{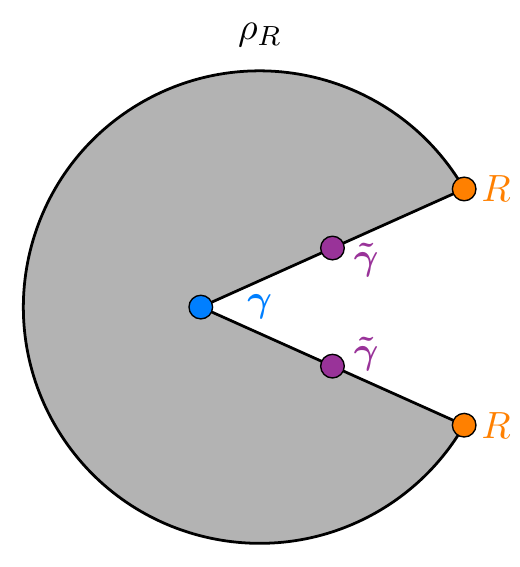}
	}
	\hspace{5em}
	\subfloat[\label{fig:one-surface-pie}]{
    	\includegraphics{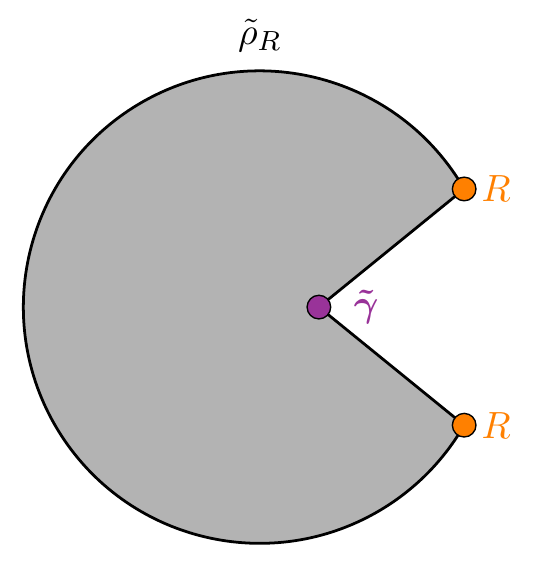}
	}
	}
    \caption{(a) A visual representation of the path integral preparing $\rho_{R}$. We can think of this as a cross-section of the path integral preparing the state shown in figure \ref{fig:two-QES-state}, with the two fixed-area surfaces shown in the dominant bulk saddle. (b) A representation of the path integral preparing $\tilde{\rho}_{R}$, which has only one fixed area surface.}
    \label{fig:pies}
\end{figure}

\begin{figure}
    \centering
    \makebox[\textwidth][c]{
	\subfloat[\label{fig:wedges}]{
	    \includegraphics{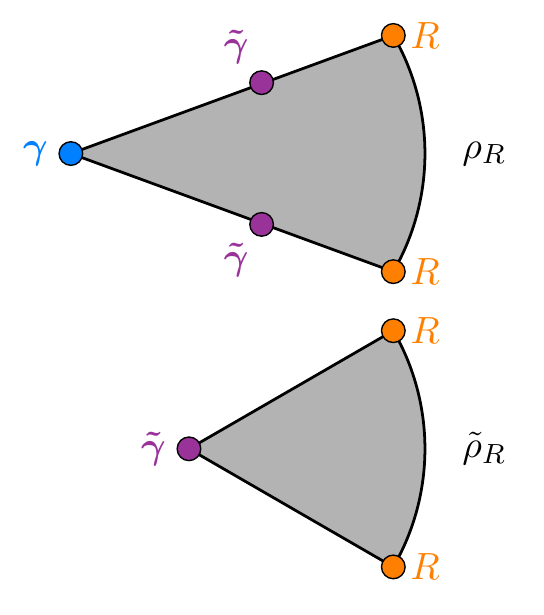}
	}
	\hspace{1em}
	\subfloat[\label{fig:renyi-path-integral}]{
    	\includegraphics{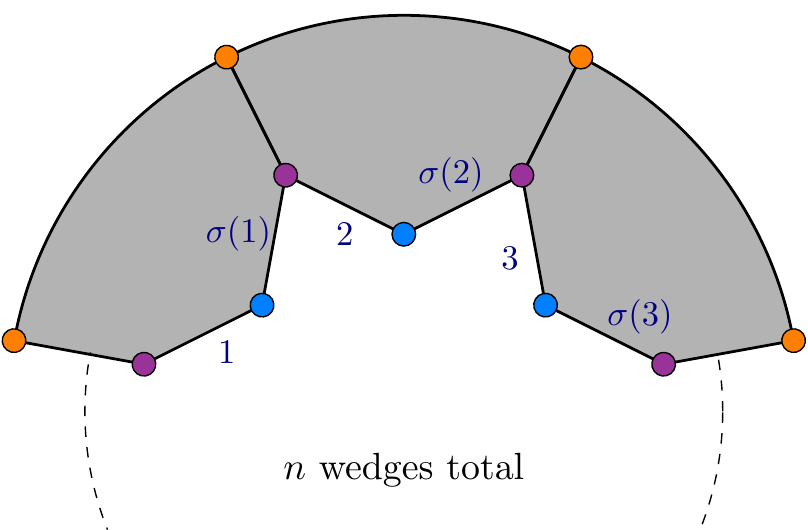}
	}
	}
    \caption{(a) The ``pies'' of figure \ref{fig:pies} unwrapped into ``wedges.'' (b) A saddle contributing to $\tr(\rho_R^n)$. The surfaces lying counterclockwise of each blue dot are labeled $1$ through $n$; for every permutation $\sigma \in S_n$, one can construct a saddle by pasting the surface immediately clockwise of surface $k$ to surface $\sigma(k).$}
\end{figure}

\begin{figure}
    \centering
    \includegraphics{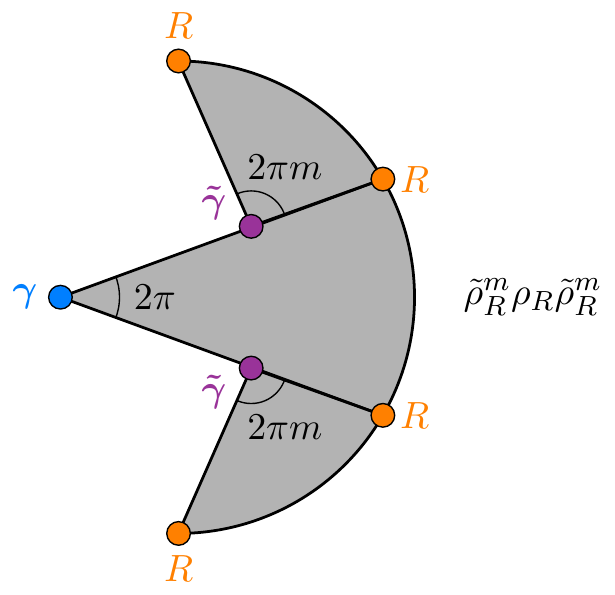}
    \caption{A visual representation of the gravity path integral for the operator $\tilde{\rho}^m_R \rho_R \tilde{\rho}^m_R.$ The function $G_{m,n}$ defined in \eqref{eq:G-definition} is formed by pasting $n$ copies of this path integral together cyclically on the boundary region $R$.}
    \label{fig:integer-operator}
\end{figure}

In this language, the operator $\tilde{\rho}_{R}^m \rho_{R} \tilde{\rho}_R^m$ is prepared by a path integral sketched in figure \ref{fig:integer-operator}. The function $G_{m,n}(\rho_R, \tilde{\rho}_{R})$ is computed by taking $n$ copies of this path integral and gluing the boundaries together cyclically. There is one saddle for each permutation $\sigma \in S_n$, with $\sigma$ determining the gluing pattern for the interior segments of the sides of the wedges. In the saddle corresponding to permutation $\sigma$, there are $C(\sigma)$ copies of surface $\gamma$, where $C(\sigma)$ is the number of cycles in the cycle decomposition of $\sigma$. Before gluing, the total angle subtending $\gamma$-type surfaces is $2 \pi n$, making the total defect
\begin{equation}
    \operatorname{defect}(\gamma) = 2 \pi n - 2 \pi C(\sigma).
\end{equation}
The total number of copies of surface $\tilde{\gamma}$ after gluing is $C(\tau^{-1} \circ \sigma)$, where $\tau$ is the cyclic permutation. The total angle around $\tilde{\gamma}$-type surfaces before gluing is $2 \pi n (2 m + 1)$. This makes the total defect around $\tilde{\gamma}$ surfaces
\begin{equation}
    \operatorname{defect}(\tilde{\gamma}) = 2 \pi n (2 m + 1) - 2 \pi C(\tau^{-1} \circ \sigma).
\end{equation}
It was shown in \cite{fixed-area-DHM} that the total path integral contribution from a saddle with fixed areas $A_j$ and defects $\phi_j$ is
\begin{equation}
    Z = \exp\left[ - \sum_{j} \frac{A_j \phi_j}{8 \pi G_N} \right].
\end{equation}
From this, we get the expression
\begin{equation}
    G_{m, n}(\rho_R, \tilde{\rho}_{R})
        = \sum_{\sigma \in S_n} \exp\left[ \frac{\area(\gamma)}{4 G_N} (C(\sigma) -  n) + \frac{\area(\tilde{\gamma})}{4 G_N} (C(\tau^{-1} \circ \sigma) - 2mn - n)   \right].
\end{equation}

The largest terms in this sum are the ones with $C(\sigma) + C(\tau^{-1} \circ \sigma) = n+1$. Terms with $C(\sigma) + C(\tau^{-1} \circ \sigma) < n + 1$ are subleading by factors like $e^{- \area(\gamma) / 4G_N}$ or $e^{- \area(\tilde{\gamma}) / 4 G_N}$. As emphasized in \cite{akers-penington}, these terms are \emph{infinitely} suppressed when the fixed area surfaces have infinite area (i.e., when region $R$ has a boundary); even when the areas are finite, these terms are suppressed nonperturbatively in $1/G_N$. We will neglect such terms --- as we have already ignored some nonperturbative corrections to $G_{m, n}$ in assuming that the fixed area path integral is dominated by a sum over replicated saddles --- but understanding how such terms contribute to the sum and whether any physics can be extracted from them may be an interesting direction for future work. The permutations $\sigma$ with $C(\sigma) + C(\tau^{-1} \circ \sigma) = n +1$ are called \emph{non-crossing permutations} --- see \cite{mingo2004annular} for a review --- and the number of such permutations with $C(\sigma) = k$ is the Narayana number
\begin{equation}
    N(n, k) = \frac{1}{n} {n \choose k} {n \choose k-1}.
\end{equation}

By summing over only non-crossing permutations, our formula for $G_{m,n}$ simplifies to
\begin{equation}
    G_{m, n}(\rho_R, \tilde{\rho}_{R})
        = \sum_{k=1}^n N(n, k) \exp\left[ \frac{\area(\gamma)}{4 G_N} (k -  n) + \frac{\area(\tilde{\gamma})}{4 G_N} (1 - k - 2mn)  \right].
\end{equation}
This sum can be computed using Mathematica, giving the expression\footnote{We refer the reader to the very nice appendix III of \cite{kudler2021relative} for an explanation of how to do sums like this by hand.}
\begin{equation} \label{eq:G-formula}
    G_{m, n}(\rho_R, \tilde{\rho}_{R})
        = e^{- [2 m n \area(\tilde{\gamma}) - (1-n) \area(\gamma)]/4 G_N} {}_2 F_1 \left(1-n, -n; 2; e^{- [\area(\tilde{\gamma}) - \area(\gamma)]/ 4 G_N} \right).
\end{equation}
The hypergeometric function is analytic in its first two arguments when the third argument is fixed and the fourth argument has magnitude less than one;\footnote{We refer the reader to section 2.1.6 of \cite{bateman1953higher} for a review of analyticity properties of ${}_{2}F_1$.} these conditions are satisfied above because the area of $\tilde{\gamma}$ was assumed to be greater than that of $\gamma$, so we can freely analytically continue to $m=n=1/2$ and obtain
\begin{equation}
    G_{\frac{1}{2}, \frac{1}{2}}(\rho_R, \tilde{\rho}_{R})
        = e^{-[\area(\tilde{\gamma}) - \area(\gamma)]/8 G_N} {}_2 F_1 \left(\frac{1}{2}, -\frac{1}{2}; 2; e^{- [\area(\tilde{\gamma}) - \area(\gamma)]/ 4 G_N} \right).
\end{equation}
Plugging this back into equation \eqref{eq:fidelity-G}, we find the fidelity is
\begin{equation} \label{eq:final-fidelity-computation}
    F(\rho_{R}, \tilde{\rho}_{R})
        = e^{-[\area(\tilde{\gamma}) - \area(\gamma)]/4 G_N} \left[{}_2 F_1 \left(\frac{1}{2}, -\frac{1}{2}; 2; e^{- [\area(\tilde{\gamma}) - \area(\gamma)]/ 4 G_N} \right) \right]^2.
\end{equation}
We stress again that this equation holds only up to corrections nonperturbative in $1/G_N$. With the substitutions $\gamma \rightarrow \RT(A^*), \tilde{\gamma} \rightarrow \KRT(A^*),$ $R \rightarrow ABB^*$, together with the $A \leftrightarrow A^*, B \leftrightarrow B^*$ symmetry of the canonical purification, equation \eqref{eq:final-fidelity-computation} and inequality \eqref{eq:interm-fidelity-inequality} together reproduce the claimed inequality \eqref{eq:FA-fidelity-bound}.

To conclude, we note that the methods of this subsection can easily be generalized to compute all the sandwiched R\'{e}nyi relative entropies of $\rho_{R}$ and $\tilde{\rho}_{R}$. The sandwiched R\'{e}nyi relative entropies were defined in \cite{muller-lennert_quantum_2013, wilde_strong_2014} by
\begin{equation}
    \tilde{D}_{\alpha}(\rho || \sigma) = \frac{1}{\alpha - 1} \log \tr[(\sigma^\gamma \rho \sigma^\gamma)^\alpha]
\end{equation}
with $\gamma = (1-\alpha)/2 \alpha.$ In terms of our function $G_{m, n}$ defined in \eqref{eq:G-definition}, we have
\begin{equation}
    \tilde{D}_{\alpha}(\rho_R || \tilde{\rho}_{R})
        = \frac{1}{\alpha - 1} \log\left[ G_{(1 - \alpha)/2 \alpha, \alpha}(\rho_R, \tilde{\rho}_{R})  \right].
\end{equation}
Plugging in the formula for $G_{m,n}$ from equation \eqref{eq:G-definition} gives
\begin{equation} \label{eq:sandwiched-Renyis}
    \tilde{D}_{\alpha}(\rho_R || \tilde{\rho}_{R})
        = \frac{\area(\tilde{\gamma}) - \area(\gamma)}{4 G_N} + \frac{1}{\alpha - 1} \log\left[ {}_2 F_1 \left(1-\alpha, -\alpha; 2; e^{- [\area(\tilde{\gamma}) - \area(\gamma)]/ 4 G_N} \right)  \right].
\end{equation}

The sandwiched R\'{e}nyi relative entropies for $\alpha \neq 1/2$ are not symmetric in their arguments, but we can compute $\tilde{D}_{\alpha}(\tilde{\rho}_R || \rho_{R})$ quite easily. A similar calculation to the one given for $G_{m, n}$ can be performed for the function
\begin{equation} \label{eq:H-definition}
    H_{m, n}(\rho_{R}, \tilde{\rho}_{R}) = \tr((\rho_{R}^m \tilde{\rho}_{R} \rho_{R}^m)^n),
\end{equation}
yielding
\begin{equation}
    H_{m, n}(\rho_{R}, \tilde{\rho}_{R}) = \sum_{\sigma \in S_{2 m n}}
        \exp\left[\frac{\area(\gamma)}{4 G_N} (C(\sigma) - 2 m n) + \frac{\area(\tilde{\gamma})}{4 G_N} (C(\tau^{-1} \circ \sigma) - n (2m+1) \right].
\end{equation}
Summing over non-crossing permutations in $S_{2 m n}$ gives
\begin{align}
    H_{m, n} (\rho_R, \tilde{\rho}_R)
        & = e^{[-(2 m n - 1) \area(\gamma)- n \area(\tilde{\gamma})]/4 G_N} \times \nonumber \\
            & \qquad 
            {}_{2} F_1\left( - 2 m n, 1 - 2 m n, 2, e^{- [\area(\tilde{\gamma}) - \area(\gamma)]/4 G_N} \right).
\end{align}
Setting $n = \alpha, m = (1 - \alpha) / 2 \alpha$ yields
\begin{align}
    H_{(1-\alpha)/2 \alpha, \alpha} (\rho_R, \tilde{\rho}_R)
        & = e^{- \alpha [\area(\tilde{\gamma}) - \area(\gamma)] / 4 G_N } \times \nonumber \\
        & \qquad
            {}_{2} F_1\left( - 1 + \alpha, \alpha, 2, e^{- [\area(\tilde{\gamma}) - \area(\gamma)]/4 G_N} \right).
\end{align}
From this, we can compute the sandwiched R\'{e}nyi relative entropy as
\begin{equation} \label{eq:reversed-sandwiched-Renyis}
    \tilde{D}_{\alpha}(\tilde{\rho}_{R} || \rho_R)
        = \frac{1}{\alpha - 1} \log\left[H_{(1-\alpha)/2 \alpha, \alpha}(\rho_R, \tilde{\rho}_{R})  \right].
\end{equation}

Note, finally, that in the limit $\alpha \rightarrow 1$, the sandwiched R\'{e}nyi relative entropy approaches the ordinary relative entropy. One can check that $\tilde{D}_{\alpha}(\rho_R || \tilde{\rho}_{R})$ has a finite $\alpha \rightarrow 1$ limit, while $\tilde{D}_{\alpha}(\tilde{\rho}_R || \rho_{R})$ diverges; we take this to mean that the support of $\tilde{\rho}_{R}$ is strictly larger than the support of $\rho_R$. This makes sense: very heuristically speaking, we might think of $\tilde{\rho}_{R}$ as a state constructed from $\rho_R$ by taking the portion of the $R$-Hilbert space encoding the bulk region lying between $\gamma$ and $\tilde{\gamma}$ and erasing the information therein by replacing it with the maximally mixed state.

A slightly different relative entropy between fixed area states was computed in \cite{kudler2021relative}; this was a relative entropy between two states where both areas were fixed but where the bulk states contained orthogonal excitations in the region bounded by $\gamma$ and $\tilde{\gamma}.$

%%%%%%%%%%%%%%%%%%%%%%%
\section{Discussion}
\label{sec:discussion}

Our goal in this paper was threefold: (i) we established a connection between the quantity $S_R(A:B) - I(A:B)$ and Markov recovery processes in canonical purifications (section \ref{sec:info-theory}), arguing in the process that boundaries in the entanglement wedge cross-section of a holographic state require $S_R - I$ to be nonzero at order $1/G_N$; (ii) we proved inequality \eqref{eq:big-technical-claim} for time-symmetric states in AdS$_3$ gravity (section \ref{sec:geometric-proof}), establishing a quantitative bound on the Markov gap in such states, and (iii) we explored a fixed area toy model of the Markov recovery process (section \ref{sec:recovery-models}) to see how information inequalities partially reproduce the geometric bound on the Markov gap given in equation \eqref{eq:SRmI-Sgen}. Along the way, we explored generalizations of inequality \eqref{eq:big-technical-claim} to states with bulk matter and states in higher-dimensional theories of gravity (section \ref{sec:generalizations}).

We now comment on some general lessons from our analysis, and some possible directions for future work.

%%%%%%%%%%%
\subsection{Generalizing the bound}
\label{subsec:future-generalizations}

A natural first step in generalizing the proof of section \ref{sec:geometric-proof} would be to prove inequality \eqref{eq:big-technical-claim} for non-time-symmetric classical states in pure AdS$_3$ gravity. The chief difficulty in pursuing this generalization is that the minimal surfaces under consideration are no longer constrained to lie in a single time slice; one must work in the full, Lorentzian geometry.

The proof of section \ref{sec:geometric-proof} worked by modeling time-symmetric slices of asymptotically AdS$_3$ states as quotients of the hyperbolic plane $\mathbb{H}_{2}$. One could imagine that a similar analysis could be undertaken without the assumption of time symmetry, by modeling an arbitrary geometric state in pure AdS$_3$ gravity as a quotient of global AdS$_3$. While this problem certainly seems tractable, we think it unlikely that the technique of passing to the universal cover will be particularly useful in the long run. After all, one would like to generalize inequality \eqref{eq:big-technical-claim} not only to non-time-symmetric states, but also to states with bulk matter; these states are generically not universally covered by AdS$_3$.

A technique we find more promising, in the sense that it is more easily generalizable, is to reinterpret the proof of inequality \eqref{eq:big-technical-claim} in terms of Stokes' theorem. We emphasized in the introduction that the surfaces $\KRT(A)$ and $\RT(A)$ are homologous to one another. By expressing the areas of $\KRT(A)$ and $\RT(A)$ as integrals over the corresponding surfaces, one can use Stokes' theorem to write the area difference $[\area(\KRT(A)) - \area(\RT(A))]$ as an integral over the homology region between the two. There are many ways to write this integral --- any smooth form on the homology region that interpolates between the volume forms on $\KRT(A)$ and $\RT(A)$ will give rise to a homology integral computing the area difference --- but no matter how we write the integral, it must have the property that it diverges at the kinks of $\KRT(A)$. We think it likely that a careful analysis of this problem will yield an alternative proof of \eqref{eq:big-technical-claim}, though we have not yet succeeded in tackling the problem from that direction. An advantage of the Stokes' theorem approach is that it is naturally generalizable: it can be applied to non-time-symmetric states, as well as to states in arbitrary dimensions, and even to states with classical bulk matter provided that one understands how suitable energy conditions on that matter affect the integral on the homology region.

Generalizing inequality \eqref{eq:big-technical-claim} to states with bulk quantum matter is trickier. As we mentioned in section \ref{subsec:quantum-matter}, the bulk contribution to the Markov gap $S_{R, \text{bulk}} - I_{\text{bulk}}$ is not necessarily nonnegative on its own; even though the Markov gap is nonnegative for any quantum state, the bulk entropy contribution to the mutual information is not actually the mutual information of a bulk state, so one cannot generally guarantee the inequality $S_{R, \text{bulk}} \geq I_{\text{bulk}}$ --- at least not using the techniques of this paper. So proving inequality \eqref{eq:big-technical-claim} in generality for the area contribution to the Markov gap would not suffice to prove the inequality in the presence of quantum matter. This state of affairs is further complicated by the fact that we don't really expect the area contribution of the Markov gap to universally satisfy inequality \eqref{eq:big-technical-claim} on its own; it likely satisfies that inequality under suitable energy conditions imposed on classical matter configurations, but such energy conditions are generally violated by quantum matter. If an inequality like \eqref{eq:big-technical-claim} is to hold for arbitrary states, we expect instead that quantum matter can cause the classical piece of the Markov gap to dip below the lower bound, but that the Markov gap of the quantum matter will make up the difference. One can think of this as a rather delicate energy condition on quantum fields, where quantum field theory states violating some classical energy conditions must have a sufficiently large $S_{R, \text{bulk}} - I_{\text{bulk}}$ to make up for the violation. Any research effort in this direction would likely start with the results of \cite{SRQFT1, SRQFT2, camargo2021long}, where the authors compared reflected entropy and mutual information in simple quantum field theories.

%%%%%%%%%%%
\subsection{Multipartite holographic entanglement}

In the introduction, we mentioned the result of Akers and Rath \cite{akers-rath-tripartite} that any pure three-party state with a mostly-bipartite entanglement structure has Markov gap close to zero. Our inequality \eqref{eq:big-technical-claim} can be viewed as a quantitative version of their observation that certain holographic states require tripartite entanglement. Any holographic state $\ket{\Psi_{ABC}}$ for which the entanglement wedge cross-section of $\rho_{AB}$ has a boundary must have tripartite entanglement at order $1/G_N$; the scaling of inequality \eqref{eq:big-technical-claim} seems to suggest --- though not prove --- that each boundary of the entanglement wedge cross-section requires some minimal amount of tripartite boundary entanglement to support it.

Let us consider this idea in light of the lore, originating in \cite{van2010building}, that geometric connections in bulk states emerge from entanglement patterns in boundary states. What we would like to say is something like, ``connections across codimension-$2$ surfaces are supported by bipartite entanglement in the boundary; connections across codimension-$3$ surfaces (like boundaries of entanglement wedge cross-sections), are supported by tripartite entanglement.''

At present, we do not have any concrete proposals for how to test this idea or make it more precise. We think it worthwhile to mention, however, a connection to another issue in holography: the difficulties involved in understanding intersections of extremal surfaces. The area operators of intersecting extremal surfaces do not commute. As emphasized in section 7 of \cite{bao2019beyond}, this means it is impossible to model the entanglement pattern of intersecting extremal surfaces with a tensor network made up of isometries or with fixed area states. It is generically the case that if one considers $n$ overlapping boundary subregions of a holographic state, the quantum extremal surfaces computing their entropies will have some complicated pattern of intersections in the bulk. The only tensor network model of holography we know of in which multipartite entanglement patterns have been computed explicitly is that of random stabilizer tensor networks; in that context, it was shown in \cite{nezami2020multipartite} that tripartite entanglement is scarce. Because tensor network models \emph{also} cannot model intersecting extremal surfaces, we might expect that the geometry of extremal surface intersections emerges from multipartite entanglement in the boundary state. This line of reasoning, if taken seriously, would provide an independent suggestion that codimension-$3$ and higher features of bulk geometries are signs of multipartite entanglement in the boundary theory.

Unfortunately, quantum extremal surface intersections appear to be a more complicated beast. After this manuscript first appeared online, we learned of forthcoming work \cite{akers2021reflected} that finds a parametrically large Markov gap in \emph{non-stabilizer} random tensor networks, indicating that these tensor networks contain multipartite entanglement. Despite this feature, these tensor networks are not capable of modeling the geometry of intersecting extremal surfaces. We still suspect that multipartite entanglement will be an important part of the story in understanding extremal surface intersections --- but the mere existence of multipartite entanglement is not enough to make such intersections emerge.

\acknowledgments{It is a pleasure to thank Chris Akers, Ning Bao, Geoff Penington, Xiaoliang Qi, and Pratik Rath for illuminating conversations. We especially thank Mark Wilde for explaining equation \eqref{eq:first-fidelity-inequality} during the early stages of this work, and Arvin Shahbazi-Moghaddam for suggesting that bulk strong subadditivity could be used to establish the results of appendix \ref{app:KRT}. This work was partially supported by AFOSR award FA9550-19-1-0369, CIFAR, DOE award DE-SC0019380 and the Simons Foundation.}

\appendix

\section{KRT surfaces and the quantum maximin formula}
\label{app:KRT}

In the introduction, in inequality \eqref{eq:SRmI-Sgen}, we lower-bounded the Markov gap in terms of the generalized entropy differences of KRT and RT surfaces. While the Markov gap is known to be nonnegative from quantum information arguments (see section \ref{sec:info-theory}), it is interesting to see how this nonnegativity is enforced in the bulk theory. If we could establish the inequality
\begin{equation} \label{eq:putative-KRT-RT-entropy}
	S_{\text{gen}}(\KRT(A)|A) \geq S_{\text{gen}}(\RT(A)|A),
\end{equation}
and the equivalent inequality for subregion $B$, then this would immediately imply nonnegativity of the Markov gap via inequality \eqref{eq:SRmI-Sgen}. Inequality \eqref{eq:putative-KRT-RT-entropy} certainly seems reasonable, because $\RT(A)$ has minimal generalized entropy among quantum extremal surfaces in its homology class. As emphasized in the introduction, however, $\KRT(A)$ is not a quantum extremal surface, and so some care must be taken in establishing \eqref{eq:putative-KRT-RT-entropy}.

We will establish that the generalized entropy of $\KRT(A)$ is non-increasing for arbitrary null perturbations towards the boundary region $A$, and then appeal to the quantum focusing conjecture \cite{bousso2016quantum} and quantum maximin \cite{quantum-maximin} to establish inequality \eqref{eq:putative-KRT-RT-entropy}. For states with no quantum matter, the result can be established via the ordinary focusing theorem and the classical maximin formula \cite{maximin}. The key point will be that the bulk state lying between $\KRT(A)$ and $A$ can be extended into two different entanglement wedges --- the $AB$ entanglement wedge in the original state $\rho_{AB}$ and the $AA^*$ entanglement wedge in the canonical purification --- that are each bounded by quantum extremal surfaces.

\begin{figure}
	\centering
	\includegraphics{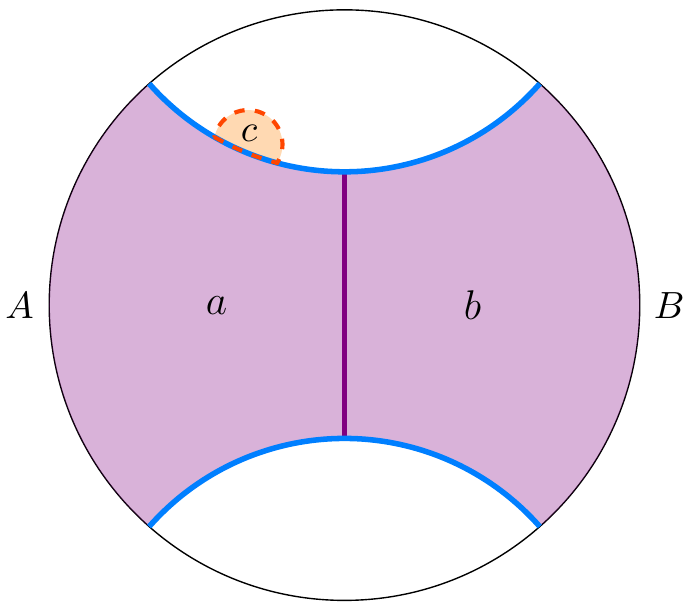}
    \caption{The bulk regions $a$ and $b$ for two intervals in the AdS$_3$ vacuum, together with a small deformation $c$ localized on the surface $\RT(AB)$. While this figure is drawn as though it lies in a single time-slice of the bulk, the deformation should actually lie in a null direction orthogonal to $\RT(AB)$.}
	\label{fig:null-deformation}
\end{figure}

Let us denote by $a$ the wedge lying spacelike between $\KRT(A)$ and $A$, and by $b$ the wedge lying spacelike between $\KRT(B)$ and $b$. This is sketched for a simple spacetime in figure \ref{fig:null-deformation}. Now, consider a null deformation of the portion of $\KRT(A)$ that was inherited from the quantum extremal surface $\RT(AB).$ For reasons that will become clear momentarily, we will choose this deformation to point \emph{away} from $A$. This amounts to adding a small bulk region $c$ to the bulk region $a$; this deformation is also sketched in figure \ref{fig:null-deformation}. The key observation is that this deformation can also be thought of as a deformation of the bulk region $ab$, whose boundary $\RT(AB)$ is quantum extremal. The change in the generalized entropy of $\KRT(A)$ under the addition of region $c$ is given by
\begin{equation}
	\Delta S_{\text{gen}}(\KRT(A)|A)
		= S(a c) - S(a) + \frac{1}{4 G_N} \Delta \area(\KRT(A)).
\end{equation}
But because the perturbation is localized to the portion of $\KRT(A)$ contained within $\RT(AB),$ we have $\Delta \area(\KRT(A)) = \Delta \area(\RT(AB)).$ Furthermore, applying the strong subadditivity inequality to the system $abc$ gives
\begin{equation}
	S(ac) - S(a) \geq S(abc) - S(ab).
\end{equation}
Combining these observations yields the inequality
\begin{equation}
	\Delta S_{\text{gen}}(\KRT(A)|A)
		\geq S(a b c) - S(a b) + \frac{1}{4 G_N} \Delta \area(\RT(AB))
		= \Delta S_{\text{gen}}(\RT(AB)|AB).
\end{equation}
But the right-hand side of this inequality vanishes for infinitesimal regions $c$ due to the quantum extremality of $\RT(AB)$. It follows that the change in $S_{\text{gen}}(\KRT(A)|A)$ due to a null deformation away from $A$ is infinitesimally nonnegative, provided that this null deformation lies in $\RT(AB)$. Linearity of perturbations to the entanglement entropy establishes that the change in $S_{\text{gen}}(\KRT(A)|A)$ is infinitesimally \emph{non-positive} for null perturbations on $\RT(AB)$ pointed \emph{toward} $A$, as claimed in the preceding paragraph.

For null perturbations lying on the portion of $\KRT(A)$ contained with $\RT(AA^*),$ the same exact result can be established using an embedding of the state into the entanglement wedge $aa^*$ in the canonical purification. Putting these results together establishes that the generalized entropy of $\KRT(A)$ is infinitesimally non-increasing for any null deformation toward $A$ that is supported on the smooth parts of $\KRT(A)$. As for null deformations toward $A$ that are supported on the corners of $\KRT(A)$, the right-angle structure of those corners guarantees that the classical expansion of any such deformation is $-\infty.$ The infinitesimal change in the bulk entropy, since this is a renormalized quantity that should not be sensitive to corner divergences, is finite. As such, the total quantum expansion of such a perturbation is negative, establishing that the generalized entropy of $\KRT(A)$ is infinitesimally non-increasing for \emph{arbitrary} null deformations toward $A$. The quantum focusing conjecture \cite{bousso2016quantum} implies that if the change in generalized entropy along a null congruence is initially locally nonpositive, then it is nonpositive everywhere on the congruence,\footnote{The quantum focusing conjecture is the statement that the quantum expansion, which is the logarithmic derivative of the generalized entropy under a local null perturbation, is non-increasing on null congruences.} provided that one takes the usual care to remove generators from the congruence after they reach caustic points. This implies that as $\KRT(A)$ is focused along a null congruence leaving $\KRT(A)$ in the direction of $A$, its generalized entropy is nonperturbatively non-increasing.

We now appeal to the quantum maximin formula \cite{quantum-maximin}, which claims the existence of a complete achronal slice of the bulk spacetime on which $\RT(A)$ has minimal generalized entropy. The intersection of the null congruence leaving $\KRT(A)$ with this slice is called the \emph{representative} $\widehat{\KRT}(A).$ We then have the chain of inequalities
\begin{equation}
	S_{\text{gen}}(\RT(A)|A) \leq S_{\text{gen}}(\widehat{\KRT}(A)|A) \leq S_{\text{gen}}(\KRT(A)|A),
\end{equation}
establishing inequality \eqref{eq:putative-KRT-RT-entropy}.

\bibliographystyle{JHEP}
\bibliography{biblio.bib}

\end{document}